%% file: arxiv.tex
\documentclass[11pt,a4paper]{article}
\usepackage[a4paper,top=2cm,bottom=2cm,left=2cm,right=2cm,marginparwidth=1.75cm]{geometry}
\usepackage{natbib}

\usepackage{graphicx}
\usepackage{xcolor}
\usepackage{algorithm}
\usepackage{algpseudocode}
\usepackage{mathtools}
\usepackage{amsmath}
\usepackage{amsthm}
\usepackage{amssymb}
\usepackage{thm-restate} 
\usepackage{bbm} 
\usepackage{booktabs} 
\usepackage[section]{placeins} 
\usepackage[section]{placeins} 
\usepackage{subcaption}
\usepackage[title]{appendix}
\usepackage{enumitem}
\usepackage[breaklinks=true,hyperfootnotes=false,colorlinks=true,allcolors=blue]{hyperref}
\usepackage{cleveref}
\Crefformat{appendix}{\S#2#1#3 (of the online supplementary material)}
\Crefformat{equation}{(#2#1#3)}
\Crefrangeformat{equation}{(#3#1#4) to~(#5#2#6)}

\theoremstyle{plain}
\newtheorem{theorem}{Theorem}
\newtheorem{proposition}[theorem]{Proposition}
\newtheorem{lemma}[theorem]{Lemma} 
\newtheorem{corollary}[theorem]{Corollary}

\theoremstyle{definition}

\newtheorem{remark}{Remark}

\theoremstyle{definition}
\newtheorem{definition}{Definition}

\providecommand{\keywords}[1]{{\textbf{Keywords:} \hspace{0.8mm}} #1}

\input{commands_cca}

\begin{document}

\title{Regularised Canonical Correlation Analysis: graphical lasso, biplots and beyond}
\author{Lennie Wells$\hspace{0.3mm}^{1}$,  Kumar Thurimella$\hspace{0.3mm}^{2}$ \& Sergio Bacallado$\hspace{0.3mm}^{1}$}
\date{}

\maketitle

\begin{abstract}
Recent developments in regularized Canonical Correlation Analysis (CCA) promise powerful methods for high-dimensional, multiview data analysis. However, justifying the structural assumptions behind many popular approaches remains a challenge, and features of realistic biological datasets pose practical difficulties that are seldom discussed. We propose a novel CCA estimator rooted in an assumption of conditional independencies and based on the Graphical Lasso. Our method has desirable theoretical guarantees and good empirical performance, demonstrated through extensive simulations and real-world biological datasets. Recognizing the difficulties of model selection in high dimensions and other practical challenges of applying CCA in real-world settings, we introduce a novel framework for evaluating and interpreting regularized CCA models in the context of Exploratory Data Analysis (EDA), which we hope will empower researchers and pave the way for wider adoption.
\\

\noindent \keywords{Canonical Correlation Analysis; Graphical Lasso; Multi-Omics; Multi-view data. \vspace{-5mm}}
\end{abstract}
{\let\thefootnote\relax\footnote{{$^1$Department of Pure Mathematics and Mathematical Statistics, University of Cambridge}}}
{\let\thefootnote\relax\footnote{{$^2$Department of Chemical Engineering and Biotechnology, University of Cambridge}}}

\input{maintext}

\section{Author contributions statement}
LW is largely responsible for the theoretical results about the gCCA procedure and the mathematical presentation in the supplementary materials. He conceived of the methodology described in \Cref{sec:comparison-interpretation}. He did the bulk of the writing, wrote all the code for the project, designed and carried out all numerical experiments. SB conceived the core idea of applying the graphical lasso to CCA, supervised the project closely, and edited the article. 
KT obtained the Microbiome dataset analysed in \Cref{sec:real-data} and wrote the corresponding commentary in \Cref{sec:biological-comment}.

\section{Acknowledgments}
\ouracknowledgements
\ourfundinginfo

\bibliographystyle{abbrvnat}
\bibliography{SCCA}

\begin{appendices}
\input{suppl}
\end{appendices}

\end{document}

%% file: commands_cca.tex
\newcommand{\eqcase}[1]{#1} 

\newcommand{\R}{\mathbb{R}}
\newcommand{\defeq}{\vcentcolon=}
\newcommand{\eqdef}{=\vcentcolon}
\newcommand{\ind}{\mathbbm{1}}
\newcommand{\genX}{\mathrm{X}}
\newcommand{\genY}{\mathrm{Y}}
\newcommand{\X}{\mathbf{X}}
\newcommand{\Y}{\mathbf{Y}}

\newcommand{\z}{\mathbf{z}} 

\newcommand{\E}{\mathbb{E}}
\newcommand{\ve}{\mathbf{e}}
\newcommand{\vf}{\mathbf{f}}
\newcommand{\C}{\mathbf{C}}
\newcommand{\Cxx}{\mathbf{C}_{xx}}
\newcommand{\Cxy}{\mathbf{C}_{xy}}
\newcommand{\Cyx}{\mathbf{C}_{yx}}
\newcommand{\Cyy}{\mathbf{C}_{yy}}

\newcommand{\Sigxx}{\Sigma_{xx}}
\newcommand{\Sigxy}{\Sigma_{xy}}
\newcommand{\Sigyx}{\Sigma_{yx}}
\newcommand{\Sigyy}{\Sigma_{yy}}
\usepackage{xfrac}
\newcommand{\mhalf}{{\shortminus \sfrac{1}{2}}}
\newcommand{\half}{{\sfrac{1}{2}}} 
\newcommand{\uth}{^\text{th}}
\newcommand{\Cov}{\operatorname{Cov}}
\newcommand{\Var}{\operatorname{Var}}
\newcommand{\Corr}{\operatorname{Corr}}
\newcommand{\empCov}{\widehat{\Cov}}
\newcommand{\empVar}{\widehat{\Var}}
\newcommand{\empCorr}{\widehat{\Corr}}
\newcommand{\dCCA}{\operatorname{CCA}} 

\newcommand{\CCA}{\operatorname{CCA}}

\newcommand{\SVD}{\operatorname{SVD}}
\newcommand{\CCAcorr}[1]{\mathcal{\rho}^\text{CCA}_{{#1}}}
\newcommand{\CCAcorrto}[1]{\mathcal{R}^\text{CCA}_{[{#1}]}}

\newcommand{\empCCAcorrto}[1]{\hat{\mathcal{R}}^\text{CCA}_{[{#1}]}}

\newcommand{\uk}{u_k} 
\newcommand{\vk}{v_k} 
\newcommand{\ninv}{\tfrac{1}{N}}

\newcommand\uprule{\rule{0mm}{2.1ex}}
\newcommand{\diag}{\operatorname{diag}}
\newcommand{\tr}{\operatorname{trace}}
\newcommand{\spann}{\operatorname{span}}
\newcommand{\cspan}[1]{{\operatorname{span}\{{#1}\}}}
\newcommand{\range}[1]{\operatorname{range}({#1})}
\DeclarePairedDelimiter{\norm}{\lVert}{\rVert}
\DeclarePairedDelimiter{\abs}{\lvert}{\rvert}%
\newcommand{\ones}[1]{\mathrm{1}_{{#1}}} 
\newcommand{\littlel}[1]{\ell_{{#1}}}

\DeclareMathSymbol{\shortminus}{\mathbin}{AMSa}{"39}
\DeclareMathOperator*{\argmax}{arg\,max}
\DeclareMathOperator*{\argmin}{arg\,min}

\newcommand{\ridge}{rCCA}
\newcommand{\wit}{sPLS}
\newcommand{\glasso}{gCCA}
\newcommand{\suo}{sCCA}
\newcommand{\ccazoo}{\texttt{cca-zoo}}
\newcommand{\githubrepo}{\url{https://github.com/W-L-W/RegularisedCCA}}

\newcommand{\ourfundinginfo}{LW is supported by the UK Engineering and Physical Sciences Research Council (EPSRC) under grant number EP/V52024X/1. KT is supported by a Gates Cambridge Scholarship.}
\newcommand{\ouracknowledgements}{The authors would like to 
thank James Chapman for valuable discussions and comments on an early draft of this paper.}

%% file: maintext.tex
\section{Introduction}\label{sec:intro}

Canonical correlation analysis (CCA) is a powerful tool in multivariate statistics and multiview learning.
It provides a principled approach to understand the correlation structure between two datasets or views of data.
CCA is closely related to the better known technique of PCA \citep{jolliffe_principal_1986,jolliffe_principal_2016}, and has similar linguistic ambiguities: the term CCA can refer to a population canonical decomposition (a mathematical object determined by the law of a pair of random vectors), to a related sample CCA object (determined by a pair of data matrices), or to associated frameworks for data analysis.
Moreover, there are many different interpretations of CCA; leading to a great number of extensions and alternative methods \citep{roweis_unifying_1999,bach_probabilistic_2005,klami_bayesian_2013,tenenhaus_variable_2014,chapman_cca-zoo_2021, chapman_efficient_2023}. 
A more complete review of this wider context is given in \Cref{app:cca-broader-background}.

This article takes a frequentist perspective of high-dimensional CCA, where we assume independent and identically distributed observations from some unknown distribution and wish to estimate certain aspects of the population canonical decomposition.
Unfortunately, the classical notion of sample CCA breaks down when there are more dimensions than samples, and in response, many regularised extensions of CCA have been proposed.
Some of these extensions are now widely used, particularly in the fields of genetics \citep{witten_penalized_2009} and neuroscience \citep{mihalik_canonical_2022}. 

Even in this narrower context, CCA can be used to accomplish many different goals.
One very popular goal is dimension reduction of high-dimensional data: to obtain low dimensional representations of the samples that preserve salient information and can be used for downstream tasks, such as classification, regression, and visualisation.
Though our results can be applied to this use-case of dimension reduction, our statistical approach to CCA places more emphasis on the variables than the samples. 
Indeed, our main goal is to understand the relationship between variables in the population distribution.
In this context, CCA can be a powerful tool for Exploratory Data Analysis (EDA) and provide interpretable conclusions to guide scientific enquiry.

Before we can outline our contributions in more detail we need to define CCA and related terminology. Please note that our choice of language and notation may differ from any given source, due to a lack of consistency in the literature where terms such as \emph{loading} can refer to various objects.

\subsection{Terminology for CCA}\label{sec:cca-terminology}
To formalise the \textit{population} CCA problem,
suppose we have two random variables, $X,Y$ taking values in $\R^p,\R^q$ respectively with joint covariance partitioned as
\begin{equation}\label{eq:sig-partition-def}
	\Sigma = \left(\begin{array}{cc}
		\Sigxx & \Sigxy \\
		\Sigyx & \Sigyy
	\end{array}\right)\; .
\end{equation}
CCA defines successive vectors $\uk \in \R^p, \vk \in \R^q$ by the programs
\begin{equation}\label{eq:def-via-Sig}
	\begin{split}
		\underset{u \in \mathbb{R}^{p}, v \in \mathbb{R}^{q}}{\operatorname{maximize}}\,& u^\top \Sigma_{x y} v \\
		\text{subject to  }& u^\top \Sigxx u \leq 1,\: v^\top \Sigyy v \leq 1,\\ 
		&u^\top \Sigxx u_j= v^\top \Sigyy v_j =0 \text{ for } 1 \leq j \leq k-1 \;.	
	\end{split}
\end{equation}
We call the optimal value $\rho_k$ the $k^\text{th}$ \textit{canonical correlation}, call $u_k,v_k$ the $k^\text{th}$ pair of \textit{canonical directions}, or simply \textit{weights} 
and call the projections $u_k^\top X, v_k^\top Y$ the first pair of \textit{canonical variates}. We shall always refer to the components of the original $X,Y$ as \textit{variables}, and to the transformed variables as (canonical) \textit{variates}. 
The quantities $\Sigxx u_k,\Sigyy v_k$ will also be important; we will refer to these as \textit{canonical loading vectors}, or simply \textit{loadings}.
In particular the $k^\text{th}$ pair maximises $\Corr(u_k^\top X, v_k^\top Y)$ subject to the orthogonality constraints.
When $\Sigxx, \Sigyy$ are both full rank, we can define $\min\{p,q\}$ canonical direction pairs by \Cref{eq:def-via-Sig} and then extend this to give a pair of bases $(u_k)_k,(v_k)_k$, for $\R^p,\R^q$, respectively.
The situation is more subtle when $\Sigxx, \Sigyy$ are not of full rank; we discuss this case in detail in \Cref{app:cca-foundations}.
We shall loosely use \textit{canonical decomposition} to refer to this full collection of vectors, or to the associate collection of canonical variates. 

In practice, a statistician will not have access to the population covariance matrices, but only samples.
We assume these data are realisations of independent and identically distributed copies of $(X,Y)$.
The \textit{sample} CCA problem is to estimate certain aspects of the canonical decomposition, for example the first few canonical direction pairs, from these samples.
The classical method for the sample problem replaces the true covariance matrices in (\ref{eq:def-via-Sig}) with sample covariance matrices; see \Cref{sec:sample-cca}.

\subsection{Initial contribution: graphical CCA (gCCA)}
Existing regularised CCA methods assume regularity of the canonical directions and apply penalties to enforce this. 
The main alternatives are ridge CCA (rCCA) \citep{vinod_canonical_1976}, corresponding to an $\littlel{2}$ penalty on the directions, and a family of sparse CCA (sCCA) methods \citep{gao_sparse_2016,mai_iterative_2019}, which penalise the $\littlel{1}$-norm of the canonical directions; $\littlel{1}$ penalties are widely used in high-dimensional statistics to promote sparsity \citep{hastie_statistical_2015,wainwright_high-dimensional_2019}.
A family of methods based on the earlier work of \citet{witten_penalized_2009} are also widely referred to as sparse CCA methods, even though the orthogonality constraints are defined rather differently; following \citet{mihalik_canonical_2022} we call these sparse Partial Least Squares (sPLS) methods. It is important to consider sPLS methods because they are widely applied, with well-designed open-source implementations \citep{tibshirani_pma_2020}.
We define versions of rCCA, sCCA and sPLS in \Cref{sec:existing}, with further discussion and a wider literature review in \Cref{app:existing-further}.

Theoretical work has shown that certain methods can effectively estimate the canonical directions when the true underlying directions are sparse \citep{gao_sparse_2016,mai_iterative_2019}.
However, it is hard to justify why sparse directions should be assumed in real world datasets.
We posit that sparsity of the population precision matrix can be a more natural structural assumption, derived from a high degree of conditional independence in a Gaussian model; see \Cref{sec:glasso-background}. Thus, we propose graphical CCA (gCCA), a plug-in method exploiting the powerful Graphical Lasso estimator \citep{yuan_model_2007,friedman_sparse_2008}. Full details are given in \Cref{sec:our-approach}, which includes non-asymptotic upper bounds on the estimation error.
However, these theoretical results are only as informative as their assumptions permit;
to understand whether gCCA is better in practice, we need to compare it to existing regularisation CCA methods on real data.

\subsection{Main contribution: frameworks for comparison and interpretation}
In real-world applications, regularised CCA presents several hurdles to the practitioner. First, an ever-growing toolbox of estimators demands model comparison. Tuning penalty parameters and choosing the number of canonical directions (essentially, the variate subspace dimension) add further layers of complexity. There is little guidance in the literature on how to choose appropriate criteria for evaluation, which will depend on the number of canonical correlations under consideration. The real challenge lies in realistic datasets, which deviate from idealized models used for theoretical analysis and rarely exhibit clear gaps in the canonical correlations. Most commonly, practitioners choose the number of directions arbitrarily or only consider the top pair, leading to potential inconsistencies between methods and unreliable conclusions, especially in data-scarce settings. This underscores the urgent need for systematic analysis frameworks which allow the practitioner to evaluate and contrast a range of regularised CCA estimators. 

The relevance of different evaluation criteria depends on the downstream task of interest.
To explore this issue, we consider a broad family of evaluation criteria, which we introduce in \Cref{sec:evaluating-cca}.
This family includes two main classes: criteria for correlation signal captured, and criteria for estimation accuracy (for directions, variates and loadings).
For each such class, one can consider multiple successive directions individually or as successive subspaces.
Moreover, each class of criteria can be applied in the \textit{oracle} case where the true distribution is known (relevant for synthetic data) or in the \textit{empirical} case where the true distribution is unknown and sample splitting is required (relevant for real data).

Our motivation for graphical CCA originated from a project studying the human gut microbiome with the hope of generating insights to help understand and treat Irritable Bowel Disease (IBD).
We wished to apply CCA to explore the relationship between metabolites and enzymes in the human gut.
The structure of biochemical pathways provides strong motivation for expecting conditional independencies in the data.
We use a dataset related to this problem for illustrative purposes, both in its original form, and to provide meaningful constructions for synthetic data using a parametric bootstrap approach.
We shall refer to this as the `Microbiome' dataset in the sequel, and give further details in \cref{sec:microbiome-dataset-intro}.

In \Cref{sec:synthetic-data} we systematically compare different CCA methods on such synthetic data using this broad family of criteria. The parametric bootstrap model derived from our Microbiome dataset exhibits markedly different characteristics to illustrative models typically studied in the regularised CCA literature. We make the following main observations:
\begin{itemize}
    \item \textbf{gCCA works well.} It outperforms rCCA and sCCA across a wide range criteria. Additionally, sPLS performs significantly worse than the three CCA methods across both correlation and estimation criteria, as might be expected since PLS targets covariance rather than correlation structure. 
    \item \textbf{Cross Validation (CV) works well.} The CV-selected tuning parameter is usually near the oracle-optimal one for any given correlation objective; moreover such CV parameters often give near-oracle-optimal performance across a range of correlation and estimation criteria.
    \item \textbf{Correlation is easier than estimation.} Many successive directions with large correlation signal are recovered, even in high dimensions, but these are not always near the true canonical directions. 
    \item \textbf{Variates and loadings are easier to estimate than weights and successive subspaces are easier to estimate than successive pairs.} In fact, very different sequences of weights can lead to very similar variate subspaces. However, it is important to use comparable accuracy metrics.     
\end{itemize}

These observations have significant practical implications.
The former two observations reassure us that CCA is feasible in certain high-dimensional practical settings.
When using CCA for EDA, we would ideally only wish to spend time closely inspecting aspects of the estimated canonical decomposition which have good statistical properties.
The latter two observations suggest that there may be situations where it is sensible to closely inspect the loading vectors but not the weight vectors.
These observations are consistent with the geometry of CCA, previous theoretical results \citep{gao_sparse_2016,ma_subspace_2020} and previous empirical studies using CCA \citep{gu_simultaneous_2018}.
Evaluating estimation accuracy on real data is challenging, but accurate estimators should certainly be stable to sample splitting and small changes in hyperparameter choice.
Stability also to the choice of regularisation gives confidence that the observations are due to the signal in the data rather than the chosen structural assumptions.

In \Cref{sec:real-data} we propose a novel framework for applying CCA in practice, leveraging these observations.
This framework allows a practitioner to fit many different CCA methods with a range of penalty parameters and then compare the different estimates.
It includes visualisations to assess which estimates capture significant correlation and determine what aspects of those estimates are stable to sample splitting, hyperparameter, or choice of regularisation.
We suggest that a practitioner only focus on stable aspects of the estimates for closer inspection.

To aid such closer inspection we propose a powerful `biplot' visualisation method extending the proposal of \citet{ter_braak_interpreting_1990}.
These biplots show correlations between the original variables and the canonical variates, but can also be interpreted through a latent variable formulation of CCA.
They geometry of the biplots reflects the subspace nature of the CCA problem with a natural correspondence between registration of the estimates and transformations in the biplot.
The biplots are particularly insightful in three dimensions, and when analysed interactively with modern plotting libraries.
All our code is available on GitHub at \githubrepo, and we recommend the reader experiment with the Jupyter notebooks provided.
We hope that the framework we introduce may lead to high-dimensional CCA becoming more widely used for EDA.

\subsection{How to read this paper}
    Our initial contribution is in \Cref{sec:our-approach}; to put this into context we present existing methods in \Cref{sec:existing}. The main numerical experiments on synthetic data are in \Cref{sec:synthetic-data} with supporting definitions in \Cref{sec:evaluating-cca}; experiments on real data are in \Cref{sec:real-data} with supporting definitions in \Cref{sec:comparison-interpretation}. However, we first present essential background and notation for the CCA problem in \Cref{sec:cca-background}.
    
    We include only minimal mathematical statements in the main text; full details are given in the supplement. 
    In particular, \Cref{app:matrix-analysis-background,app:cca-foundations} provides a self-contained treatment of the mathematical foundations of CCA, including technical results which are often stated without proof or citation.
    This includes versions of all the results we present in the following \Cref{sec:cca-background}.
    We hope this may provide a useful reference for future work.
    
\section{CCA Background and Notation}\label{sec:cca-background}
\subsection{Formulation via Singular Value Decomposition}\label{sec:SVD-formulation-maintext}
    CCA is well known to be equivalent to a certain Singular Value Decomposition (SVD).
    The SVD is a very well studied object in matrix analysis \citep{stewart_matrix_1990}, with many attractive geometrical properties, which can be translated to CCA.
    In particular, suppose the within-view covariance matrices are invertible, and define
    \begin{equation}\label{eq:SVD-target-T-def}
    	T \defeq \Sigxx^{-1/2} \Sigxy\Sigyy^{-1/2}\; .
    \end{equation}
    Then the singular vectors of $T$ are $\Sigxx^{1/2}u_k$ and $\Sigyy^{1/2}v_k$ and the singular values are precisely $\rho_k$ where $u_k,v_k,\rho_k$ are the $k^\text{th}$ canonical pair and canonical correlation respectively.
    This follows easily from a standard formulation of SVD, as outlined in \Cref{app:cca-foundations}.
    An immediate consequence of this equivalence is the identity
	\begin{equation}\label{eq:sigxy-from-T-and-URV}
            \Sigxy 
            = \Sigxx^{1/2}\, T\, \Sigyy^{1/2}
            = \Sigxx \left(\sum_{k=1}^K \rho_k u_k v_k^\top \right) \Sigyy
	\end{equation}
    for any $K \leq \min(p,q)$ such that $\rho_k = 0$ for each $k>K$.

\subsection{Sample CCA}\label{sec:sample-cca}
	Write $\X = (x_1,\dots,x_N)^\top \in \R^{N \times p}, \Y = (y_1,\dots,y_N)^\top \in \R^{N \times q}$ for our data-matrices.
    We follow the convention, inspired by \citet{hastie_statistical_2015}, that vectors and matrices of samples (with leading dimension $N$) all appear in bold. 
	This includes the columns of the data matrices, which we write as $\X_i = ((x_n)_i)_{n=1}^N \in \R^N$ for $i \in 1,\dots,p$ and $\Y_j$ for $j \in 1,\dots,q$. 
	We assume the data is mean-centred. When implementing CCA algorithms in practice, there are many different ways one might apply this mean-centering, but this usually has a only a small impact on the results, so we will not discuss the issue further. With mean-centred variables, one may use the shorthand for sample covariance matrices:
	\begin{equation} \label{eq:sample-cov}
	    \C 
	    = \left(\begin{array}{ll} \Cxx & \Cxy \\
				\Cyx & \Cyy \end{array}\right)
		\def\arraystretch{1.3}
		= \left(\begin{array}{ll} \ninv \X^\top \X & \ninv \X^\top \Y \\
				\ninv \Y^\top \X & \ninv \Y^\top \Y \end{array}\right)\; .
	\end{equation}
	
	Hence the classical CCA estimator solves
	\begin{equation}\label{eq:cca-obj}
	\begin{split}
	\underset{u \in \mathbb{R}^{p}, v \in \mathbb{R}^{q}}{\operatorname{maximize}}\,& u^\top \Cxy v \\
	\text{subject to  }& u^\top \Cxx u \leq 1,\: v^\top \Cyy v \leq 1,\\ 
	&u^\top \Cxx u_j= v^\top \Cyy v_j =0 \text{ for } 1 \leq j \leq k-1 \;.	
	\end{split}
	\end{equation}

\subsection{Variate-based formulation of CCA}\label{sec:CCA-variate-space-population}
	Rewriting the program from \Cref{eq:def-via-Sig} in terms of covariances between (the canonical) variates rather than covariance matrices yields
    \begin{equation}\label{eq:def-via-population-covs} 
	\begin{split}
	\underset{u \in \mathbb{R}^{p}, v \in \mathbb{R}^{q}}{\operatorname{maximize}}\,
	& \Cov(u^\top X, v^\top Y) \\
	\text{subject to  }& \Var(u^\top X) \leq 1,\: \Var(v^\top Y) \leq 1,\\ 
	& \Cov(u^\top X, u_j^\top X) = \Cov(v^\top Y, v_j^\top Y) =0 \text{ for } 1 \leq j \leq k-1 \;.	
	\end{split}
	\end{equation}
 
    Before giving a natural sample analogue, we need to introduce notation for sample covariances and correlations.
    Let $\ve, \vf \in \R^{N}$ be vectors with mean zero.
    Then we define
    \begin{align*}
    \empCov(\ve,\vf) \defeq \ninv \ve^\top \vf, \quad
    \empVar(\ve)  \defeq  \empCov(\ve,\ve), \quad 
    \empCorr(\ve,\vf) \defeq \frac{\empCov(\ve,\vf)}{\sqrt{\empVar(\ve) \empVar(\vf)}} \; . 
    \end{align*}
    Then the sample analogue of \Cref{eq:def-via-population-covs} is
	\begin{equation}\label{eq:sample-cca-cov-def}
	    \begin{split}
	\underset{u \in \mathbb{R}^{p}, v \in \mathbb{R}^{q}}{\operatorname{maximize}}\,
	& \empCov(\X u, \Y v) \\
	\text{subject to  }& \empVar(\X u) \leq 1,\: \empVar(\Y v) \leq 1,\\ 
	& \empCov(\X u, \X u_j) = \empCov(\Y v, \Y v_j) =0 \text{ for } 1 \leq j \leq k-1 \;.	
	\end{split}
	\end{equation}

\begin{remark}\label{rem:sample-covariate-breakdown}
    It is now easy to see how CCA breaks down when $q \geq n$.
    Indeed, suppose $\Y$ is of full row rank; then for any $u \in \R^p$ there is some $v \in \R^q$ such that $\Y v = \X u$, and therefore $\empCorr(\X u, \Y v) = 1$.
    Therefore all the canonical correlations are in fact 1.
\end{remark}

\begin{remark}[Random variables to vectors of samples]
    Here we saw that the sample CCA version could be formulated by replacing random variables with sample vectors and corresponding population covariance operators with sample covariance operators.
    We will see this repeatedly in the rest of the article.
\end{remark}

\subsection{Further notation and background}
    It will often be helpful to use the shorthand 
    \begin{equation}\label{eq:Uk-Vk-block}
    U_K = \left( u_1 , \dots,u_K \right) \in \R^{p\times K}, \quad V_K = \left( v_1 , \dots,v_K \right) \in \R^{q \times K}
    \end{equation}
    for matrices formed by concatenation of successive canonical directions.

    We shall also use the notion of \textit{canonical angles}.
    These are a sequence of angles that reflect how well aligned two subspaces of a vector space are, and are well-studied in matrix analysis \citep{stewart_matrix_1990}.
    They are used to define the $\sin\Theta$ distance and $\cos\Theta$ similarity between subspaces, which we will use to evaluate how close estimated subspaces are from a ground truth or between each other.
    Definitions and an extended background may be found in \Cref{sec:canonical-bases}.

\section{Existing Regularised Sample CCA Methods}\label{sec:existing}
A large number of methods for regularised CCA have been proposed.
These methods generally add some regularisation to (\ref{eq:cca-obj}) to obtain an objective function whose optimum defines an estimator.
For brevity we only describe three methods that are representative of most of these approaches; we will compare our method with these representative methods in \Cref{sec:synthetic-data,sec:real-data}.
We give a more thorough review and discussion of the many variants and extensions of these representative methods in \Cref{app:cca-broader-background}.

\subsection{Ridge-regularised CCA (rCCA)}
Ridge regularisation for CCA \citep{vinod_canonical_1976} is the oldest and most established regularised CCA method \citep{uurtio_tutorial_2017}. 
For our purposes, it is most convenient to define estimated canonical directions through the successive programs
\begin{equation}\label{eq:ridge-cca-obj}
	\begin{split}
	\underset{u \in \mathbb{R}^{p}, v \in \mathbb{R}^{q}}{\operatorname{maximize}}\,& u^\top \Cxy v \\
	\text{subject to  }& u^\top \left((1-c_x) \Cxx+ c_x I_p\right) u \leq 1\; ,\: v^\top \left((1-c_y)\Cyy + c_y I\right) v \leq 1\;,\\ 
	&u^\top \left((1-c_x) \Cxx+ c_x I_p\right) u_j= v^\top \left((1-c_y)\Cyy + c_y I\right) v_j =0 \text{ for } 1 \leq j \leq k-1 \;,	
	\end{split}
\end{equation}
where $c_x,c_y \in [0,1]$ are tuning parameters. 
The estimator can therefore be viewed as a plug-in estimator, corresponding to regularised covariance estimate
\begin{align*}
    \hat{\Sigma}^{\text{ridge}}
	    = \left(\begin{array}{ll} (1-c_x) \Cxx+ c_x I_p & \Cxy \\
				\Cyx & (1-c_y)\Cyy + c_y I \end{array}\right)
		\def\arraystretch{1.3} \; .
\end{align*}
Note that taking $c_x=c_y=0$ recovers the classical population CCA.
By contrast, taking $c_x=c_y=1$ replaces the within-view covariance matrices by identity matrices; analogously to the discussion in \Cref{sec:SVD-formulation-maintext} this precisely to an SVD of the sample cross-covariance matrix $\Cxy$.
We follow \citet{mihalik_canonical_2022} and refer to this case as Partial Least Squares (PLS). 
In fact there is a whole family of methods referred to as PLS in the applied statistics literature; these are all closely related to the SVD of $\Cxy$ but often defined by iterations, rather than explicit objectives.
We refer the interested reader to \citet{rosipal_overview_2006} for a broader review.

\subsection{Sparse Partial Least Squares (sPLS)}\label{sec:SPLS}
One of the most common methods for sparse CCA used in practice is the Penalised Matrix Decomposition (PMD) approach proposed in \citet{witten_penalized_2009} and implemented in the R package PMA \citep{tibshirani_pma_2020}. In fact PMD is a more general method for computing low rank approximations to matrices, which can also be applied other problems such as sparse Principal Component Analysis. In the context of multiview data, their method estimates the $k^\text{th}$ pair of canonical directions by optimisers of
	\begin{equation}\label{eq:obj_witten}
	\begin{split}
		\operatorname{maximize}_{u, v}\:& u^\top \Cxy^{(k)} v \\
		\text { subject to}\: & \|u\|_{2}^{2} \leqslant 1, \enspace\|v\|_{2}^{2} \leqslant 1, \\
		& \|u\|_1 \leqslant s_{x},\enspace \|v\|_1 \leqslant s_{y} \;, 
	\end{split}
	\end{equation}
where 
\begin{equation*}
    \Cxy^{(k)} = \Cxy - \sum_{l=1}^{k-1} d_{l} u_l v_l^\top, \quad d_{l}=u_l^\top \Cxy v_l 
\end{equation*}
are deflated versions of the cross-covariance, and $s_x, s_y$ are tuning parameters.
This deflation encourages orthogonality between successive estimates, rather than being forced explicitly.

An important difference between this program (\ref{eq:obj_witten}) and the sample CCA program (\ref{eq:sample-cca-cov-def}) is that the constraints $\uk^\top \Cxx \uk = \vk \Cyy \vk = 1$ in  have been replaced by $\uk^\top \uk = \vk^\top \vk = 1$.
This is mainly for computational convenience. 
The justification is that the two quantities are equivalent when $\Cxx = I_p, \Cyy = I_q$;
the authors claim that this is a reasonable approximation in certain high-dimensional genomics applications, where within view population covariances $\Sigxx$, $\Sigyy$ may be assumed to be near identity.


However, in cases where this approximation is poor, the method does not correspond to CCA. 
Indeed, on removing the $\littlel{1}$ penalty terms the programs recover the SVD of $\Cxy$, and so correspond to PLS.
Following \citet{mihalik_canonical_2022}, we therefore refer to this method as sparse PLS (sPLS).


\subsection{Sparse CCA (sCCA)}\label{sec:sCCA}
Many works have since proposed computationally feasible CCA estimators which penalise the $\ell_1$-norm of canonical directions.
For brevity we will only compare to the method proposed in \citet{suo_sparse_2017}; this is arguably the most natural generalisation of sample CCA (\ref{eq:def-via-Sig}) and is faster or more performant than existing alternatives. This sCCA method defines estimated canonical directions through the successive programs
\begin{equation}\label{eq:obj-suo}
	\begin{split}
		\operatorname{minimize}_{u, v} \;& -u^\top \Cxy v  + \lambda_x \|u\|_{1}+\lambda_y \|v\|_{1} \\
		\text { subject to} \;& u^\top \Cxx u \leqslant 1\;, \enspace v^\top \Cyy v \leqslant 1\;, \\
        &u^\top \Cxx u_j = v^\top \Cyy v_j =0 \text{ for } 1 \leq j \leq k-1 \, .
	\end{split}
\end{equation}

\subsection{Notes on implementation and application}

Algorithmic details for solving the programs \Cref{eq:ridge-cca-obj,eq:obj_witten,eq:obj-suo} are only of secondary interest in this article.
It is important that algorithms are reliable and fast enough to run on a large dataset for many different tuning parameter values in a reasonable amount of time, but beyond this point the limiting factor is human ability to interpret the output.

We now briefly outline the algorithms we used in the forthcoming numerical experiments.
The rCCA program \Cref{eq:ridge-cca-obj} is in fact equivalent to population CCA on the regularised joint covariance matrix so can be solved with standard numerical algorithms for CCA; we used the efficient implementation in the valuable \ccazoo{} python package \citep{chapman_cca-zoo_2021}.
We implemented the remaining two methods ourselves.
The sPLS program \Cref{eq:obj_witten} can be solved efficiently by iterative soft-thresholding and rescaling \citep{witten_penalized_2009}.
The sCCA program \Cref{eq:obj-suo} is harder to solve. The original algorithm from \citet{suo_sparse_2017} used alternating convex search (ACS), i.e. iteratively optimise over $u$ with $v$ fixed and $v$ with $u$ fixed. We where able to speed up this algorithm significantly by interlacing partial updates for $u$ and $v$; we give full details and further background in \Cref{app:sCCA-implementation}.

One potential disadvantage of the sPLS and sCCA programs is that they have non-convex objectives.
Therefore there may be many local minima, and the solution returned may depend on the initialisation used. 
In practice, we did not find this to a problem: reasonable choices of initialisation found local minima which were `good' even if not globally optimal.

Another potential problem with the sPLS and sCCA programs is that they have a potentially large number of tuning parameters. Though this provides flexibility, it is very expensive to perform tuning parameter selection over a high-dimensional grid of parameters; even a two dimensional parameter grid will be significantly more expensive than a one parameter family.
We therefore tied all the tuning parameters to be the same in our numerical experiments; this had the added advantage of easier comparability with rCCA.

We next define our proposed CCA method: graphical CCA (gCCA).
As we will see, it avoids these two potential disadvantages: it has pleasant optimisation properties from a convex formulation and only a single tuning parameter.

\section{Our method: Graphical CCA (gCCA)}\label{sec:our-approach}

\subsection{Background for the Graphical Lasso}\label{sec:glasso-background}
    Graphical models are a well-studied class of probabilistic models that use a graph structure to constrain the distribution of a random vector and have applications across statistics and machine learning \citep{bishop_pattern_2006}.
    These models associate a vector-valued random variable $Z=(Z_1,\dots,Z_{\bar{p}})$ with the vertex set $\mathcal V=\{1,\dots,\bar{p}\}$ of some graph $\mathcal{G}=(\mathcal{V},\mathcal{E})$, but the probabilistic constraint can take a variety of different forms, depending on the statistical setting \citep{lauritzen_graphical_1996}.    
    A central statistical problem is to estimate the edge set $\mathcal E$ from many independent observations of $Z$ and is called structure estimation.
    
    Gaussian graphical models have particularly pleasant properties \citep{lauritzen_graphical_1996}.
    These models are associated with undirected graphs and constrain the probability distribution so that pairs of variables $Z_i,Z_j$ are conditionally independent given $(Z_k)_{k\neq i,j}$ whenever $(i,j)\notin \mathcal{E}$.
    It is a fundamental result \cite[Proposition 5.2]{lauritzen_graphical_1996} that such conditional independencies correspond to zeros in the precision matrix $\Omega = \Var(Z)^{-1}$:
    \begin{lemma}
    Let $Z \in \R^{\bar{p}}$ be a Gaussian random vector with (invertible) precision matrix $\Omega$. Then $Z_i \perp Z_j \mid (Z_k)_{k\neq i,j} \iff \Omega_{ij} = 0$.
    \end{lemma}    
     
    The Graphical Lasso (GLasso) is an established approach to structure estimation,  which involves adding an $\ell_1$ penalty to a Gaussian likelihood to define a convex program. Formally, suppose we have a collection of samples $z_1 ,\dots,z_N$ from some zero-mean multivariate Gaussian with empirical covariance matrix $\mathbf{C} = \frac{1}{N}\sum_{n=1}^N z_n z_n^\top \in \R^{\bar{p}\times \bar{p}}$ (in practice, we apply this to mean-centred samples). Then our estimator for the precision matrix is defined by
		\begin{equation}\label{eq:glasso}
		\hat{\Omega} \in \operatorname{arg max}_{\tilde{\Omega}\succeq \mathbf{0}} \{ \log \det \tilde{\Omega} - \operatorname{trace}(\mathbf{C}\tilde{\Omega}) - \lambda \rho_1(\tilde{\Omega})\}\;,
	\end{equation}
	where $\lambda \in (0,\infty)$ is a penalty parameter, and $\rho_1(\Omega)$ is the $\littlel{1}$-norm of the off-diagonal entries of $\Omega$.
	This objective was originally proposed in \citet{yuan_model_2007} but soon prompted a wide variety extensions and analysis; see \citet{hastie_statistical_2015} for an overview and further references. 
 
    The GLasso has good theoretical behaviour under the assumption that the true precision matrix $\Omega$ is sparse.
    Indeed, strong theoretical guarantees can be proven even for non-Gaussian data by interpreting the objective \Cref{eq:glasso} as an $\littlel{1}$-penalised Bregman divergence \citep{ravikumar_high-dimensional_2011}.
    However, beyond certain special cases \citep{loh_structure_2013}, there is in general no reason to assume that non-Gaussian data will have sparse precision matrix, even under a sparse graphical model.
    Nevertheless, one may expect approximately Gaussian data to be approximately sparse and therefore for the GLasso to perform well on a variety of real datasets. Indeed, the GLasso is now widely used across a variety of scientific disciplines, as testified by the thousands of citations accrued by early papers by \citet{yuan_model_2007,friedman_sparse_2008}.

\subsection{Definition of our estimator}
	To apply the Glasso to CCA, we take $\bar{p} = p+q$ and define random variables $Z \in \R^{\bar{p}}$ by concatenation of the variables $(X,Y)$; then the covariance of $Z$ is precisely $\Sigma$ as defined in (\ref{eq:sig-partition-def}).
	Our key observation is that $\Omega$ determines $\Sigma$ and therefore also determines the canonical decomposition. 
	We propose first estimating $\Omega$ by $\hat{\Omega}$ using a graphical lasso algorithm, and then computing canonical decomposition corresponding to this estimate. 
	For clarity we present this in pseudo-code in Algorithm \ref{alg:gCCA}. We note that there are various different GLasso implementations available, and any of them could be used. Furthermore, our plug-in procedure could be modified to accept alternative sparse precision estimators \citep[see][and~extensions]{cai_constrained_2011}, as well as more complex estimators from sparse plus low-rank decompositions or latent variable models \citep[e.g.][]{chandrasekaran_latent_2012}.
	
	\begin{algorithm}
		\caption{Graphical CCA (\glasso{})}\label{alg:gCCA}
		\begin{algorithmic}
			\Procedure{\texttt{Graphical-CCA}}{$\X,\Y$}
			\vspace{4pt}
			\State $\C \gets \texttt{sample-covariance}(\X,\Y)$ \Comment{As defined by (\ref{eq:sample-cov})}
			\vspace{4pt}
            \State $\hat{\Omega} \gets \texttt{glasso} (\C)$ \Comment{Use any GLasso variant}
            \vspace{4pt}
			\State $\hat{\Sigma} \gets \hat{\Omega}^{-1}$
            \vspace{4pt}
            \State $\hat{U},\hat{V},\hat{R} \gets \texttt{canonical-decomposition}(\hat{\Sigma})$ \Comment{Use any numerical method}
			\vspace{4pt}
			\State \textbf{return} $\hat{U},\hat{V},\hat{R}$ 
			\vspace{4pt}
			\EndProcedure
		\end{algorithmic}
	\end{algorithm}
 
\subsection{Theoretical guarantees}
    Theoretical guarantees can be obtained by working through Algorithm \ref{alg:gCCA} and showing that there is `good' estimation at each step.
    We now outline the key ideas of this argument, present the main result of \Cref{prop:sketch-graphical-CCA-guarantee}, and leave full mathematical details and further discussion to \Cref{app:graphical-cca}.
    
    When working through Algorithm \ref{alg:gCCA}, we were able to leverage some sophisticated results from the mathematical statistics literature.
    To show the precision estimate is good, we can use the theory for the graphical lasso from \citet{ravikumar_high-dimensional_2011}, as we explain later in this section.
    To show this implies good CCA estimates, we work via the SVD form of CCA from \Cref{sec:SVD-formulation-maintext}, and the target matrix $T$ for SVD from \Cref{eq:SVD-target-T-def}.
    By using a variant of the Davis-Kahan theorem \citep{yu_useful_2015}, one can prove that a good estimate $\hat{T}$ of the true SVD target $T$ will lead to good estimates of its singular value decomposition.
    
    The main technical work remaining was to show that a good precision estimate gives a good estimate of $\hat{T}$.
    We formalised this by showing that the function $\hat{\Omega} \mapsto \hat{T}$ is uniformly Lipschitz, as stated in \Cref{prop:lipschitz-plug-in} and proved in \Cref{app:lipschitz-plug-in}.
    The proof amounts to matrix analysis and tracking Lipschitz constants through each step of the formula that expresses $\hat{T}$ in terms of $\hat{\Omega}$.
    This only required the following simple assumption on the minimal and maximal eigenvalues of the within-view covariance matrices, which is natural in the context of CCA, and used in previous analysis of high-dimensional CCA \citep{gao_sparse_2016}.

    \begin{restatable}[Well-conditioned within-view variances]{assump}{wellConditionedWithinViewVariances}
    \label{ass:well-conditioned-within-view}
    There exist constants $M, m \in (0,\infty)$ such that
        $$\norm{\Sigxx}_\textrm{op},\norm{\Sigyy}_\textrm{op} \leq M,\quad\norm{\Sigxx^{-1}}_\textrm{op} \;, \norm{\Sigyy^{-1}}_\textrm{op} \leq m\;.$$
    \end{restatable}
    
    \begin{restatable}[Lipschitz plug-in] {proposition}{lipschitzPlugIn}
        \label{prop:lipschitz-plug-in}
        Assume Assumption \ref{ass:well-conditioned-within-view}$(M,m)$.
        Then there exists a $C \in (0, \infty)$ only depending on $M,m,\rho_1$ such that whenever $\norm{\Omega - \hat{\Omega}}_\textrm{op} \leq \frac{1}{4M}$ we have 
        \begin{equation}
            \norm{\hat{T} - T}_\textrm{op} \leq C \norm{\hat{\Omega} - \Omega}_\textrm{op}\; .
        \end{equation}
    \end{restatable}

    Before we can state our main result, we need to describe the results of \citet{ravikumar_high-dimensional_2011}, which give finite sample bounds on the error of $\hat{\Omega}$.
    It is most convenient for us to use their operator-norm analysis.
    Their rates critically depend on the maximal node degree and number of edges of the graph $G$, but also depend on the following structural quantities determined by the true covariance matrix $\Sigma$, or equivalently its inverse $\Omega$, the precision matrix.
    
    \begin{restatable}[Structural assumptions for GLasso]{assump}{structuralAssumptionsGlassoRavikumar}
    \label{ass:glasso-structural-assumptions}
    Assume that: 
    \begin{itemize}
        \item the full covariance matrix $\Sigma = {\Omega}^{-1}$ has $\ell_\infty$-operator-norm $\kappa_{\Sigma}$
        \item the log-det objective is $\kappa_{\Gamma}$-strictly convex when restricted to the non-zero entries of $\Omega$ 
        \item \textit{mutual incoherence:} no zero interaction is too well aligned with the true interactions, quantified by a constant $\alpha$. 
    \end{itemize}    
    \end{restatable}
    
    \begin{restatable}[Graphical CCA guarantee] {proposition}{sketchGraphicalCCAGuarantee}
        \label{prop:sketch-graphical-CCA-guarantee}
        Consider observing independent samples $z_n = (x_n,y_n)_{n=1}^N$ from a distribution with sparse joint precision matrix $\Omega$.
        Let $d$ be the maximal node degree of the corresponding interaction graph, and $s$ be the total number of edges in the graph.
        Assume that the standardised random variables in both views are sub-Gaussian with some common parameter.
        Assume Assumptions \ref{ass:well-conditioned-within-view}$\,(M,m)$ and \ref{ass:glasso-structural-assumptions}$\,(\kappa_{\Gamma}, \kappa_{\Sigma}, \alpha)$ hold.
        Then there exists some constant $C$ only depending on these structural quantities $(m,M,\kappa_{\Gamma},\kappa_{\Sigma},\alpha)$ and the top canonical correlation $\rho_1$, such that for sample size $N \gtrsim d^2 \log p$ and appropriate choice of penalty parameter $\lambda$ we have
        \begin{equation}\label{eq:sin-theta-bound}
            \|\sin \Theta(X^\top \hat{U}_{K}, X^\top U_{K})\|_{\mathrm{F}} 
            \leq 
            C \frac{K^{1 / 2}}{\rho_{K}^{2}-\rho_{K+1}^{2}} 
            \sqrt{\frac{\min\{s+\bar{p},d^2\} \log \bar{p}}{N}}
        \end{equation}
        with high probability.
    \end{restatable}

    We now note certain properties of this bound.
    Firstly, the bound considers subspace error, quantified by $\sin\Theta$ distance, in variate space. One could use Assumption \ref{ass:well-conditioned-within-view} to propagate this to an error bound in weight-space, but this would add an extra factor of $m$ to the constant, and $m$ may well be large. 
    Secondly, the rate $\sqrt{\frac{\log p}{N}}$ is standard for $\littlel{1}$ regularised problems, is the same order as the minimax bound from \citet{gao_sparse_2016} for sparse-weight CCA, and we conjecture that it is minimax in this sparse-precision regime.
    Thirdly, the strength of the result from \citet{ravikumar_high-dimensional_2011} is really in the regime when $d^2 \ll s+p$, and bounds error in high-dimensional regimes $p/n \to \infty$ provided that $d^2 \log \bar{p} / n \to 0$.
    Also note that the eigen-gap $\rho_K^2 - \rho_{K+1}^2$ inherited from the Davis-Kahan bound means that the bound becomes vacuous as the eigengap tends to zero --- we should not expect good estimation if there is small separation of correlation signal; similar eigen-gaps are present in all previous bounds for CCA, e.g., \citet{gao_sparse_2016, ma_subspace_2020}.

    Finally, we note that this general proof idea could give a number of alternative bounds, by using different matrix perturbation results, taking slightly different theorems from \citet{ravikumar_high-dimensional_2011}, or indeed using a different regularised precision matrix estimator, such as the latent variable GLasso \citep{chandrasekaran_latent_2012}.    

\subsection{Relationship to existing regularised CCA methods}
    There are certain connections between sparse precision matrices and sparse canonical directions.
    In fact, canonical directions can inherit block-wise sparsity from block-wise sparsity of the precision matrix; we now state this formally, and give a proof in \Cref{app:sparse-precision-to-directions}.
    We also give a second, alternative proof using partial correlations in \Cref{app:partial-correlation-to-sparse-directions} to provide complementary intuition.
        
    \begin{restatable}[Sparse directions from sparse precision] {proposition}{sparseDirections}
        \label{prop:sparse-directions-from-sparse-precision}
        Let $A \subset [p]$ correspond to a subset of the $X$ variables.
        Suppose that $(\Omega_{xy})_{aj} = 0$ for all $a \in A, j \in [q]$.
        If $\rho_k>0$ then any $k^{\uth}$ canonical variate can be written as $X^\top u_k$ where  $\operatorname{supp}(u_k) \subset [p] \setminus A$.
    \end{restatable}

    Based on this relationship, we might well expect certain similarities between the behaviour of gCCA to sCCA because both will tend to enforce some sparsity of the estimated directions.
    However, this is only a loose heuristic because there are many cases where \Cref{prop:sparse-directions-from-sparse-precision} is not relevant. 
    On the one hand, it is easy to construct a sparse precision matrix such that each variable in view 1 was related to a variable in view 2 and vice versa; then \Cref{prop:sparse-directions-from-sparse-precision} would be vacuous.
    On the other hand, is also easy to construct covariance matrices which have very sparse directions, but have non-sparse precision matrices.

    There are also close parallels between gCCA and rCCA.
    In particular, both can be seen as applying population CCA to some regularised covariance estimate.
    Moreover, each of these regularisations is specified by a single scalar tuning parameter and enforces that the covariance estimate is strictly positive definite.
    Correspondingly, one may expect close behaviour between gCCA and rCCA, especially for small tuning parameters.

    These heuristic parallels are reflected in our empirical results: gCCA often appears to perform regularisation `somewhere between' sCCA and rCCA.
    In the remaining sections we therefore plot gCCA between sCCA and rCCA for ease of comparison.

\subsection{Notes on implementation and application}
    To implement \Cref{alg:gCCA} we use the excellent Graphical Lasso implementation in the \texttt{gglasso} python package \citep{schaipp_gglasso_2021}. 
    The computational complexity of the GLasso is $\mathcal{O}((p+q)^3)$ so becomes infeasible for very high-dimensional problems.
    However, because the algorithm requires a single convex optimisation problem followed by a SVD operation, rather than a sequence of successive problems, it is faster in practice than the sCCA and sPLS approaches for problems of up to a few hundred dimensions.

\section{Evaluating CCA Estimates}\label{sec:evaluating-cca}
    \textit{What does it mean for one CCA estimate to be `better' than another?}
    Answering this question is crucial for a meaningful comparison of different CCA algorithms on simulated data; it is also crucial for model selection on real data.

    A good CCA estimate will capture a lot of correlation signal and also be close to the underlying population quantity; these two ideas both lead to evaluation approaches. We present approaches based on correlation in \Cref{sec:correlation-metrics} and approaches based on estimation in \Cref{sec:estimation-metrics}.
    In each case we first suggest `oracle' quantities to be used when the population distribution is known, and relevant for evaluation on synthetic data; then we suggest `empirical' quantities for real data situations when the population distribution is unknown, and which are based on sample splitting.
    Additionally, in each case we both suggest approaches which work with canonical pairs successively and approaches that work with subspaces.

    The sample splitting approaches require additional notation.
    We will always be in the setting of $V$-fold cross-validation with folds $\nu = 1,\dots,V$.
    We write the data-matrices restricted to the $\nu\uth$ fold as $\X^{[\nu]},\Y^{[\nu]}$ (validation set) and the remaining data in the matrices $\X^{[-\nu]},\Y^{[-\nu]}$ (training set).
    Given an algorithm which obtains estimates $\hat{U}(\X,\Y),\hat{V}(\X,\Y)$, we introduce the shorthand $\hat{U}^{[-\nu]}=\hat{U}(\X^{[-\nu]},\Y^{[-\nu]}), \hat{V}^{[-\nu]}=\hat{V}(\X^{[-\nu]},\Y^{[-\nu]})$ and corresponding lower case quantities for direction estimates.

\subsection{Correlation captured}\label{sec:correlation-metrics}
\subsubsection{Oracle successive correlations}\label{sec:oracle-correlations}
    One obvious way to evaluate a single pair of estimated canonical directions is through the correlations between projections of the respective random variables onto these estimated directions (this is after all the optimisation objective in \Cref{eq:def-via-Sig}). We refer to these as \textit{oracle} correlations because computing them requires knowledge of the true covariance matrix, which is not available on real data:
    \begin{equation}\label{eq:cc-oracle}
        \rho^\text{oracle}(\hat{u},\hat{v}) 
        = \Corr(\hat{u}^{\top} X, \hat{v}^{\top} Y) 
        = \frac{\hat{u}^\top\Sigxy\hat{v}}{\sqrt{\uprule \hat{u} \Sigxx \hat{u}}\,\sqrt{\hat{v}^\top\Sigyy \hat{v}}}\; .
    \end{equation}

    If one has $K$ estimated pairs, one could consider the whole vector of successive oracle correlations $\left(\rho^\text{oracle}(\hat{u}_k,\hat{v}_k)\right)_{k=1}^K$ but usually desires a single scalar metric to evaluate the estimates. We propose a general family of such scalar metrics of the form
    \begin{align}\label{eq:succ-cc-agg}
        \texttt{succ-cc-agg}(f; \hat{U},\hat{V}) = f\left(\left(\rho^\text{oracle}(\hat{u}_k,\hat{v}_k)\right)_{k=1}^K\right)\;,
    \end{align}
    where $f: \R^* \to \R$ is some coordinate-wise increasing function.
    Choosing $f = \|\cdot\|_1$ corresponds to sums of correlations; choosing $f = \| \cdot \|_2^2$ gives sums of squared correlations, which are particularly natural in light of Section \ref{sec:comparison-biplots}. 
    Other choices of $f$ may also be interesting. For example the function $f: \rho \mapsto -\frac{1}{2}\sum_k \log(1-\rho_k^2)$ corresponds to mutual information in a Gaussian setting (see \Cref{app:mutual-information}) and is therefore also reasonable to consider.
    These choices of $f$ are justified by the following result, which we state precisely and prove as \Cref{prop:sum-convex-functions-CCA} in \Cref{app:multiple-correlations}.

    \begin{restatable}[Valid aggregation functions]{proposition}{combiningCorrelations}
            \label{prop:combining-correlations}
            Let $\phi : \R \to \R$ be a convex function such that $\phi(x) \geq \phi(-x)~ \forall x>0$. Then, top-$K$ CCA subspaces can be characterised by the program
            \begin{align*}
                \underset{\substack{\{u_1,\dots,u_K\}, \{v_1,\dots,v_K\}  \\ u_k^\top \Sigxx u_l = \delta_{kl}, v_k^\top \Sigyy v_l = \delta_{kl}: \,\forall k,l \in \{1, \dots, K\} }}{\operatorname{maximise}}{\vphantom{\Bigg(\Bigg)}\:\:\sum_{k=1}^K \phi \left((\rho^\text{oracle}(u_k,v_k)\right)} \; .
            \end{align*}
        \end{restatable}
        
    This naive approach of combining correlations of successive pairs has one major drawback though --- and that is when the estimated directions are not suitably orthogonal. 
    In this case we may end up double-counting some of the correlation signal.
    One way to resolve this is to consider instead subspace correlations.

\subsubsection{Oracle subspace correlations}
    Because the canonical correlations are unique, we can make the following definition.	
	\begin{definition}\label{def:cc-vector}
	    Write $\CCAcorrto{K}(X,Y) \in \R^K$ to be the vector containing the top-$K$ canonical correlations of the pair of random variables $(X,Y)$ as defined by (\ref{eq:def-via-Sig}). 
	    Write $\empCCAcorrto{K}(\X,\Y) \in \R^K$ to be vector containing the top-$K$ canonical correlations of the sample CCA problem corresponding to $(\X,\Y)$ as defined by (\ref{eq:sample-cca-cov-def}).
	\end{definition}
 
    We suggest analysing the vector of canonical correlations corresponding to the pair of $K$-dimensional random variables $\hat{U}_K^{\top} X, \hat{V}_K^{\top} Y$, 
    i.e. $\CCAcorrto{K}(\hat{U}_K^{\top} X, \hat{V}_K^{\top} Y)$ with the notation of Definition \ref{def:cc-vector}.
    Note from the discussion in Section \ref{app:cca-foundations} that this only depends on the subspace spanned by the random variables $\hat{U}_K^{\top} X, \hat{V}_K^{\top} Y$;
    therefore it only depends on the subspace spanned by the estimates $\hat{U}_K, \hat{V}_K$; these are the primary objects of interest if we are using CCA for dimension reduction. 
    Similarly to in (\ref{eq:succ-cc-agg}), we propose a family of metrics of the form
    \begin{align}\label{eq:subsp-cc-agg}
        \texttt{subsp-cc-agg}_K(f; \hat{U},\hat{V}) = f\left(\CCAcorrto{K}(\hat{U}_K^{\top} X, \hat{V}_K^{\top} Y)\right)\;.
    \end{align}

\subsubsection{Empirical correlations via cross validation}\label{sec:empirical-correlations}
    Each of the metrics above can be approximated by cross validation (CV).
    The CV analogue of (\ref{eq:cc-oracle}) is
    \begin{align*}
        \rho^{\text{CV}}(\hat{u},\hat{v}) = \frac{1}{V} \sum_{\nu=1}^V \empCorr\left(\X^{[\nu]} \hat{u}^{[-\nu]}, \Y^{[\nu]}\hat{v}^{[-\nu]}\right) \; .
    \end{align*}
    One CV analogue of (\ref{eq:succ-cc-agg}) is
    \begin{align}\label{eq:cv-succ-cc-agg}
        \texttt{cv-succ-cc-agg}_K(f; \hat{U},\hat{V}) 
        \defeq 
        \frac{1}{V} \sum_{\nu=1}^V f\left(\left(
        \empCorr\left(\X^{[\nu]} \hat{u}^{[-\nu]}_k, \Y^{[\nu]}\hat{v}^{[-\nu]}_k\right)
        \right)_{k=1}^K \right) \; .
    \end{align}
    Finally, a CV analogue of the subspace correlation notion from (\ref{eq:subsp-cc-agg}) is 
    \begin{align}\label{eq:cv-subsp-cc-agg}
        \texttt{cv-subsp-cc-agg}_K(f; \hat{U},\hat{V}) 
        \defeq 
        \frac{1}{V} \sum_{\nu=1}^V f\left(\empCCAcorrto{K}\left(\X^{[\nu]} \hat{U}_K^{[-\nu]}, \Y^{[\nu]} \hat{V}_K^{[-\nu]}\right)\right) \; .
    \end{align}

\subsection{Estimation Error}\label{sec:estimation-metrics}
  \begin{table}[t]
    \caption{Estimation metrics}
    \label{tab:metric-summary-estimation}
    \centering 
    \renewcommand{\arraystretch}{1.1}
    \begin{tabular}{rlrl} 
    \toprule
    \multicolumn{2}{c}{Oracle version} & \multicolumn{2}{c}{Empirical version} \\ 
    \cmidrule(lr){1-2} \cmidrule(lr){3-4}
    Quantity  & Definition & Quantity  & Definition  \\
    \cmidrule(lr){1-2} \cmidrule(lr){3-4}
        \texttt{wt-uk} & $\sin^2\Theta(u_k,\hat{u}_k)$ &
        \texttt{wt-uk-cv} & $\binom{V}{2}^{-1} \sum_{\nu < \nu'} \sin^2 \Theta(\hat{u}_k^{[-\nu]},\hat{u}_k^{[-\nu']})$ \cr 
        \\ 
         \texttt{vt-uk} & $\sin^2\Theta(\Sigxx^\half u_k,\Sigxx^\half \hat{u}_k)$ &
         \texttt{vt-uk-cv} & $\binom{V}{2}^{-1} \sum_{\nu < \nu'} \sin^2 \Theta(\X \hat{u}_k^{[-\nu]},\X \hat{u}_k^{[-\nu']})$ \cr \\
         \texttt{wt-Uk} & $\sin^2\Theta(U_k,\hat{U}_k)$ &
         \texttt{wt-Uk-cv} & $\binom{V}{2}^{-1} \sum_{\nu < \nu'} \sin^2 \Theta(\hat{U}_k^{[-\nu]},\hat{U}_k^{[-\nu']})$ \cr \\ 
         \texttt{vt-Uk} & $\sin^2\Theta(\Sigxx^{1/2}U_k,\Sigxx^{1/2}\hat{U}_k)$ &
         \texttt{vt-Uk-cv} & $\binom{V}{2}^{-1} \sum_{\nu < \nu'} \sin^2 \Theta(\X \hat{U}_k^{[-\nu]},\X \hat{U}_k^{[-\nu']})$  \\     
    \addlinespace[0.4ex] \bottomrule 
    \end{tabular}
  \end{table}
  
 \subsubsection{Oracle}
    As in \Cref{sec:correlation-metrics} there are many reasonable choices of metrics for evaluating estimation.
    Indeed, it is again reasonable to consider estimates successively or with a subspace-based approach.
    But now there are at least three different aspects of the estimates of interest: variates, weights or loadings.

    For simplicity we will only consider squared $\sin\Theta$ distances.
    These have a natural sense of scale, with a maximum value of $K$ between two $K$ dimensional subspaces, this linear structure helps makes the graphs more readable.

    To apply in variate space we simply pre-multiply directions by $\Sigxx^\half$, and to apply in loading space we pre-multiply by $\Sigxx$ (here, we focus on estimation of variates).
    Recall that in this section we are primarily concerned with numerical experiments where $\Sigxx$ is indeed known, as are the true directions and so these distances are easily computable.
    We define all relevant quantities, and their corresponding shorthand, in the left hand column of \Cref{tab:metric-summary-estimation}.

\subsubsection{Empirical: cross-validated instability}
    It becomes difficult to estimate these quantities empirically, but one can define related notions of \textit{instability}.
    Though these do not give good estimates of the oracle error, they give approximate lower bounds and may tell us that a CCA estimator has too high variance to be accurate.
    
    The right hand column of \Cref{tab:metric-summary-estimation} defines various notions of instability, and gives corresponding shorthand.
    We propose looking at average difference between pairs of estimates trained on different training folds.
    Notice that the definitions given for instability in variate space (\texttt{vt-uk-cv} and \texttt{vt-Uk-cv}) use the full dataset $\X$ to obtain sample canonical variates; though this may raise concern about `re-using' data, it has lower variance than conceivable alternatives based on further sample splitting and gave sensible empirical results. 


\section{Synthetic Data}\label{sec:synthetic-data}

    Using synthetic data is essential for developing statistical methodology but conclusions rely heavily on the choice of model and metrics used.
    Previous work has used simplistic synthetic models, with very clear signal.
    By contrast, in this section we will only consider synthetic data designed to be biologically plausible, and obtained via parametric bootstrap resampling of a real dataset.
    We use mean-centred multivariate normal distributions in line with previous work, which can be justified by central limit theorem arguments and does not appear restrictive.
    We keep this section short, with minimal content to support our claims in \Cref{sec:intro} and to provide context for our real data suggestions in \Cref{sec:comparison-interpretation,sec:real-data}.

    \begin{figure}[t]
        \centering  \includegraphics[width=\textwidth]{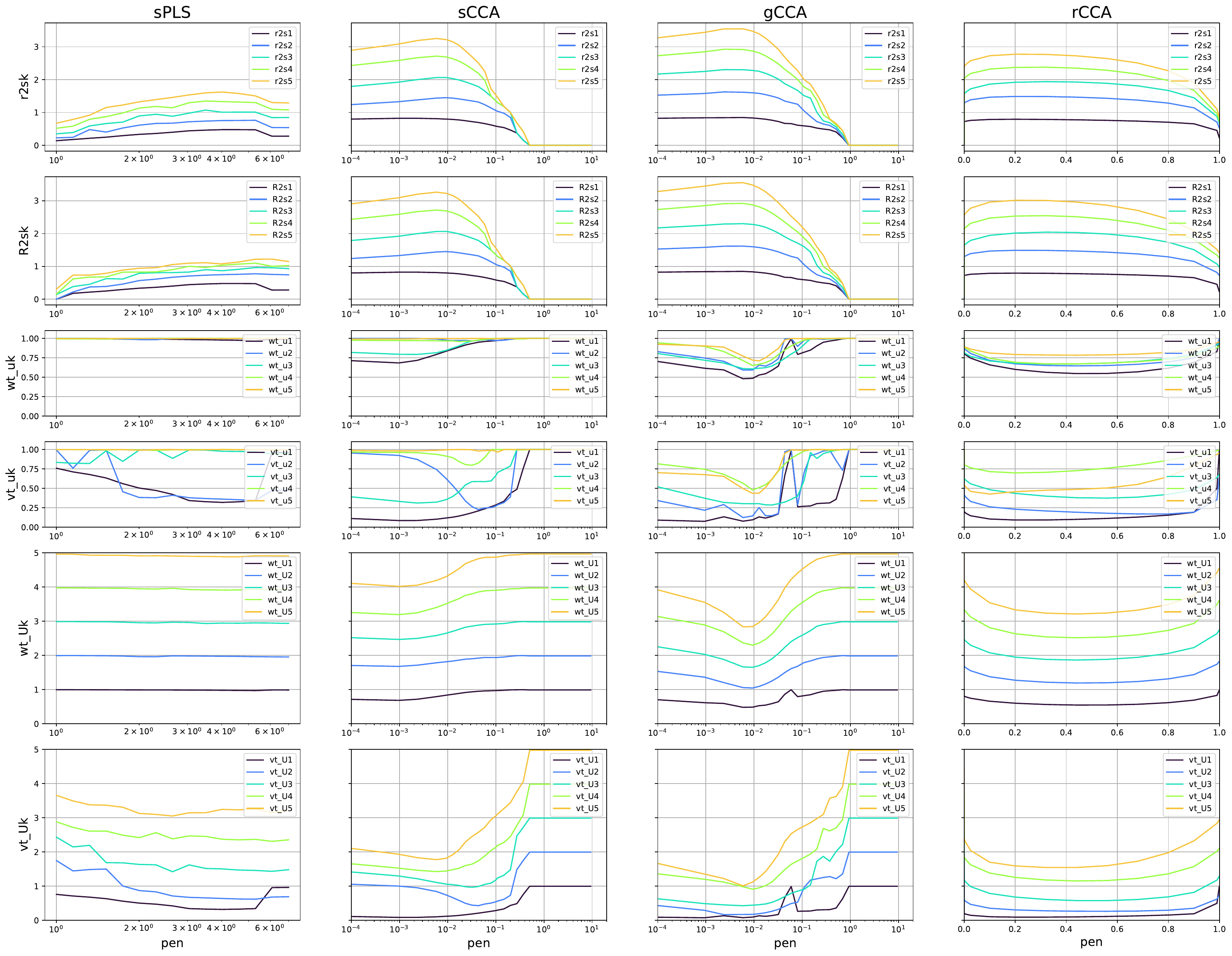}
        \caption{Oracle metrics on the parametric-bootstrapped Microbiome dataset. Each row corresponds to a different (type of) metric, each column to a different algorithm, and the x-axis to the penalty parameter. See \Cref{tab:metric-summary-estimation,tab:metric-summary-correlation} for a glossary of the legends.}
        \label{fig:boots_mb_panel}
    \end{figure}

    \begin{figure}[t]\centering
    \begin{minipage}{1\textwidth}\centering
            \includegraphics[width=0.9\textwidth]{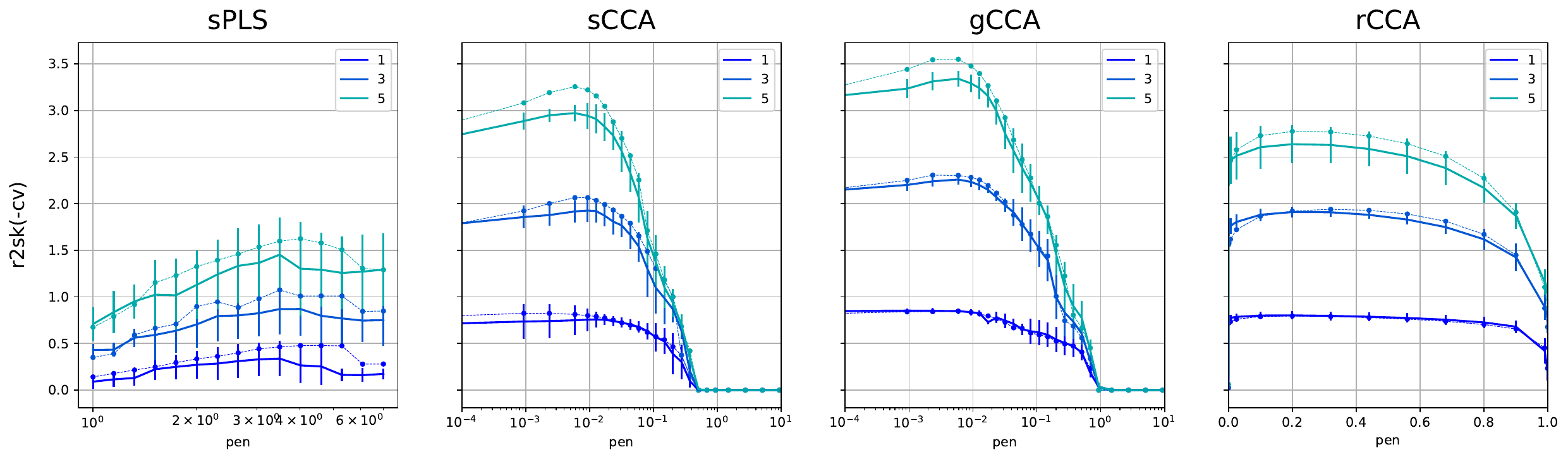}
        \\
            \includegraphics[width=0.9\textwidth]{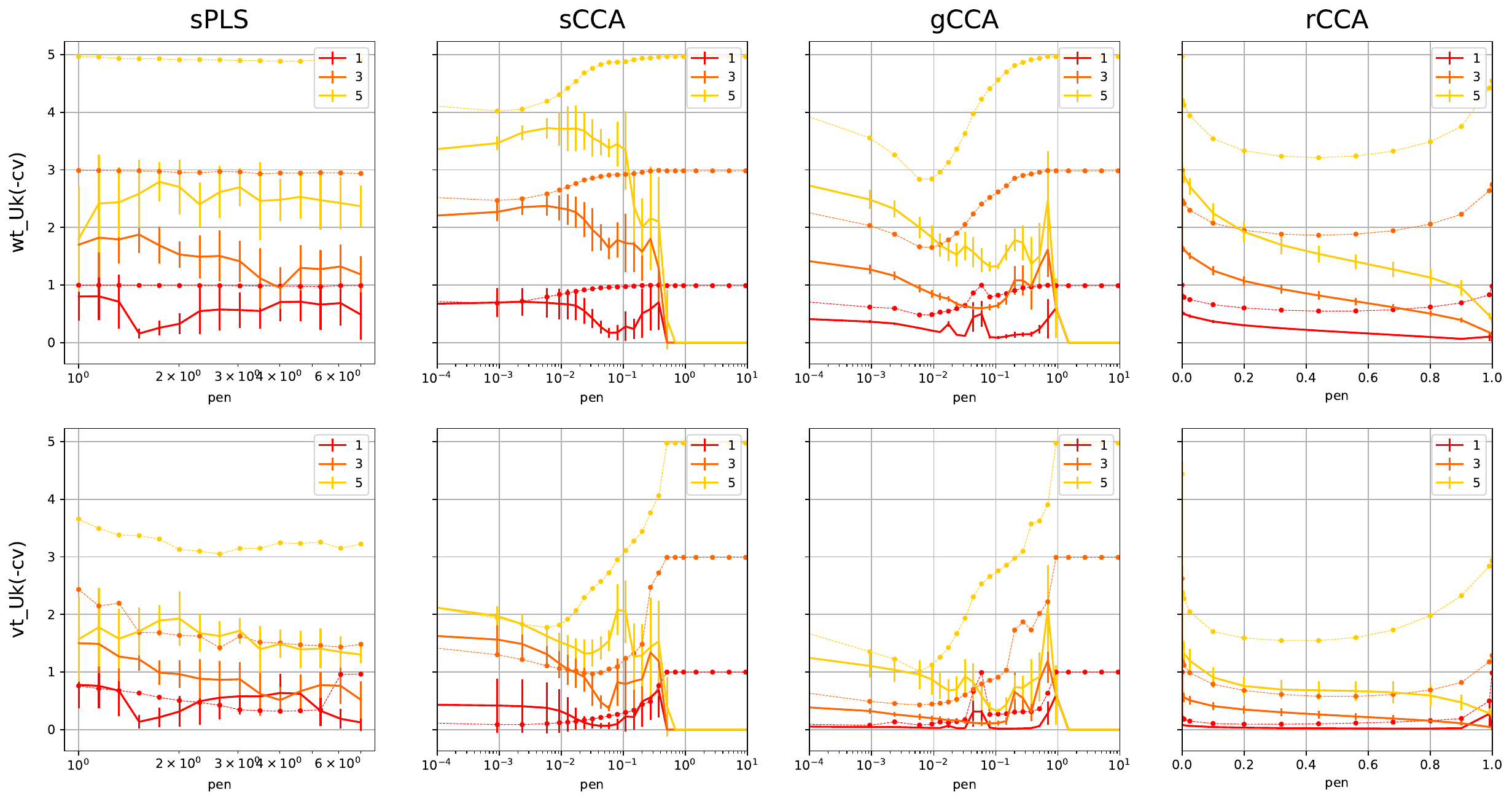}
    
        \caption{Correlation and instability metrics over the regularisation path for the four methods on the parametric bootstrapped Microbiome dataset; error bars for the aggregated quantities; oracle values are dotted for comparison.
        Top row: sums of squared correlations, as defined by \texttt{r2sk-cv} and \texttt{r2sk} in \Cref{tab:metric-summary-correlation} for $k=1,3,6$.
        Middle row: subspace instability in weight space, as defined by \texttt{wt-Uk-cv} and \texttt{wt-Uk} in \Cref{tab:metric-summary-estimation} for $k=1,3$.
        Bottom row: subspace instability in variate space, as defined by \texttt{vt-Uk-cv}and \texttt{vt-Uk} in \Cref{tab:metric-summary-estimation} for $k=1,3$.}
        \label{fig:boots-mb-corr-and-stab}
    \end{minipage}
    \end{figure}

    In particular, we consider a synthetic version of the Microbiome dataset described in \Cref{sec:microbiome-dataset-intro}. 
    To construct the dataset, we first obtained an estimated precision matrix with the GLasso algorithm; we used the tuning parameter $\alpha=0.0059$, which maximised the sum of the top-5 squared test correlations (\texttt{r2s5-cv} from \Cref{tab:metric-summary-correlation}).
    We then generated a set of $n=500$ samples from a multivariate normal distribution with this precision matrix.
    We will see that many behaviours observed closely mirror those observed on the original dataset in \Cref{sec:real-data}, suggesting that this synthetic dataset is a reasonable proxy.

    We felt that for this section it was clearest to illustrate our points using a single set of pseudo-random samples with an arbitrary random seed; we reassure the reader that our observations are robust to the choice of this random seed. 
    We perform a (much) larger-scale simulation study with the same data-generating mechanism in \Cref{app:parametric-bootstrap-vary-n-still-glasso-reg}, where we vary the number of samples generated, and average over many random seeds; this showcases the sample efficiency of \glasso{}.
    Further synthetic data experiments based on toy data-generating mechanisms, as well as different parametric bootstrap models are included in \Cref{app:toy-synthetic-experiments,app:parametric-bootstrap-alternative-suo-regularisation} respectively.

\subsection{Metrics behave differently}\label{sec:metric-choice-bootstrap}
    \begin{table}[t]
        \caption{Correlation metrics}
        \label{tab:metric-summary-correlation}
        \centering 
        \begin{tabular}{rlcrlc} 
        \toprule
        \multicolumn{3}{c}{Oracle version (error)} & \multicolumn{3}{c}{Empirical version (instability)} \\ 
        \cmidrule(lr){1-3} \cmidrule(lr){4-6}
        Quantity  & Definition & Refs. & Quantity  & Definition & Refs. \\
        \cmidrule(lr){1-3}\cmidrule(lr){4-6}
        $\texttt{r2sk}$  &  $\texttt{succ-cc-agg}_K(\norm{\cdot}_2^2; \hat{U},\hat{V}) $ & (\ref{eq:succ-cc-agg}) 
                & $\texttt{r2sk-cv}$ & $\texttt{cv-succ-cc-agg}_K(\norm{\cdot}_2^2; \hat{U},\hat{V}) $ & (\ref{eq:cv-succ-cc-agg})\\   
        $\texttt{R2sk}$  &  $\texttt{subsp-cc-agg}_K(\norm{\cdot}_2^2; \hat{U},\hat{V}) $ & (\ref{eq:subsp-cc-agg}) 
                & $\texttt{R2sk-cv}$ & $\texttt{cv-subsp-cc-agg}_K(\norm{\cdot}_2^2; \hat{U},\hat{V}) $ & (\ref{eq:cv-subsp-cc-agg})\\   
        \bottomrule
        \end{tabular}
      \end{table}
  
    For our single bootstrapped dataset we show how the `oracle' behaviour of our different algorithms varies with tuning parameter.
    To evaluate correlation captured we look at sums of squared oracle correlations.
    To evaluate estimation error, we use $\sin^2\Theta$ distances, for both successive directions and for subspaces, and both in weight space and in variate space; these distances have convenient interpretations in terms of fraction of signal shared.
    We summarise these different metrics in \Cref{tab:metric-summary-correlation}, with the text labels that will be used in the legends of our plots (estimation metrics are in \Cref{tab:metric-summary-estimation} for reference).
    
    \Cref{fig:boots_mb_panel} supports the following observations:
    \begin{itemize}
        \item \textbf{many directions containing signal can be recovered:} the \texttt{r2sk} values in the first row are large --- for example the \glasso{} has squared fifth correlation of 0.6, corresponding to correlation of 0.8, which is very significant.
        \item \textbf{variates are easier to estimate than weights:} the losses for \texttt{vt-uk} are much smaller than for \texttt{wt-uk} and the losses for \texttt{vt-Uk} are much smaller than the losses for \texttt{wt-Uk}.
        \item \textbf{subspaces are easier to estimate than directions:} the losses for \texttt{vt-Uk} can be much smaller than aggregating the \texttt{vt-uk} losses.
        \item \textbf{optimal tuning parameters for different metrics are often similar:} for example the \glasso{} optimal parameters are almost all between $5\times 10^{-3}$ and $10^{-2}$; moreover all parameters in this range are near-optimal for most different metrics
        \item \textbf{there is often no obvious best model or set of hyperparameters:} in this example, it seems like the \glasso{} with parameter $6 \times 10^{-3}$ would be a good overall choice; however \suo{} has comparable correlation captured, and \ridge{} has comparable variate loss.
        \item \textbf{\wit{} does not perform CCA:} it captures meaningfully lower correlation than CCA methods with conventional orthogonality constraints, and produces weights with near-maximal loss. However, it does recover some of the shared variate subspace, because PLS and CCA subspaces can often be similar.
    \end{itemize}

        Because there is no obvious best model or set of hyperparameters, it may not be sensible to perform strict model selection.
        Instead, it may make more sense to consider a variety of different models with near-optimal hyperparameters, and compare the corresponding estimates.
        This will give a rough measure of how robust an observation drawn from these estimates is to the choice of regularisation; such robustness is certainly desired if such an observation is to merit further investigation.
        We provide more details of a possible approach in \Cref{sec:real-data}.
    
\subsection{Cross validation is effective}\label{sec:synth-cv-is-effective}
    In \Cref{fig:boots-mb-corr-and-stab} we compare CV notions of various criteria with their oracle counterparts.
    In each plot the CV quantities are plotted with solid lines and error bars, while the oracle quantities are plotted with fainter dotted lines; in \Cref{sec:real-data} we will use the same style of plots for CV quantities, but clearly cannot access the oracle quantities.

    The important observation is that all the CV correlation criteria closely match their oracle counterparts.
    Moreover, the penalty parameters optimising the CV quantities are usually the same as the parameters optimising the oracle quantities.

    The picture is much more complicated for the notions of stability.
    As one might expect, the estimates tend to be more stable for higher penalty parameter.
    However, the variation in stability is not monotonic, and for \suo{} and \glasso{} there appear local minima of empirical instability, though these do not align well with minima of oracle estimation error, or maxima of correlation signals.
    Also, note that the CV instability versions are nearly always below oracle estimation error counter parts --- confirming our expectation that instability does give a lower bound on estimation error.



\section{Comparison and interpretation}\label{sec:comparison-interpretation}
We now move on to the questions of key practical importance.
Given estimated CCA directions, how should we interpret them?
Moreover, if there are multiple CCA estimates that are plausible how can we compare them?
Following the results of \Cref{sec:synthetic-data} and motivation of \Cref{sec:intro}, we take a subspace-and-variate centric approach.
Our main goal is to introduce biplots for interpretation of CCA estimates.
We will also define a notion of overlap matrices to visually compare the overlapping signal of two CCA estimates.
First, it is important to discuss registration.

\subsection{Registration}\label{sec:registration}
    One practical problem when comparing CCA estimates is that they are not well aligned.
    Sometimes two estimates appear very different, but actually span very similar subspaces.
    One way to deal with this is registration.

    It is sensible to perform this registration `in variate space', because the variates are more stable than the weights; indeed, registering in `weight space' might fail to notice that two sets of very different weights can lead to very similar variates.
    However, beyond this, there are many different types of registration that might be reasonable.
    We now outline a reasonable family of techniques for registration, which are closely related, but may still yield complementary insights.

    Suppose we have some reference estimates $\hat{U}_K \in \R^{p\times K}$ and some new estimate for comparison $\hat{U}'_{K'} \in \R^{p \times K'}$ for some $K' \geq K$.
    We wish to find a $K$-dimensional subspace of columns of $\hat{U}'_{K'}$ that most closely correspond to $\hat{U}_K$.
    Let our data be $\X \in \R^{n \times p}$.
    Then our registered direction estimates are $\hat{U}'_{K'} M^*$ where
    \begin{align*}
        M^* \in \argmin_{M \in \mathcal{M}(K' \times K)} \| \X \hat{U}_K - \X \hat{U}'_{K'} M \|_F^2
    \end{align*}
    for some suitable set of matrices $\mathcal{M}(K' \times K) \subset \R^{K' \times K}$.

    By picking appropriate sets $\mathcal{M}$ we can register up to sign flips, permutations (and sign flips), orthogonal transformations, or indeed up to arbitrary linear transformations.
    In \Cref{app:registration} we show how to implement these types of registration and discuss their relationship to subspace distance notions.
    We have found each of these types of registration insightful in different contexts, as we will see in the remaining sections. 

\subsection{Overlap matrices}\label{sec:overlap-introduction}
    \newcommand{\bZ}{\textbf{Z}}
    \newcommand{\bW}{\textbf{W}}
    We now provide background material for a visualisation tool which we will demonstrate in  \Cref{sec:overlap-matrices-real-data}.
    \begin{definition}[Overlap matrices]\label{def:overlap-matrix}
        Given two matrices $\bZ,\bW \in \R^{N \times K}$ we call $\bZ^\top \bW \in \R^{K\times K}$ the \textit{overlap matrix} between $\bZ,\bW$; we define the \textit{squared overlap matrix} to be the elementwise-squared overlap matrix.        
    \end{definition}

    The key motivation for this definition is that if $\bZ,\bW$ have orthonormal columns $(\bZ_k)_{k=1}^K,(\bW_k)_{k=1}^K$, then from \Cref{sec:canonical-bases}, the squared cosine similarity between these subspaces is
    \begin{align}\label{eq:general-overlap-equation}
        \cos^2\Theta(\bZ,\bW) = \norm{\bZ^\top \bW}_F^2 = \sum_{k=1}^K \sum_{l=1}^K \langle \bZ_k, \bW_l \rangle^2 \; ,
    \end{align}
    which is precisely the sum of squared entries of the overlap matrix $\bZ^\top \bW$. 
    Moreover, by the same argument as in \Cref{eq:general-overlap-equation} we see that sum of entries of any squared overlap submatrix correspond to squared cosine similarity of the corresponding subspaces.
    The (squared) overlap matrices therefore give an interpretable graphical display of how the signal in $\bZ$ overlaps with the signal in $\bW$.    

    In \Cref{sec:overlap-matrices-real-data} we apply this to two matrices of estimated canonical variates, i.e.
    \begin{align}\label{eq:overlap-matrix-full-sample-default}
        \bZ = \X \hat{U}_K, \quad \bW = \X \hat{U'}_K;\quad \text{so } \bZ^\top \bW = \empCov(\X \hat{U}_K,\X \hat{U'}_K) \; .
    \end{align}
    If our estimated canonical variates are orthonormal (as they should be) then the cosine similarity interpretation applies, but even if not, the overlap matrix may still be informative. Moreover, by applying Gram--Schmidt one can plot the overlap matrix corresponding to successive subspaces; by comparing to the original overlap matrix this can give a clear visualisation of the degree of non-orthogonality of the data.
    Note that we can also apply this in the oracle setting%
    by instead considering the matrix
    \begin{align*}
        \Cov(U_K^\top X, \hat{U}_K^{\prime \, T} X) = (\Sigxx^\half U_K)^\top (\Sigxx^\half U'_K) \; .
    \end{align*}

    Note it is also straightforward to consider analogous `validation versions' for a sample split by taking
    \begin{align*}
        \bZ = \X^{[\nu]} \hat{U}_K^{[-\nu]}\;, \quad \bW = \X^{[\nu]} \hat{U'}_K^{[-\nu]} \; .
    \end{align*}

\subsection{Biplots}\label{sec:comparison-biplots}

    Biplots are a powerful tool for visualising canonical correlation analysis, originally proposed in \citet{ter_braak_interpreting_1990};
    these plot the variables from each of the two views of data in a shared low-dimensional space.
    Though used in a wide variety of contexts and presented in tutorials such as \citet{uurtio_tutorial_2017}, they appear under-utilised in scientific applications of CCA.
    The biplot was adapted in \citet{ter_braak_interpreting_1990} from an earlier notion of biplot for PCA \citep{gabriel_biplot_1971,jolliffe_principal_1986};
    we warn the reader that the word biplot is often used loosely in the applied literature to refer to a variety of plots based on similar ideas.
    
    The original biplot proposed in \citet{ter_braak_interpreting_1990} was for sample CCA and breaks down in the high-dimensional sample setting. 
    We instead present a population notion of biplot, whose sample version is equivalent to that of \citet{ter_braak_interpreting_1990} in low dimensions but which remains meaningful in high dimensions.
    Though we define this using the $X$-canonical variates, one could equally consider the $Y$-canonical variates; this is also true with all the related formulations in the following sections.

\subsubsection{Population biplots}
    \begin{definition}[Population Correlation Biplot]\label{def:pop-biplot-corr}
    Let $\{u_k, v_k\}_{k=1}^K$ be successive canonical directions. 
    The $X$-view biplot plots the $i\uth$ variable in view 1 and $j\uth$ variable in view 2 at positions $\phi(X_i), \phi(Y_j)$ defined by
        \begin{equation}
            \phi(X_i) \defeq \left(\Corr(X_i,u_k^\top X)\right)_{k=1}^K\;,  \quad 
            \phi(Y_j) \defeq \left(\Corr(Y_j,u_k^\top X)\right)_{k=1}^K\;,
        \end{equation}
    respectively.
    \end{definition}

    The vectors $\left(\Corr(X_i,u_k^\top X)\right)_{i=1}^p, \left(\Corr(Y_j,u_k^\top X)\right)_{j=1}^q$ are called (the $k^\text{th}$ pair of) \textit{structure correlations}.
    These can be seen as normalised versions of the canonical loadings via
    \begin{align*}
        (\Sigxx u_k)_i = \Cov(X_i, u_k^\top X) = \Var(X_i)^\half \Corr(X_i, u_k^\top X) \; .
    \end{align*}        
    The biplot can therefore be seen as a visualisation of the canonical loading vectors and interpreted via the latent variable formulation of CCA from \Cref{sec:dimension-reduction-reconstruction}.
    

    \begin{proposition}
        Let $\phi$ be the biplot mapping from \Cref{def:pop-biplot-corr}.
        \begin{itemize}
        \item The squared norm of a variable's representation is the proportion of the variable's variance explained by the the first $K$ canonical variates, and so is less than one. 
        Writing $\xi_k = u_k^\top X$ for the $k^\text{th}$ canonical variate for the $X$-view, we have
        \begin{align}\label{eq:biplot-magnitude-less-than-one}
            \norm{\phi(X_i)}_2^2 = \sum_{k=1}^K \Corr(X_i, \xi_k)^2 = \frac{\sum_{k=1}^K \Cov(X_i, \xi_k)^2}{\Var(X_i)} \leq 1 \; .
        \end{align}
        
        \item Inner products of representations with large squared norm approximate correlations between the corresponding variables: for any $i, i' \in [p]$
        \begin{align}\label{eq:biplot-inner-products-approximate-correlations}
            \abs*{\Corr(X_{i}, X_{i'}) - \langle \phi(X_{i}), \phi(X_{i'}) \rangle_2} 
            \leq \left(1 - \norm{\phi(X_{i})}^2 \right)^\half \left(1 - \norm{\phi(X_{i'})}^2 \right)^\half \; ,
        \end{align}
        and the analogous statement also holds for the variables in the $Y$-view.
        
        \item For inner products between representations in different views we have the stronger bound
        \begin{align}\label{eq:biplot-approx-good-between-views}
            \abs*{\Corr(X_i, Y_j) - \langle \phi(X_i), \phi(Y_j) \rangle_2} 
            \leq \rho_{K+1} \left(1 - \norm{\phi(X_i)}^2 \right)^\half \left(1 - \norm{\phi(Y_j)}^2 \right)^\half \; .
        \end{align}
        \end{itemize}
    \end{proposition}
    \begin{proof}
        \Cref{eq:biplot-magnitude-less-than-one,eq:biplot-inner-products-approximate-correlations} follow from the orthonormality of the $\xi_k$, whereas
        \Cref{eq:biplot-approx-good-between-views} also uses the fact that $\Cov(\xi_k, Y_j) = \rho_k \Cov(v_k^\top Y, Y_j)$ due to the canonical correlation structure.
    \end{proof}

    This final bullet point suggests that inner products between representations of variables in different views may be `good' approximations of the correlation between those variables.
    In fact, \citet{ter_braak_interpreting_1990} proves that this approximation is \textit{optimal} in a certain sense, using the Eckhart--Young inequality; we present the result with associated discussion in \Cref{sec:biplot-optimality-eckhart-young}.

\subsubsection{Sample biplots for regularised CCA, and their interpretation}    
    Conveniently, \Cref{def:pop-biplot-corr} has a natural sample equivalent.

    \begin{definition}[Sample correlation biplot]
        Given an estimate $\hat{U}_K$, respectively plot the $i\uth$ variable in view 1 and $j\uth$ variable in view 2 with coordinates
        \begin{equation*}
            (\empCorr(\X_i,\X \hat{u}_k))_{k=1}^K\; ,  \quad 
            (\empCorr(\Y_i,\X \hat{u}_k))_{k=1}^K\; .
        \end{equation*}
    \end{definition}

    If one uses estimates from (unregularised) sample CCA then observations \Cref{eq:biplot-magnitude-less-than-one,eq:biplot-inner-products-approximate-correlations,eq:biplot-approx-good-between-views} from the previous section carry through, with quantities replaced by their sample counterparts. 
    Unfortunately, as outlined in \Cref{sec:intro}, the break-down of sample CCA means that typically all sample canonical correlations are one ($\hat{\rho}_k = 1$ for each $k$) and therefore \Cref{eq:biplot-approx-good-between-views} is no stronger than \Cref{eq:biplot-inner-products-approximate-correlations}.

    The hope is that using estimates from appropriate regularised CCA methods will give good estimators for the true population biplots in a statistical sense.
    This appears justified in certain high-dimensional situations.
    Indeed, for the Microbiome dataset, the synthetic experiments display accurate estimation of the canonical variates (see \Cref{fig:boots_mb_panel}) and therefore also loading vectors and biplots.
    Additionally, in \Cref{sec:biplots-real-data} we see that variates for the different regularised CCA methods on the Microbiome dataset are extremely similar after appropriate registration (see \Cref{fig:microbiome-traj-comp-variate}); they are also stable to sub-sampling the data (see the stability plots in \Cref{fig:microbiome-corrs}).
    
    However, we observed empirically that the biplots may give meaningful visualisations of the data, even when they are poor statistical estimators for the population quantities.
    Indeed, for the `Breastdata' dataset that we discuss in \Cref{sec:contrast-with-nutrimouse-breastdata}, the biplots for \suo{} and \ridge{} look very different, so cannot both be good statistical estimates; yet both appear to recover meaningful chromosomal information (see \Cref{fig:bd-3d-plots}).


    \begin{remark}[CV for biplots is difficult]
        One could evaluate the statistical properties of biplots with cross-validation.
        One might estimate the canonical directions with some training subset of the data, and evaluate the empirical correlations using a held out testing subset of data; however, this leads to large variance in the plots, particularly with small sample sizes.
        We therefore suggest using the same (full) dataset to construct the estimates and visualise them in the biplot, but we cannot expect it to be a good estimate of the population biplot in general.
    \end{remark}

\section{Real Data}\label{sec:real-data}
    We now evaluate the various CCA methods on real data sets and propose a pipeline of tools for EDA.
    We describe this pipeline in \Cref{sec:pipeline-toolbox}, applied to our motivating microbiome dataset, which we first describe in \Cref{sec:microbiome-dataset-intro}.
    We also applied the pipeline to two existing datasets, the Nutrimouse dataset (available in the \texttt{CCA} R package \citep{gonzalez_cca_2021}, and originally from \citet{martin_novel_2007}), and the BreastData dataset (available in the PMA R package \citep{tibshirani_pma_2020} and originally from \citet{chin_genomic_2006}); these datasets have far fewer samples, leading to rather different CCA behaviour, see \Cref{sec:contrast-with-nutrimouse-breastdata}.
    We present only a minimal selection of plots to demonstrate our main points and leave more extensive supporting plots to \Cref{app:real-data-supporting-plots}.
    
\subsection{Microbiome dataset}\label{sec:microbiome-dataset-intro}
    This whole study was motivated by a desire to perform EDA on data concerning the human gut microbiome.
    We use an illustrative dataset obtained from the Integrative Human Microbiome Project, concerning the microbiome's role in Inflammatory Bowel Disease (IBD) \citep{lloyd-price_multi-omics_2019}.
    This dataset contains counts of microbial metabolites (C0s) and genes (K0s) collected logitudinally from patients with two different types of IBD diagnosis, Ulcerative Colitis or Crohn's Disease, as well as healthy control patients. 
    We filtered the raw counts data to ensure that there were no zero counts, then applied the following steps common in compositional data analysis: first convert raw counts into proportions, then apply a log transform, finally mean-centre the columns. 
    After this we had $p=711$ metabolites, $q=148$ proteins and $n=458$ samples.

\subsection{Background and setup}\label{sec:framework-setup}
    Our framing is very different to most previous work on high-dimensional CCA.
    In the extreme case, some of these works only consider the top pair of weight vectors, for a single penalty parameter chosen by CV, and essentially inspect the entries.
    By contrast, we would like to find subspaces containing `most' of the correlation signal, and would like our observations to be indicative of true population quantities.
    We therefore consider both estimated canonical variates and directions, and ensure any observations we make are at a minimum robust to sample splitting.
    In addition, we acknowledge that all our regularised CCA models are likely to be mis-specified;
    we expect that different sorts of regularisation may reveal complementary insights and therefore care more about model comparison than model selection.
    However this leads to the following difficult situation.

    Suppose that a variety of different estimators have been fitted, with a variety of penalty parameters, both to the original dataset and to different CV folds.
    We shall refer to a combination of algorithm and penalty parameter as an \emph{estimator}, as is the convention in the sklearn package. 
    We therefore have large family of direction estimates, indexed by algorithm, penalty, and fold. 
    We seek visualisation tools to help compare between this family of estimates at a high-level, and decide which estimates to analyse in more detail.

\subsection{A framework for applying high-dimensional CCA}\label{sec:pipeline-toolbox}
    As in \Cref{sec:evaluating-cca}, it is sensible both to compare correlations captured by different estimators and the differences between the estimates themselves.
    In each case, there is a trade-off between giving a high-level comparison of a large number of estimators, and richer, more-detailed comparisons of a smaller number of estimators.
    High level summaries of CV correlation and instability, and of distances between many estimates are described in \Cref{sec:corr-along-trajectories,sec:trajectory-comparison-subspace} respectively.
    Richer views of CV correlation captured and comparison between a small number of estimates are described in \Cref{sec:test-corr-decay,sec:overlap-matrices-real-data} respectively.
    Finally, we illustrate biplots in \Cref{sec:biplots-real-data}.

\subsubsection{Correlation and instability along trajectories}\label{sec:corr-along-trajectories}
    \begin{figure}[t]\centering
    \begin{minipage}{0.9\textwidth}
            \includegraphics[width=0.9\textwidth]{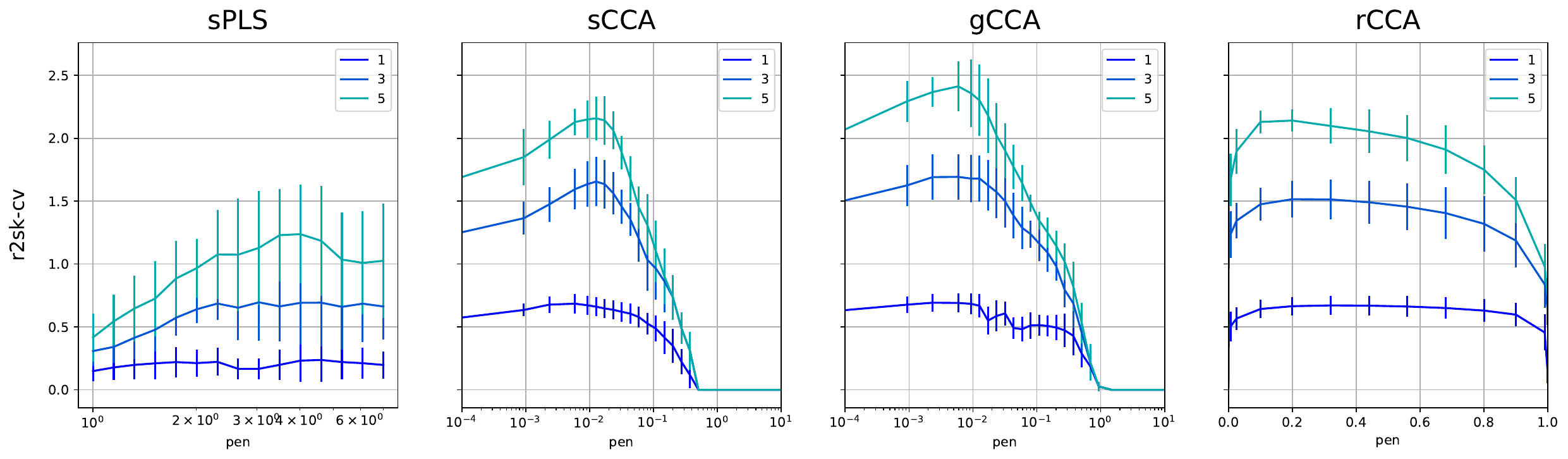}
        \\
            \includegraphics[width=0.9\textwidth]{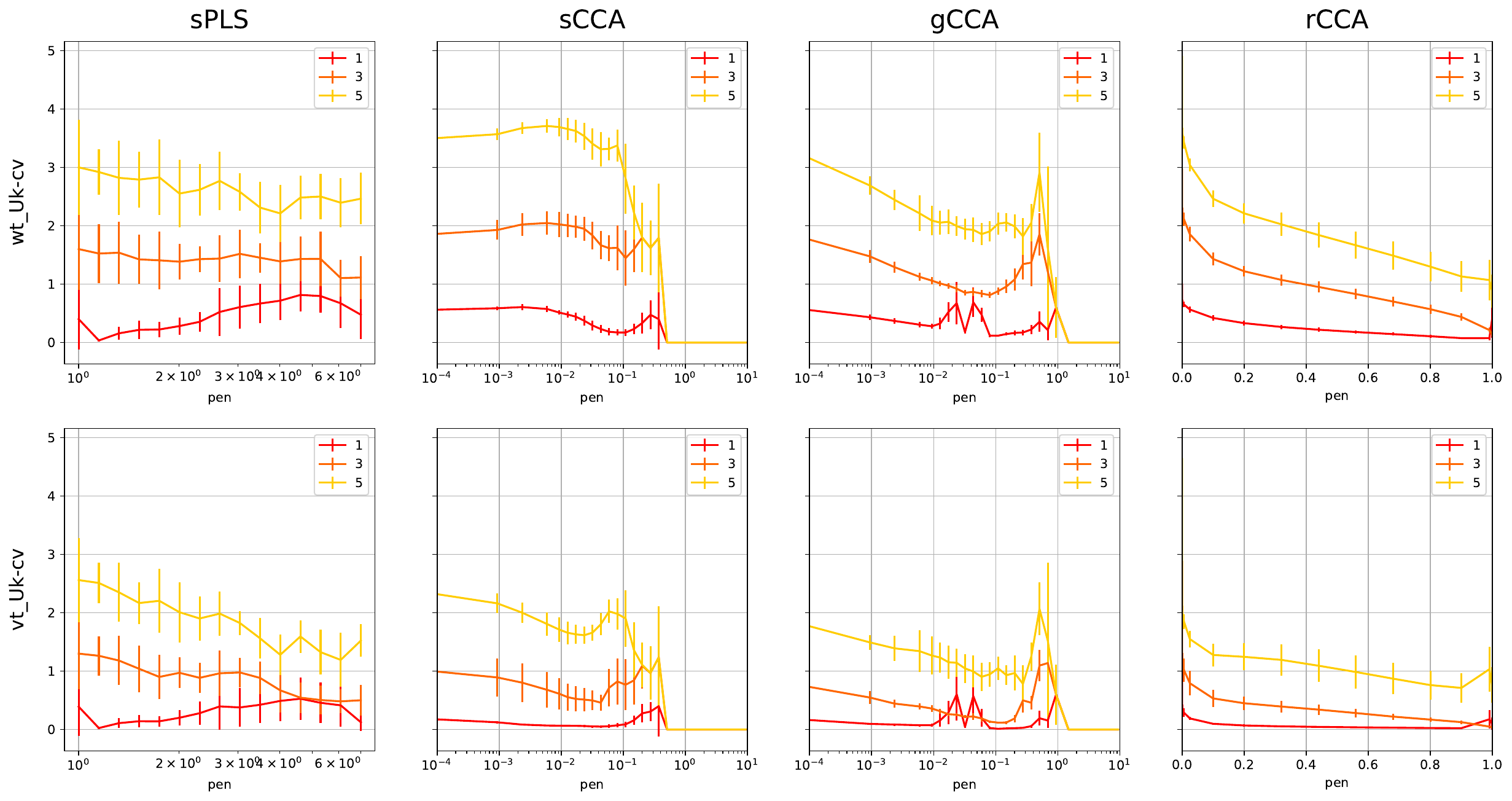}
        \caption{Top: CV sums of correlations as function of regularisation path for the four methods on the Microbiome dataset; error bars for the aggregated quantities.
        Bottom: stability both in weight space and variate space along the same trajectories.}
        \label{fig:microbiome-corrs}
    \end{minipage}
    \end{figure}
    
    \Cref{fig:microbiome-corrs} provides a way to visualise correlation and instability at a very high-level.
    It is analogous to \Cref{fig:boots-mb-corr-and-stab}: the top row shows cross-validated correlations, using both direction and subspace based approaches, while the bottom row shows subspace instability, both in direction space and variate space (see \Cref{sec:evaluating-cca}).
    This gives a visual comparison of all the different estimators (algorithm, penalty pairs) of interest, and illustrates the difficulty of regularising hard enough to encourage stability but while still capturing the main signal.
    
    The choices of $k$ of 1,3,5 are reasonable, but arbitrary, default dimensions to consider; it might be worth considering different values of $k$ after further analysis (particularly inspection of plots in Section \ref{sec:test-corr-decay}).    

    Figure \ref{fig:microbiome-corrs} does indeed look very similar to \Cref{fig:boots-mb-corr-and-stab}.
    We see evidence of the claims previously alluded to: many successive directions with significant signal are recovered; variates are much more stable than weights; gCCA captures slightly more correlation than the other CCA methods; sPLS captures significantly less correlation and is relatively unstable.
    We therefore omit sPLS from the remaining plots in this section, and defer further consideration of sPLS to \Cref{app:real-data-supporting-plots}.

\subsubsection{Decay of test correlations}\label{sec:test-corr-decay}
    \begin{figure}[t]\centering
         \includegraphics[width=0.9\textwidth]{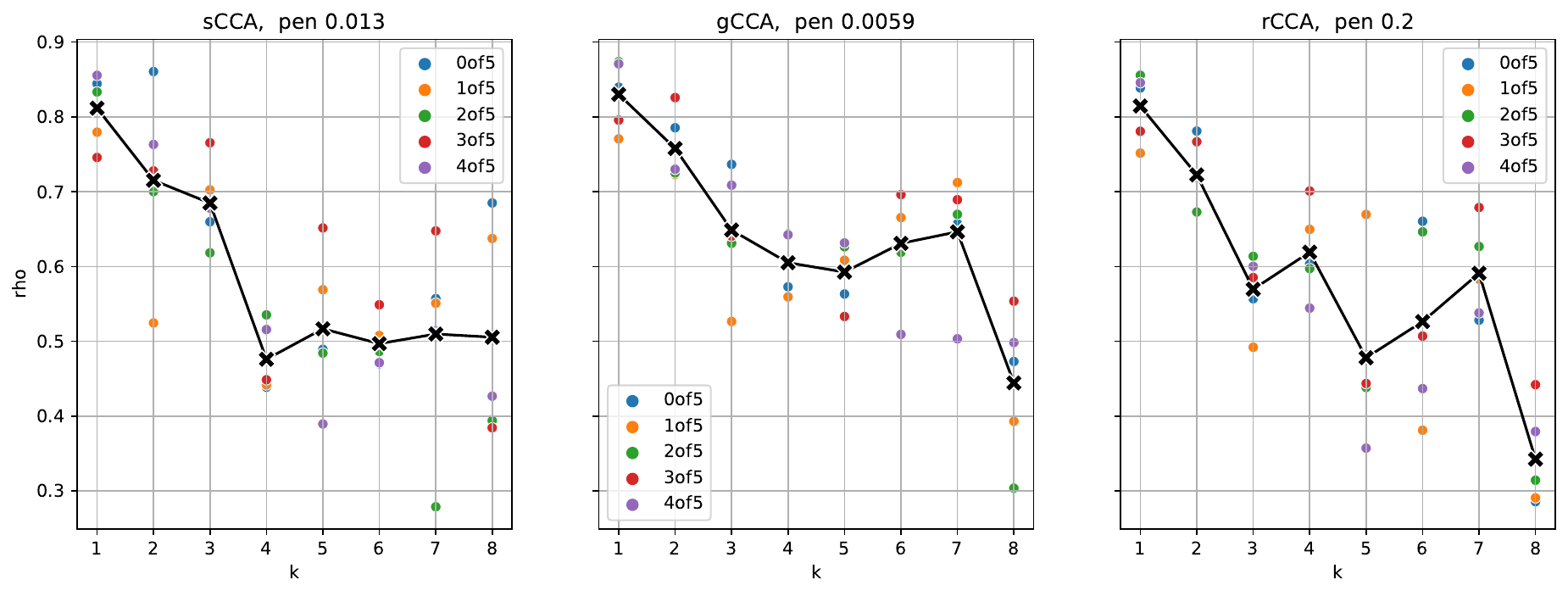}
         \caption{CV correlations (colours) and average values (black) for successive direction estimates using \suo{}, \glasso{}, \ridge{} on microbiome dataset; in each case \texttt{r2s5-cv} optimal penalty parameters were used.}
         \label{fig:corr_decay_microbiome_gglasso}
    \end{figure}
    
    Figure \ref{fig:corr_decay_microbiome_gglasso} shows how the test correlations captured by three estimators decreases over successive components on the Microbiome dataset; the estimators correspond to \texttt{r2s5-cv}-optimal penalties for \suo{}, \glasso{}, \ridge{} respectively.
    
    The main aim of the plot is to determine the number of components to consider.
    An `elbow' in the plot would indicate a sensible number of components.
    Although it can only compare a limited number of estimators, this plot gives much richer information than in Figure \ref{fig:microbiome-corrs}.
    Firstly, we can consider all successive directions up to a large maximum value without cramping the picture.
    Plotting individual folds rather than error bars gives clear visual guide of statistical significance (all significantly non-zero).
    Color-coding the folds gives even richer information that occasionally merits further investigation, see further discussion in \Cref{obs:color-coding}.    
    
    The main observations from Figure \ref{fig:corr_decay_microbiome_gglasso} are that a large number of directions carry significant correlation signal, and that there are no clear elbows.
    Indeed, for \glasso{}, there is very little decay of correlations between the 5th and 10th components, and there seems to be a large number of successive correlations around 0.6. 
    Moreover, all of the individual test correlations are well above zero; which suggests the signal is strongly significant.
    We also observe that \glasso{} consistently records higher correlations than \suo{} or \ridge{}.
    
    
\subsubsection{Trajectory comparison via subspace distance}\label{sec:trajectory-comparison-subspace}
    \begin{figure}[t]\centering
         \includegraphics[width=0.8\textwidth]{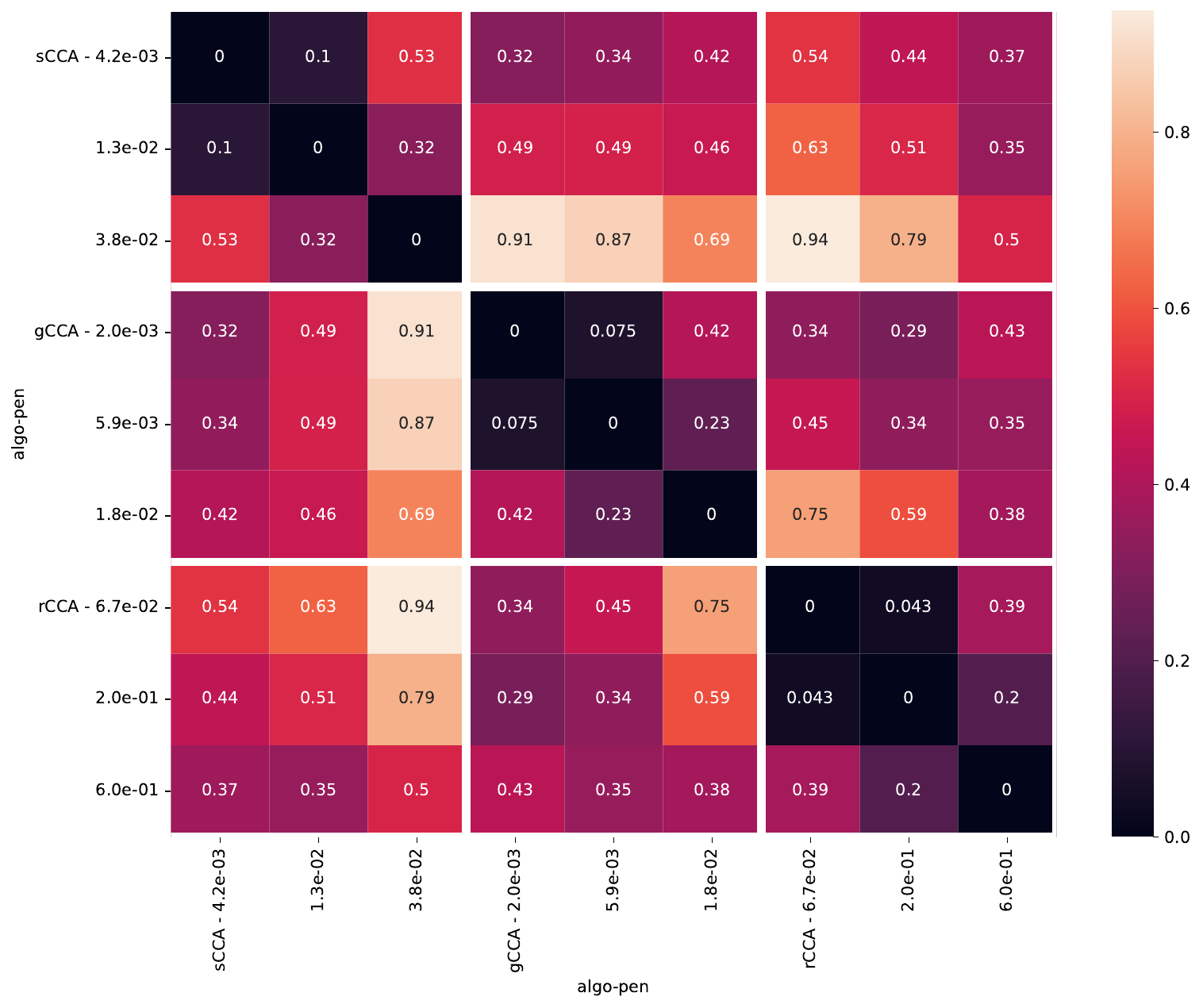}
         \caption{$\sin^2\Theta$ distances for top-3 variate subspace between pairs of full sample estimates; axes refer to three values of the penalty parameter for each of the three regularisation methods: sCCA, gCCA, and rCCA.}
         \label{fig:microbiome-traj-comp-variate}
    \end{figure}

    To get a high-level comparison of many different estimators, we suggest plotting a trajectory comparison matrix, as in \Cref{fig:microbiome-traj-comp-variate}.
    This is a symmetric matrix whose axes refer to various values of the penalty parameter for each of the regularisation methods (sCCA, gCCA, and rCCA), and may be thought of as `concatenated trajectories'; the entries display distances between the corresponding pairs of estimators.
    Figure \ref{fig:microbiome-traj-comp-variate} shows such a matrix on the Microbiome dataset using top-3 $\sin^2\Theta$ distance in variate space (i.e. \texttt{vt-U3} from Table \ref{tab:metric-summary-estimation}).
    As in the previous Section \ref{sec:test-corr-decay}, we chose the central penalty parameters by CV (with objective \texttt{r2s5-cv} from \Cref{tab:metric-summary-correlation}) then chose further parameters a factor of 3 larger and smaller.
    
    The reassuring observation from Figure \ref{fig:microbiome-traj-comp-variate} is that the estimates are all fairly similar.
    Indeed, $\sin^2\Theta$ distance between subspaces of dimension 3 can be at most 3; most values plotted are far smaller than this, and even the maximum value of 1.1 corresponds to `about two thirds the same subspace'.
    Moreover, we can obtain a submatrix with maximal value 0.38 by taking sCCA 0.0042, gCCA 0.002 and rCCA 0.33.
        
    This main observation of `closeness' is particularly interesting because it does not hold for all metrics.
    Indeed, plotting top-3 $\sin^2\Theta$ distance in direction space gave far higher distances, see \Cref{app:microbiome-extra-plots}. 
    This illustrates how `variate distance' is a weaker notion in the context of CCA; different forms of regularisation may recover similar canonical variates from rather different sets of directions.
    In practice, considering a variety of different metrics with multiple values of $k$ would help give richer view.

\subsubsection{Overlap Matrices help visualise sharing of signal across directions}\label{sec:overlap-matrices-real-data}
    \begin{figure}[t]\centering
        \centering
        \includegraphics[width=\textwidth]{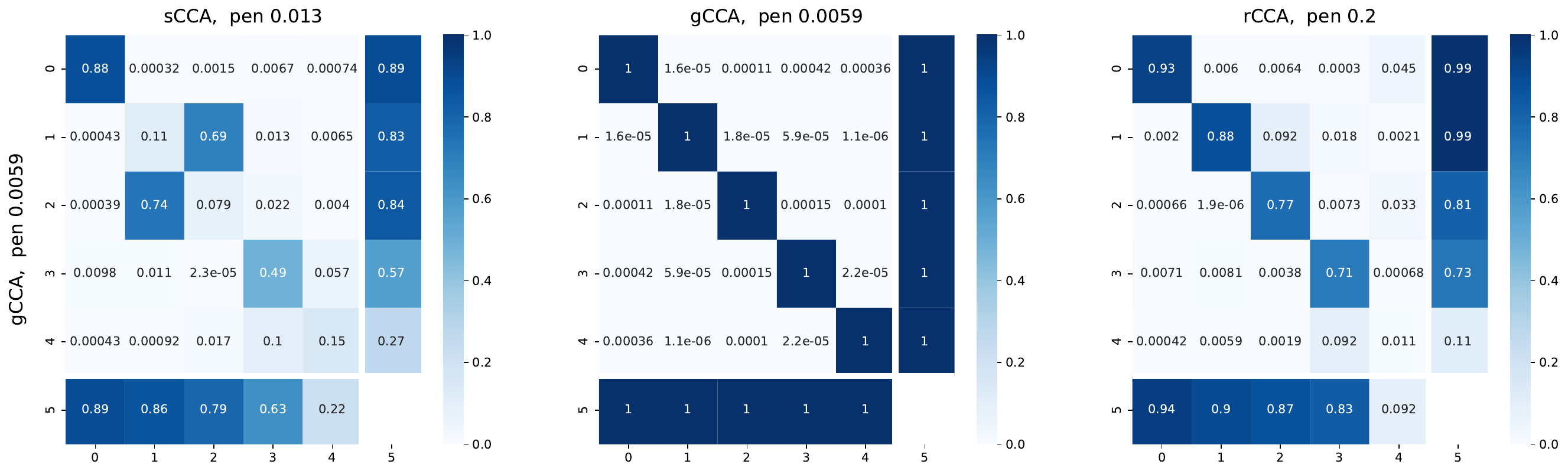}
    \caption{Squared overlap matrices comparing the \suo{}, \glasso{}, \ridge{} estimators with maximal CV correlation (\texttt{r2s5-cv}); see \Cref{sec:overlap-matrices-real-data} for full details.}
    \label{fig:microbiome-sqoverlap-algos}
    \end{figure}
    
    Figure \ref{fig:microbiome-sqoverlap-algos} shows squared overlap matrices (as in \Cref{eq:overlap-matrix-full-sample-default} from \Cref{sec:overlap-introduction}) on the Microbiome dataset; these have additional rows and columns with the corresponding row and column sums.
    These matrices help give a much richer comparison of different estimates, and how their signal overlaps across different directions.

    The squared overlap matrices in \Cref{fig:microbiome-sqoverlap-algos} compare the \suo{}, \glasso{}, \ridge{} estimators with maximal CV correlation (\texttt{r2s5-cv});
    in each case the y-axis corresponds to the \glasso{} estimate, while the x-axes correspond to \suo{}, \glasso{}, \ridge{} respectively.
    The \suo{} and \ridge{} estimates have only been registered to the \glasso{} estimates up to signs.
    In each case the full dataset is used to obtain estimated canonical variates.
    It is therefore unsurprising that the middle plot is almost the identity matrix; this only says that the fitted canonical variates for the CV-optimal graphical lasso estimates are near-orthogonal, as would be expected for small penalty parameters.

    Figure \ref{fig:microbiome-sqoverlap-algos} illustrates why it is important to take a subspace perspective when comparing CCA estimates.
    Indeed, the first few row and column sums are close to 1, despite the matrices not being close to the identity.
    In particular, observe that \suo{} appears to have `mixed-up' the 2nd and 3rd directions relative to \glasso{} and \ridge{}.
    On closer inspection, observe that there are also many small, but significant off-diagonal entries which explain why certain column sums are significantly higher than their corresponding diagonal entries.
    It is particularly interesting to see how stable the estimated top-4 subspace is given there was no significant separation between the 4th and 5th cross validated correlations in \ref{fig:corr_decay_microbiome_gglasso}.
    
    \Cref{fig:microbiome-sqoverlap-algos-orthog} in \Cref{app:microbiome-extra-plots} plots the analogous squared overlap matrices but where the \suo{} and \ridge{} estimates have been registered to the \glasso{} estimates up to orthogonal transformations.
    In this case the \suo{} and \ridge{} overlap matrices also become very close to the identity matrix.
    This again illustrates the importance of taking a subspace perspective, and the benefits of orthogonal registration.
    
    Many other versions of overlap matrices can also give useful insights, and \Cref{app:microbiome-extra-plots} illustrates a few possibilities. For example, one could use a single algorithm and compare along the trajectory of penalty parameter (\Cref{fig:microbiome-sqoverlap-suo-path}), or indeed compare between different CV folds which gives a richer view of the stability of variate subspaces (\Cref{fig:fold_stab_microbiome_gglasso}).

\subsubsection{Biplots}\label{sec:biplots-real-data}
    \begin{figure}[t]\centering
         \includegraphics[width=0.9\textwidth]{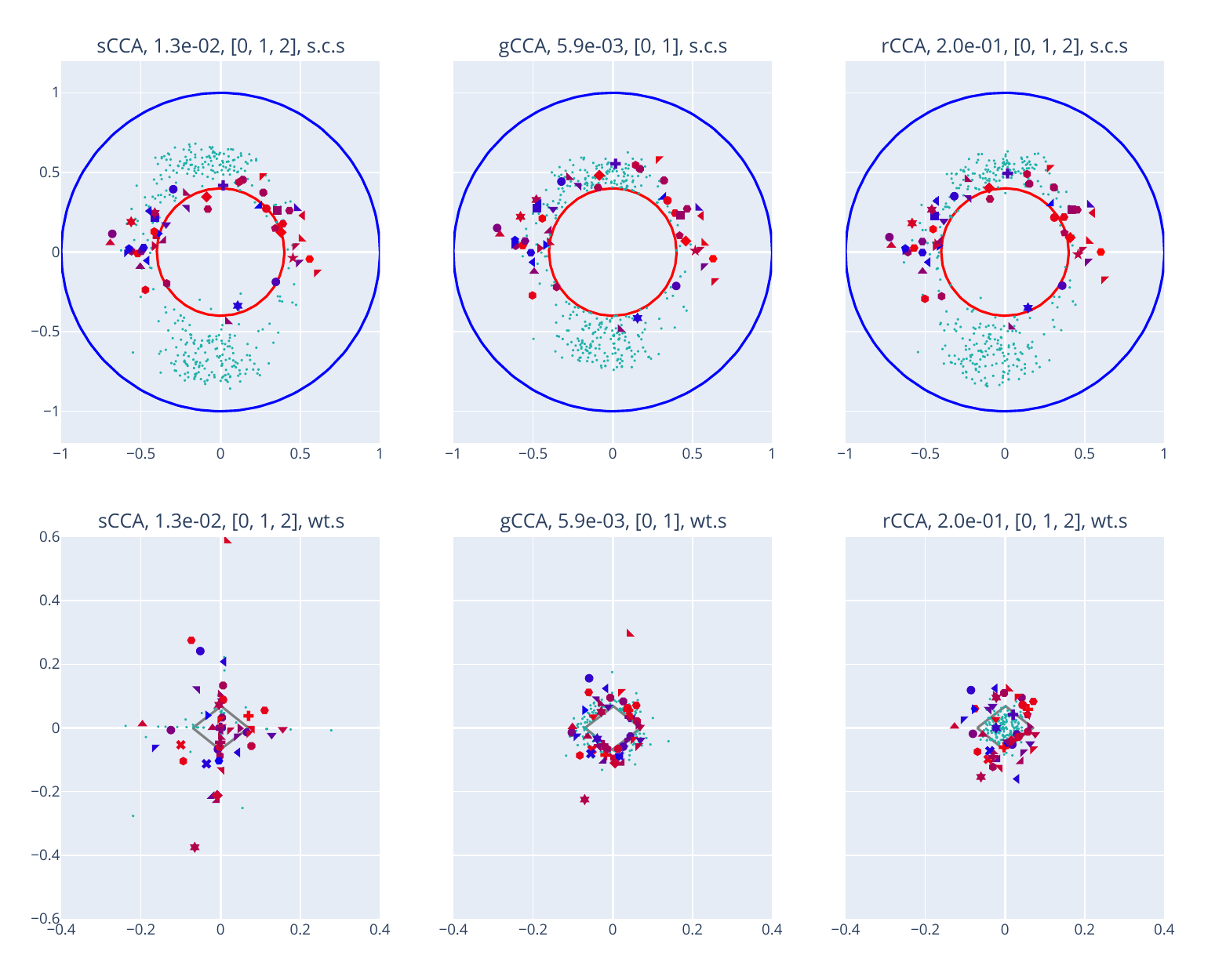}
         \caption{2D biplots for the microbiome dataset, for \suo{}, \glasso{}, \ridge{} methods with respective \texttt{r2s5-cv}-optimal penalty parameters. 
         The \suo{} and \ridge{} estimates have been registered with the central \glasso{} estimate up to orthogonal transformations in variate space. 
         Variates in question are from K0s.
         C0 variables are plotted with large, coloured, distinctive symbols that are consistent between the six plots but otherwise arbitrary.
         K0 variables are plotted with small dark green symbols to avoid cluttering the plot.}
         \label{fig:2d-biplots-microbiome}
    \end{figure}
    
    Having developed an understanding of how different estimates compare, we are now in a good position to inspect certain estimates more closely.
    One particularly useful way to do so is to use biplots, introduced in \Cref{sec:comparison-biplots}; additional registration and thresholding can help with comparison and interpretation of multiple biplots.

    The top row of \Cref{fig:2d-biplots-microbiome} shows such biplots for the first two canonical variates on the Microbiome dataset, using sample variates from the K0 view; the second row plots the corresponding weights.
    The columns correspond to \suo{}, \glasso{}, \ridge{} respectively; the penalty used for each algorithm is that maximising the CV sum of top-5 squared correlations (\texttt{r2s5-cv}).
    The following processing steps were used to help comparison between the plots.
    The \suo{} and \ridge{} estimates were registered with the \glasso{} estimate up to orthogonal transformations (see \Cref{sec:registration}).
    The thresholding was different for biplots and weight plots: for biplots we selected variables with sums of squared correlation greater than 0.4 for the two \glasso{} directions (this same subset of variables is plotted for each estimator); for weights we selected variables whose first two \glasso{} weights had sum of absolute values greater than 0.07 (again same variables plotted for each estimator).
    There are many more K0 variables than C0 variables with significant structure correlations; we plot the K0 variables with smaller symbols to highlight the C0 variables, which better illustrate our main observation.
    
    Firstly, \Cref{fig:2d-biplots-microbiome} illustrates that different methods may estimate similar canonical variates, but with very different weights.
    Indeed, loadings plots are almost identical; and, though there are some similarities in the weight plots (such as the red-isosceles triangle in the top, corresponding to C00163) they look much more different.
    
    It may be of interest to further contrast the weight and loading plots. For example, one could also compare the subset of variables selected for the weight plot with those selected for the loading plot. Here, the subsets are generally similar; but there are also some notable exceptions, such as the variable C00163, which has the largest magnitude weight, but does not attain required threshold for structure correlation plot.

    These biplots give even richer information in three dimensions. 
    We demonstrate this in the following section, with further examples in \Cref{app:microbiome-extra-plots}. 
    It is particularly convenient to explore such plots with modern plotting libraries that allow one to rotate the plots interactively.
    The plots can also be augmented by label information.
    This leads to the striking plots in \Cref{fig:bd-3d-plots}.

\subsubsection{Biological comment}\label{sec:biological-comment}
    \begin{figure}[t]\centering
        \centering
        \includegraphics[width=0.7\textwidth]{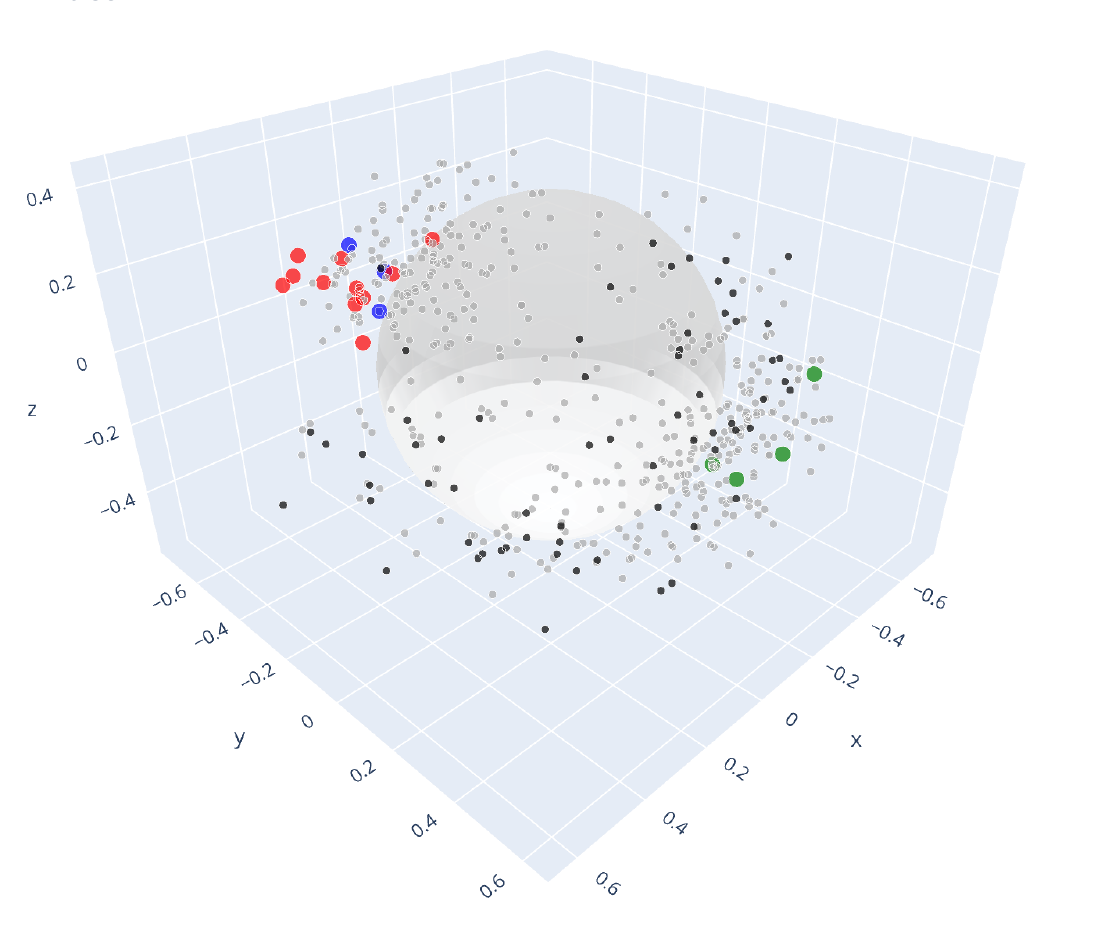}
    \caption{3D biplot for Microbiome dataset, using CV-optimal gCCA, with different pathways colored, see \Cref{sec:biological-comment} for full details}
    \label{fig:mb-kumar-labels}
    \end{figure}
    
    Figure \ref{fig:mb-kumar-labels} shows a 3D biplot for Microbiome dataset using the \texttt{r2s5-cv}-optimal penalty parameter. 
    Our Microbiome dataset measures metabolites and different gene pathways for healthy and diseased patients with IBD. In the 3D biplot we examined the pathways enriched on either pole of the plot to hypothesize any latent variables that may explain the opposing pathways. On one pole we observe pathways and metabolites that are linked to lipid metabolism. Lipid metabolism is a crucial component of bacterial growth and survival, and includes pathways such as: Fatty Acid Synthesis, Lipid A Synthesis, Phospholipid Synthesis and Lipid Desaturation \citep{ghazalpour_expanding_2016}. On the other pole we observed many pathways/metabolites linked to amino acid metabolism and several genes that corresponded to ABC Transporters. Amino acid metabolism in microbes generally regards either the building or breakdown of certain amino acids which are the building blocks of proteins, the functional unit of cells. ABC (ATP Binding Cassette) Transporters use ATP to transport various chemical substrates across a cell membrane \citep{chandel_amino_2021, hollenstein_structure_2007}. We have highlighted three biological processes with large coloured markers: Amino Acid Metabolism: blue, ABC transporters: red, Lipid Metabolism: green.

    During IBD disease progression a multitude of shifts can change in the microbial flora of the gut \citep{kostic_microbiome_2014}. A flare up in IBD is noted to include an influx of host immune cells to aid and defend different tissue regions within the gut \citep{neurath_targeting_2019, zhang_cellular_2006}.
    These changes in IBD can potentially describe the driven differences between amino acid metabolism/ABC transporters and lipid metabolism. On a more molecular level, changes in nutrient availability for the microbes due to diet or inflammation may lead to an availability of more amino acids and their corresponding transporters and a down regulation of certain lipid sequestration. On the other hand, if the environment has many lipids due to a breakdown of cell membranes, many of those metabolic pathways may be upregulated in the remaining taxonomy. Additional community pressures due to diet changes or environmental changes may give rise to the growth of certain microbes and that may include pathways of either amino acid metabolites or lipid intermediates. With the aid of CCA, we are able to explore groups of genes and metabolites that are maximally correlated, shedding light on how the genetic catalogue of the perturbed IBD microbiome might affect metabolic pathways.

\subsection{Contrast with datasets from previous CCA studies}\label{sec:contrast-with-nutrimouse-breastdata}
    We compare with two datasets previously used to evaluate CCA methods:
    \begin{itemize}
        \item Nutrimouse: a dataset previously proposed for demonstration of various multiview learning techniques \citep{rodosthenous_integrating_2020}, available in the CCA R package \citep{gonzalez_cca_2021}, and originally from \citet{martin_novel_2007}. This contains the expression measure of 120 genes potentially involved in nutritional problems and the concentrations of 21 hepatic fatty acids (lipids) for 40 mice. 
        \item BreastData: proposed in \citet{witten_penalized_2009} for demonstration of \wit{} with data available in the PMA R package \citep{tibshirani_pma_2020} and originally from \citet{chin_genomic_2006}. This contains gene expression (RNA, $q=19672$ genes) and DNA copy number measurements ($p=2149$ CGH spots --- comparative genomic hybridization) on a set of 89 samples. 
    \end{itemize}
    
    Both of these datasets have far fewer samples than the Microbiome dataset. 
    The Nutrimouse dataset has smaller dimension, with ratios $(p/n,q/n) = (0.525,3)$.
    By contrast, the Breastdata dataset has much larger dimension, and much larger ratios $(p/n,q/n) = (24.1, 221)$.
    We summarise the datasets in Table \ref{tab:datasets} for reference.

    \begin{table}[t]
        \caption{Summary of different real datasets analysed}\label{tab:datasets}
        \centering 
        \vspace{5pt}
        \begin{tabular}{lrrrcc}
            \toprule
            \textbf{Dataset} & \textbf{n}  & \textbf{p}  & \textbf{q} & \textbf{X-variables} & \textbf{Y-variables} \\
            \midrule
                Microbiome & 458 & 711 & 142 & Metabolites (C0s) & Proteins (K0s) \\
                Nutrimouse & 40 & 21 & 120 & Lipids & Gene expression \\
                BreastData & 89 & 2149 & 19672 & DNA copy number & RNA gene expression \\ \bottomrule
        \end{tabular}
    \end{table}
    
    We now give a summary of conclusions from applying the techniques of \Cref{sec:pipeline-toolbox} to these datasets.
    We leave all justification to \Cref{sec:nutrimouse-extra-plots,sec:breastdata-extra-plots}, where we state observations and supporting results more systematically.   
    
    Despite the smaller sample size, most of our previous observations remain valid.
    Indeed, many successive estimated canonical pairs have significant cross validation signal; variates are significantly more stable than weights; and the three genuine CCA methods captured similar correlation signal. Again sPLS captured less signal than the CCA methods, and was relatively unstable between folds, but this was somewhat less significant than for the Microbiome set.

    However there are some also subtle qualitative differences between the Nutrimouse and the Microbiome cases.
    Firstly, Figure \ref{fig:correlation-stability-trajectories-nutrimouse} shows that very small penalty parameters are almost optimal in terms of correlation captured for each of \ridge{}, \suo{}, \glasso{}.
    In particular, for \glasso, taking a very small penalty of $10^{-4}$ seems to give near-optimal correlation captured while also having near-optimal stability.
    Secondly, the relevant trajectory comparison matrices show that the ridge estimator is almost constant for small parameter values, and only starts to change for parameters approaching 1. This behaviour also manifests in all of the other diagnostic plots --- correlations along trajectories are very flat, correlation decay has persistent colour structure, overlap matrices are near identity, and biplots vary only a little with penalty parameter.
    Thirdly, \glasso{} also moves only a little along trajectories - and is always very close to the ridge solution; again this manifests in each different plot.

    There are more pronounced qualitative differences with the Breastdata dataset.
    The main difference is that the estimated variates are no longer stable between algorithms; indeed, \suo{} and \ridge{} now give very different estimates, even after registration.
    This can be seen at a high level with the relevant trajectory comparison matrices and at a more detailed level with the relevant overlap matrices.
    In addition, though \suo{} is now very sensitive to penalty parameter (overlap matrices show that only the top direction is stable), it is less sensitive to fold at the optimal penalty (the top-3 subspace is stable).
    By contrast, \ridge{} is now almost invariant to its penalty parameter, while still being fairly stable with respect to fold; this is remarkable and suggests some mathematical phenomenon of potential interest.

    To summarise these results, it is helpful to think about signal to noise ratios (SNR).
    In high SNR regimes, such as with Microbiome and Nutrimouse dataset, we may expect to consistently estimate stable canonical variates, even if the directions are unstable.
    But for lower SNR regimes we may expect to find very different canonical variate estimates --- but that these may still contain significant CV correlations signal --- so give different insights about the data.
    In these regimes it may be particularly useful to combine insights from a variety of different regularised CCA methods, which will pick up different sorts of signal in the data.

\section{Discussion}\label{sec:discussion}
    Many regularised CCA methods are now practical for high-dimensional data due to recent computational and theoretical advances.
    However, it is not clear that the assumptions behind the regularisation are appropriate; moreover, it is not accepted how to choose between different methods or indeed interpret the output.
    Our first contribution was to propose a very different regularised CCA method, gCCA, based on the Graphical Lasso, and to show it has reassuring theoretically properties.
    Our second contribution was to investigate different criteria for evaluating CCA methods, for comparing between different CCA estimates, and indeed for interpreting the estimates.
    We observed that it is much easier to estimate canonical variates than canonical directions, and to consider successive subspaces rather than successive canonical pairs --- which was consistent with our theoretical results and intuition.
    This led to a concrete framework for applying and interpreting different CCA methods to real data; this framework culminates in the powerful visualisation tool of biplots.
    Within this framework, our gCCA estimator often outperformed existing methods; moreover, by using it alongside existing methods, we were able to get a much broader perspective on high-dimensional CCA.
    We recommend future practitioners use a variety of methods, and hope this leads to a wider use of CCA within the statistics community.

    Our work has opened up many directions for future research.
    To start with, there are natural extensions of our gCCA method that might be worth exploring; the plug-in graphical lasso approach can also be applied with any multiview CCA framework \citep{tenenhaus_regularized_2011}; there are also extensions of the graphical lasso to deal with (unobserved) latent variables \citep{chandrasekaran_latent_2012}, which have looser theoretical assumptions, are correspondingly more flexible, and may perform better empirically. Additionally, there has been a variety of recent work on improving the computational cost of the graphical lasso, with some methods already implemented in \texttt{gglasso} \citep{schaipp_gglasso_2021}, and which may help scale up the method to higher-dimensional situations.

    There are also open theoretical questions surrounding CCA more generally.
    To start with, there is little theoretical understanding of ridge regularised CCA (\ridge{}); this is a simple, effective and popular high-dimensional CCA method and therefore is important to understand.
    Firstly, we observed that even in very high-dimensional regimes the ridge methods were reasonably stable and recovered consistent signal.
    Secondly, in the very high-dimensional regimes, ridge CCA was almost unchanged by varying the penalty parameter.
    Given recent advances in random matrix theory applied to ridge estimators for linear regression, we suspect both of these questions may be tractable.
    However, we also observed that the sparse CCA methods also frequently recovered top variate subspaces for CCA, even when the true directions were not sparse.
    Our gCCA method was also remarkably effective at recovering variate subspaces even in the mispecified case.
    More generally there appears to be room for theory bridging the gap from \citet{gao_sparse_2016,ma_subspace_2020} to our observations of robustness to mis-specification.

    However, the most impactful area for research would be on the practical side.
    To start it would be helpful to integrate the evaluation criteria and comparison tools we have proposed into existing open source implementations of CCA methods, such as \ccazoo{} \citep{chapman_cca-zoo_2021}.
    It would also be helpful to have a public repository of multiview datasets, that practitioners could use to develop intuition for CCA, or indeed related multiview learning algorithms.
    There is also scope for a broader systematic numerical comparison between CCA algorithms than we have performed here, using a wider array of synthetic and real datasets, and investigating how the various CCA criteria interact with downstream tasks. Of particular interest would be further understanding how CCA integrates with non-linear dimension reduction techniques in a broader EDA pipeline.
    There is also scope to develop more flexible regularised CCA techniques that can incorporate more flexible biological priors,
    and to investigate more creative multiple applications of CCA, perhaps by comparing CCA on different subgroups of data.
    We hope the framework and analysis we have provided will encourage further research into these interesting directions.

%% file: suppl.tex
\newpage
\section{Matrix Analysis background}\label{app:matrix-analysis-background}
\subsection{Warm-up: uniqueness of eigendecomposition}
\begin{lemma}\label{lem:same-subspace-orthogonal}
    Let $V,V'$ span the same subspace of $\R^d$, each with orthonormal columns. Then $V' = V O$ for some orthogonal matrix $O$.
\end{lemma}
\begin{proof}
    Note that $V V^\top$ is an orthogonal projection onto the shared subspace, so defining $O\defeq V^\top V'$ we do indeed have
    \begin{align*}
        V' = V V^\top V' = V O \; .
    \end{align*}
    Moreover, again because $V V^\top$ is an orthogonal projection
    \begin{align*}
        O^\top O = V'^\top V V^\top V' = V'^\top V' = I \; ,
    \end{align*}
    so $O$ is indeed orthogonal.
\end{proof}

\begin{lemma}\label{lem:eigenvector-orthogonality}
    Let $A$ be a symmetric matrix.
    Let $v,v'$ be eigenvectors of $A$ with eigenvalues $\lambda,\lambda'$ respectively.
    If $\lambda' \neq \lambda$ then $\langle v, v' \rangle=0$.
\end{lemma}
\begin{proof}
    Simply compute
    \begin{align*}
        0 = v^\top A v' - v'^\top A v = \lambda' v^\top v' - \lambda v'^\top v = (\lambda' - \lambda) \langle v, v' \rangle
    \end{align*}
    so dividing through by $\lambda' - \lambda$ gives the result.
\end{proof}

\begin{proposition}\label{prop:uniqueness-of-eigendecomp}
    Let $A \in \R^{d \times d}$ have eigenvalues $\lambda_1 \geq \dots \geq \lambda_d$.
    Suppose we have two matrices $V,V'$ such that
    \begin{align*}
        A = V \Lambda V^\top = V' \Lambda V'^\top \; .
    \end{align*}
    Suppose that an eigenvalue $\lambda$ occurs with multiplicity $\delta$, and let $V_\lambda, V_\lambda'$ denote the corresponding columns of $V,V'$ (i.e. for some $k$ we have $\lambda_{k-1} > \lambda = \lambda_{k} = \dots = \lambda_{k+\delta-1} > \lambda_{k+\delta}$ and $V_\lambda = \begin{pmatrix}V_{k} & \dots & V_{k+\delta - 1}\end{pmatrix}$).
    Then there is some orthogonal matrix $O_\lambda \in \R^{\delta \times \delta}$ such that
    \begin{align}
        V_\lambda' = V_\lambda O_\lambda \; .
    \end{align}
\end{proposition}
\begin{proof}
    Consider $v'_j$ for $j \in \{k,\dots,k+\delta-1\}$; then $v'_j$ is a $\lambda$-eigenvector of $A$.
    Now consider any $l \notin \{k,\dots,k+\delta-1\}$; $v_l$ is a $\lambda_l$-eigenvector of $A$. 
    Therefore by \Cref{lem:eigenvector-orthogonality} we have $\langle v_l, v'_j \rangle = 0$.

    But since $V$ spans $\R^d$ we must therefore have $v'_j \in \spann{V_\lambda}$.
    Hence $\spann{V'_\lambda} \subset \spann{V_\lambda}$.
    By symmetry we deduce that $V_\lambda,V'_\lambda$ span the same $\delta$-dimensional subspace.
    We conclude by \Cref{lem:same-subspace-orthogonal}. 
\end{proof}

\subsection{SVD}\label{app:SVD}
    The SVD is a widely used tool in both pure and applied mathematics, which has many different motivations and interpretations. To make our arguments concrete, we summarise some relevant properties of the SVD below; these are all well known, but rarely presented precisely as we desire them, so we give proofs. Our presentation closely follows that of \citet{stewart_matrix_1990}. Indeed, the statement and proof of the first result below (existence of SVD) are directly copied from \citet{stewart_matrix_1990}, but repeated here to help with the later, less standard, results.

    \subsubsection{Definition, existence and uniqueness}
    \begin{theorem}[Characterisation via $A^\top A$]\label{thm:svd-eigenvalue-problem-charac}
        Let $A\in \R^{p\times q}$ have rank $K$.
        Let $U \in \R^{p\times p}$ and $V \in \R^{q \times q}$ be orthogonal matrices partitioned as $U = \begin{pmatrix} U_{:K} & U_{K:} \end{pmatrix},\: V =\begin{pmatrix} V_{:K} & V_{K:} \end{pmatrix}$ with $U_{:K} \in \R^{p \times K}, V_{:K} \in \R^{q \times K}$.
        Let $\sigma_1 \geq \dots \geq \sigma_K > 0$ and define $\Sigma_+ = \diag(\sigma_1,\dots,\sigma_K) \in \R^{K\times K}$, $\Sigma_{p \times q} = \diag(\sigma_1, \dots, \sigma_K, 0, \dots, 0) \in \R^{p \times q}$ and $\Sigma_{q \times q} = \diag(\sigma_1, \dots, \sigma_K, 0, \dots, 0) \in \R^{q \times q}$.        
        
        Then the following are equivalent.
        \begin{enumerate}[label=(\roman*)]
            \item \label{blt:eigendecomp} $V$ defines an eigen-basis for $A^\top A$ and $V_{:K}$ determines $U_{:K}$ by:
                \begin{equation}\label{eq:svd-eval-ATA-def}
                    U_{:K} = A V_{:K} \Sigma_+^{-1} \quad \text{ and } \quad V^\top A^\top A V = \begin{pmatrix}\Sigma_+^2 & 0 \\ 0 & 0 \end{pmatrix} =\Sigma_{q \times q}^2
                \end{equation}
            \item \label{blt:SVD} We have the singular value decomposition
            \begin{equation}\label{eq:svd-characterisation-def-app}
                U^\top A V = \begin{pmatrix}\Sigma_+ & 0 \\ 0 & 0 \end{pmatrix} \eqdef \Sigma_{p \times q}
            \end{equation}
            
        \end{enumerate}
    \end{theorem}

    \begin{proof}
        \ref{blt:eigendecomp} $\implies$ \ref{blt:SVD}:
        
        By considering sub-matrices of \Cref{eq:svd-eval-ATA-def} we have
        \begin{align}\label{eq:parition-V-relations}
            V_{:K}^{\top} A^{\top} A V_{:K}=\Sigma_{+}^2, \quad V_{K:}^{\top} A^{\top} A V_{K:}=0 \; .
        \end{align}
        The second of these equations implies that in fact $A V_{K:}=0$.
        
        Since $U$ is orthogonal we have $U_{:K}^{\top} U_{:K}=I$.
        Substituting this into the first equation of \Cref{eq:parition-V-relations} gives $I = U_{:K}^\top U_{:K} = U_{:K}^\top A V_{:K} \Sigma_+^{-1}$ and therefore that $U_{:K}^\top A V_{:K} = \Sigma_+$.
        Combining this with the fact that $A V_{K:}=0$ then gives
        \begin{align*}
            U^{\top} A V=\left(\begin{array}{cc}
            U_{:K}^{\top} A V_{:K} & U_{:K}^{\top} A V_{K:} \\
            U_{K:}^{\top} A V_{:K} & U_{K:}^{\top} A V_{K:}
            \end{array}\right)=\left(\begin{array}{cc}
            \Sigma_{+} & 0 \\
            0 & 0
            \end{array}\right)
        \end{align*}
        as required.

        \ref{blt:SVD} $\implies$ \ref{blt:eigendecomp}:
        For the first equation in \Cref{eq:svd-eval-ATA-def}, left multiply \Cref{eq:svd-characterisation-def-app} by $U$ and right-multiply by $\begin{pmatrix} \Sigma_+^{-1} & 0 \\ 0 & 0\end{pmatrix}\in \R^{q \times q}$
        to get
        \begin{align*}
            A V \begin{pmatrix} \Sigma_+^{-1} & 0 \\ 0 & 0\end{pmatrix} = U \begin{pmatrix} I_{k\times k} & 0 \\ 0 & 0\end{pmatrix} = \begin{pmatrix} U_{:K} & 0 \end{pmatrix}
        \end{align*}
        then equate the first $K$ columns.

        For the second equation in \Cref{eq:svd-eval-ATA-def}, use $\Cref{eq:svd-characterisation-def-app}$ and orthogonality of $U$ to compute
        \begin{align*}
            V^\top A^\top A V = (U^\top A V)^\top (U^\top A V) = \Sigma_{p \times q}^\top \Sigma_{p \times q} =  \Sigma_{q \times q}^2
        \end{align*}
        as required.
    \end{proof}

    \begin{corollary}[Existence and Uniqueness of SVD]\label{cor:svd-existence-uniqueness-app}
        Let $A\in \R^{p\times q}$ have rank $K$. 
        Then there exist orthogonal matrices $U \in \R^{p\times p}$ and $V \in \R^{q \times q}$ such that
        \begin{equation}\label{eq:svd-existence-def-app}
            A  = U \begin{pmatrix}\Sigma_+ & 0 \\ 0 & 0 \end{pmatrix} V^\top \eqdef U \Sigma V^\top
        \end{equation}
        where $\Sigma_+ = \diag(\sigma_1,\dots,\sigma_K)$ with $\sigma_1 \geq \dots \geq \sigma_K > 0$. 
        Here, $\sigma_1, \dots, \sigma_K$ are the square roots of the non-zero eigenvalues of $A^\top A$ (or equivalently $A A^\top$) and so uniquely determined.

        Moreover, the matrices $U,V$ are unique up to orthogonal transformations within the spaces of repeated singular values:
        indeed, suppose we have another pair of orthogonal matrices $U' \in \R^{p\times p}$ and $V' \in \R^{q \times q}$ with $U'^\top A V' = \Sigma$.
        Suppose that a non-zero singular value $\sigma > 0$ occurs with multiplicity $\delta$ and let $V_\sigma, V_\sigma', U_\sigma, U_\sigma'$ denote the corresponding columns of $V,V', U, U'$ respectively (i.e. for some $k$ we have $\sigma_{k-1} > \sigma = \sigma_{k} = \dots = \sigma_{k+\delta-1} > \sigma_{k+\delta}$ and $V_\sigma = \begin{pmatrix} V_k \dots V_{k+\delta-1}\end{pmatrix}$ etc).
        Then there is some orthogonal matrix $O_\sigma \in \R^{\delta \times \delta}$ such that
        \begin{align}
            V_\sigma' = V_\sigma O_\sigma \quad\text{and}\quad U_\sigma' = U_\sigma O_\sigma \; .
        \end{align}
        By contrast, for the zero singular value, we can only say that $\cspan{U_{K:}} = \cspan{U'_{K:}} = \ker(A A^\top)$, and $\cspan{V_{K:}} = \cspan{V'_{K:}} = \ker(A^\top A)$ (i.e. we have $V_{K:} = V'_{K:} O_0^{(V)}, U_{K:} = U'_{K:} O_0^{(U)}$ but do not necessarily have $O_0^{(V)} = O_0^{(U)}$).
    \end{corollary}
    \begin{proof}
        For both existence and uniqueness we use the characterisation from \Cref{thm:svd-eigenvalue-problem-charac}. 
        For existence, let $V = \begin{pmatrix}V_1 & V_2 \end{pmatrix} \in \R^{q \times q}$ be a matrix of successive eigenvectors for $A^\top A$. Let the corresponding eigenvalues be $\sigma_1^2 \geq \dots \geq \sigma_K^2 > 0 = \sigma_{K+1}^2 = \dots = \sigma_q^2$, and $\sigma_k = (\sigma_k^2)^\half \geq 0$ be their square roots. Then define diagonal matrices of successive $\sigma_k$ as in the setup of \Cref{thm:svd-eigenvalue-problem-charac}.
        Then by definition, the second equation of \Cref{eq:svd-eval-ATA-def} holds.
        
        Further define $U_{:K} = A V_{:K}$ so that the second equation of \Cref{eq:svd-eval-ATA-def} holds; and finally define $U_{K:} \in \R^{(p-K) \times p}$ be a matrix whose columns form an orthonormal basis for $\cspan{U_{:K}}^{\perp}$. Then $U = \begin{pmatrix}U_{:K} & U_{K:}\end{pmatrix}$ formed by concatenation is an orthogonal matrix, so we can conclude \Cref{blt:SVD} by the equivalence of \Cref{thm:svd-eigenvalue-problem-charac}.

        For uniqueness, whether $\sigma$ is non-zero or zero, we have that $V_{\sigma}, V'_{\sigma}$ are matrices whose columns span the $\sigma^2$-eigenspace of $A^\top A$, so by \Cref{prop:uniqueness-of-eigendecomp}, there is an orthogonal matrix $O_\sigma$ such that $V'_\sigma = V_\sigma O_\sigma$.
        If $\sigma > 0$ then by the first equation of \Cref{eq:svd-eval-ATA-def} we see that this same orthogonal matrix gives $U'_\sigma = U_\sigma O_\sigma$.
        By contrast, if $\sigma = 0$, we can apply \Cref{thm:svd-eigenvalue-problem-charac} to $A^\top$ with $U, V$ reversed to conclude that $U_0, U'_0$ both span the zero-eigenspace of $A A^\top$ and therefore by \Cref{lem:same-subspace-orthogonal} that there is an orthogonal matrix $O_0^{(U)}$ such that $U'_0 = U_0 O_0^{(U)}$ but can no longer conclude that this is the same orthogonal matrix as transforms $V_0$ to $V_0'$.
    \end{proof}

    \begin{definition}[SVD]
        We call a factorisation of the form \Cref{eq:svd-existence-def-app} from \Cref{cor:svd-existence-uniqueness-app} a (full) \textit{singular value decomposition} of the matrix $A$.
        The $\sigma_k$ are called the \textit{singular values} of $A$.
        The successive columns of $U, V$ are called left and right singular vectors of $A$ respectively; those corresponding to non-zero singular values are naturally paired by the uniqueness statement of \Cref{cor:svd-existence-uniqueness-app}, and so we may refer to these as pairs of singular vectors, or simply \textit{singular pairs}.

        Note that we can write a more compact form of \Cref{eq:svd-existence-def-app} by cutting out the zero components to give
        \begin{align}
            A &= U_{:K} \Sigma_+ V_{:K}^\top \label{eq:SVD-compact}
        \end{align}
        and call this a \textit{compact} SVD of $A$.        
    \end{definition}

    \begin{remark}
        The full SVD is helpful when we need full bases for the spaces $\R^p,\R^q$ respectively, but it can be more convenient to consider the compact SVD, which, for example inherits a cleaner notion of uniqueness from \Cref{cor:svd-existence-uniqueness-app}.
        We will see that one can make a similarly helpful distinction between full and compact CCA.
    \end{remark}  
    
    \subsubsection{Sequential variational characterisation}
    
    \begin{theorem}[Variational Characterisation]\label{thm:svd-variational-app}
        Consider a sequence of vectors $(u_k,v_k)_{k=1}^K$ corresponding to columns of the matrices $U=(u_1,\dots,u_K),V=(v_1,\dots,v_K)$. Then $U \Sigma V^\top$ is a compact SVD for $A$ if and only if $(u_k,v_k)_{k=1}^K$ is a sequence of successive solutions to the programs
        \begin{equation}\label{eq:svd-variational-charac-app}
            \begin{split}
            \underset{u \in \R^p,v\in\R^q}{\operatorname{maximize}}\: & u^\top A v \\
            \textrm{\normalfont subject to}\:& \|u\|_2 \leq 1,\: \|v\|_2 \leq 1,\\ 
            & u^\top u_j = v^\top v_j =0 \text{ for } 1 \leq j \leq k-1 .	
            \end{split}
        \end{equation}
        of maximal length such that each program has strictly positive objective value.
    \end{theorem}
    
    \begin{proof}
        Note first that in each case there is some pair attaining the maximum because we are optimising a continuous objective over a compact set.
        We will prove inductively that sequences of maximisers to these programs are precisely sequences of singular pairs. The $k=1$ case can be seen as a special case of general induction step but we prove it first as a warm-up.
    
        For $k=1$, a maximiser of \Cref{eq:svd-variational-charac-app}, must correspond to a maximiser of the Lagrangian
        \begin{align*}
            \mathcal{L}(u,v; \lambda_1,\mu_1) = u^\top A v + \tfrac{1}{2} \lambda (1 -\norm{u}^2_2) + \tfrac{1}{2} \mu(1 - \norm{v}^2_2) \; .
        \end{align*}
    
        This leads to the first order conditions:
        \begin{align}
            0 = \partial_u \mathcal{L} &= A v - \lambda u \label{eq:first-order-conditions1}\\
            0 = \partial_v \mathcal{L} &= A^\top u - \mu v \label{eq:first-order-conditions2}
        \end{align}
        for feasible $u,v$ this gives
        \begin{align*}
            \lambda = \lambda \norm{u}^2 = u^\top A v = v^\top A^\top u = \mu \norm{v}^2 = \mu
        \end{align*}
        so in fact $\lambda = \mu$ and combining the equations we obtain
        \begin{align*}
            A^\top A v = A^\top \lambda u = \lambda^2 v
        \end{align*}
        whence $v$ is an eigenvector of $A^\top A$.
        Moreover, the corresponding objective value is $u^\top A v = \lambda$.
        So any maximising $u,v$ must be precisely some top singular pair.

        For the general induction step, suppose that $(u_1,v_1),\dots,(u_{k-1},v_{k-1})$ are some top-$(k-1)$ singular pairs.
        The added orthogonality constrains give the larger Lagrangian
        \begin{align*}
            \mathcal{L}(u,v; \lambda_1,\mu_1) = u^\top A v + \tfrac{1}{2} \lambda (1 -\norm{u}^2_2) + \tfrac{1}{2} \mu(1 - \norm{v}^2_2) + \sum_j^{k-1} \lambda_j u_j^\top u + \sum_j^{k-1}  \mu_j v_j^\top v
       \;. \end{align*}
    
        This now gives first order conditions
        \begin{align*}
            0 = \partial_u \mathcal{L} &= A v - \lambda u + \sum_j^{k-1} \lambda_j u_j \\
            0 = \partial_v \mathcal{L} &= A^\top u - \mu v + \sum_j^{k-1} \mu_j v_j.        \end{align*}
        Taking the inner product with respect to some $u_j$ and using feasibility shows
        \begin{align*}
            0 = u_j^\top A v - \lambda u_j^\top u + \lambda_j = \sigma_j v_j^\top v - 0 + \lambda_j = \lambda_j \; .
        \end{align*}
        Similarly $\mu_j=0$.
        Therefore we recover the previous first order conditions \Cref{eq:first-order-conditions1,eq:first-order-conditions2}.
        So indeed any maximiser must be a singular pair.
        Therefore any maximiser must be some choice of top remaining singular pair, as required.
    \end{proof}

    \subsubsection{SVD of a bi-linear form}
    Conveniently, the definition of SVD above can be extended to finite-dimensional inner product spaces\footnote{Or indeed separable Hilbert spaces, see \citet{carlsson_von_2021}} in a basis independent way.

    \begin{definition}[SVD of a bi-linear form]
        Let $\mathcal{X}, \mathcal{Y}$ be finite-dimensional real vector spaces of dimensions $p,q$ respectively.
        Let $\mathcal{A}: \mathcal{X}, \mathcal{Y} \mapsto \R$ be a bi-linear form.
        Then we say that the pair of orthonormal bases $(z_k)_{k=1}^p, (w_k)_{k=1}^q$ give a singular value decomposition of the operator $\mathcal{A}$ if there exist constants $\sigma_1 \geq \dots \geq \sigma_{\min(p,q)} \geq 0 = \sigma_{\min(p,q) + 1} = \dots = \sigma_{\max(p,q)}$ such that
        \begin{align*}
            \mathcal{A}(z_k, w_l) = \delta_{kl} \sigma_k \quad \forall \: k \in [p], l \in [q] \; .
        \end{align*}       
    \end{definition}

    \begin{remark}[Existence, uniqueness, variational characterisation]\label{rem:operator-SVD-from-matrix-SVD}
    \newcommand{\rep}{\operatorname{repr}}
        Suppose we have arbitrary orthonormal bases $(x_i)_{i=1}^p$ and $(y_i)_{i=1}^q$ for $\mathcal{X}, \mathcal{Y}$ respectively.
        Write $\rep_{(x_k)_k}: \mathcal{X} \mapsto \R^p$, $\rep_{(y_k)_k}: \mathcal{Y} \to \R^q$ to be linear maps giving the representation of vectors in the given bases, and let $A \in \R^{p \times q}$ be the matrix representing $\mathcal{A}$ in the given bases. 
        Then for any vectors $z,z' \in \mathcal{X}, w,w' \in \mathcal{Y}$ we have
        \begin{align*}
            \mathcal{A}(z, w) &= \rep_{(x_k)_k}(z)^\top \: A \: \rep_{(y_k)_k}(w), \\
            \langle z, z' \rangle_{\mathcal{X}} &=\rep_{(x_k)_k}(z)^\top \rep_{(x_k)_k}(z') , \\
            \langle w, w' \rangle_{\mathcal{Y}} &= \rep_{(x_k)_k}(w)^\top \rep_{(x_k)_k}(w') \; ,
        \end{align*}
        by definition.
        Therefore, the sets of orthonormal vectors $(z_k)_{k=1}^p, (w_k)_{k=1}^q$ give a singular value decomposition for $\mathcal{A}$ precisely when their vector representations give a singular value decomposition for the matrix representation of $\mathcal{A}$.
        We can therefore deduce notions of existence, uniqueness and a sequential variational characterisation corresponding to those from the previous two subsections.
    \end{remark}

\subsection{Canonical Bases}\label{sec:canonical-bases}
    \begin{definition}[Canonical bases]\label{def:canonical-bases}
        Let $\mathcal{X}, \mathcal{Y}$ be finite-dimensional subspaces of some inner product space $\mathcal{H}$ of dimensions $p, q$ respectively.
        Consider the bi-linear form defined by the inner product
        \begin{align*}
            \mathcal{C}: \mathcal{X} \times \mathcal{Y} \mapsto \R, \: (x,y) \mapsto \langle x, y \rangle \; .
        \end{align*}
        We call a choice of pairs of singular vectors $(z_k)_{k=1}^p, (w_k)_{k=1}^q$ for $\mathcal{C}$ a choice of paired canonical bases for the pair of subspaces $\mathcal{X}, \mathcal{Y}$.
    \end{definition}

    Firstly note that the canonical bases therefore inherit a notion of uniqueness from the corresponding SVD.
    Moreover, by Cauchy--Schwarz, for any unit vectors $z, w \in \mathcal{H}$ we have $\mathcal{C}(z,w) = \langle z, w \rangle \leq \norm{z}_\mathcal{H} \norm{w}_\mathcal{H}$ and therefore the singular values of $\mathcal{C}$ are bounded by 1.
    This means that they can naturally be interpreted as cosines of angles.    

    \begin{definition}[Canonical angles, $\cos^2\Theta$ similarity, $\sin^2\Theta$-discrepancy]
        We define canonical angles between the subspaces $\mathcal{X}, \mathcal{Y}$ by $\theta_k = \arccos{\sigma_k(\mathcal{C})}$ for $k = 1, \dots, \min{(p,q)}$. Define $\cos^2\Theta$ similarity
        \begin{equation}\label{eq:cos2theta-sim-def}
            \quad \cos^2\Theta(\mathcal{X},\mathcal{Y}) = \sum_{k=1}^{\min{p,q}} \cos^2\theta_k = \norm{\mathcal{C}}_{HS}^2 \; .
        \end{equation}
        If $\mathcal{X},\mathcal{Y} \subset \mathcal{H}$ each have dimension $K$, then we write $\Theta = \diag(\theta_1, \dots, \theta_K)$ for a diagonal matrix of canonical angles, and $\sin\Theta, \cos\Theta$ for corresponding diagonal matrices with $\sin, \cos$ applied element-wise respectively.
        
        We also use the following shorthand for the $\sin^2\Theta$ distance: $\sin^2\Theta(\mathcal{X},\mathcal{Y}) = \norm{\sin\Theta(\mathcal{X},\mathcal{Y})}_F^2$.
        We correspondingly recover $\cos^2\Theta(\mathcal{X},\mathcal{Y}) = \norm{\cos\Theta(\mathcal{X},\mathcal{Y})}_F^2$ and the pleasing identity
        \begin{align*}
            \cos^2\Theta(\mathcal{X},\mathcal{Y}) + \sin^2\Theta(\mathcal{X},\mathcal{Y}) = K \;.
        \end{align*}
    \end{definition}

    \begin{remark}[Conventions and zero-padding]
         There are multiple possible conventions for how many singular values to define;
         $\min(p,q)$ works well for most purposes but one could also consider there being $p, q$ or $\max(p,q)$ singular values by padding with zeros appropriately.
         Conveniently, the $\cos^2\Theta$ definition by \Cref{eq:cos2theta-sim-def} does not depend on how many zero singular values are included.

         However, the $\sin^2\Theta$ distance does depend on the zero singular values, because their corresponding $\sin$ is 1; the dimension of $\Theta, \cos\Theta, \sin\Theta$ could also be defined in multiple reasonable ways.
         To avoid potential ambiguity we therefore only consider these objects when the dimensions of the subspaces are the same.
    \end{remark}

    From this definition, it is easy to deduce the following lemma, which will help us compute canonical angles in practice.

    \begin{lemma}
        Let $X, Y \in \R^{N \times K}$ be matrices with orthonormal columns spanning the subspaces $\mathcal{X},\mathcal{Y} \subset \R^N$.
        Then the cosines of the canonical angles between these subspaces are precisely the singular values of the matrix $X^\top Y \in \R^{K \times K}$.
    \end{lemma}
    \begin{proof}
        This follows immediately from \Cref{rem:operator-SVD-from-matrix-SVD} upon noting that the covariance operator is represented by the matrix $X^\top Y$ with respect to the orthonormal bases for $\mathcal{X}, \mathcal{Y}$ formed by the columns of $X, Y$ respectively.
    \end{proof}
    
    \begin{remark}[Consistency with definition from \citet{stewart_matrix_1990}]\label{rem:canonical-angle-consistency-with-stewart-sun}
        Our definitions are consistent with standard definition of canonical angles from the literature, and so inherit all standard results.
        For example, \citet{stewart_matrix_1990} define canonical angles via the following theorem below.
        It is straightforward to prove this theorem and show that the resulting canonical angles are consistent using the machinery we have developed; however, this is tedious and therefore omitted.
    \end{remark}

    \begin{theorem}[Pairs of subspaces, from \citet{stewart_matrix_1990}]\label{thm:ss-pairs-of-subspaces}
    Let $X, Y \in \R^{N \times K}$ with $X^{\top} X=I_K$ and $Y^{\top} Y=I_K$. If $2 K \leq N$, there are unitary matrices $Q, U$ and $V$ such that
    \begin{align*}
        Q X U = 
        \begin{pmatrix}
        I \\ 0 \\ 0
        \end{pmatrix}\hspace{-5pt}
        \begin{array}{l}
        \} K \\ \} K \\ \} N-2K
        \end{array}, \quad
        Q Y V = \begin{pmatrix}
        \Gamma \\ \Sigma \\ 0
        \end{pmatrix}\hspace{-5pt}
        \begin{array}{l}
        \} K \\ \} K \\ \} N-2K
        \end{array}
    \end{align*}
    where $\Gamma=\operatorname{diag}\left(\gamma_1, \ldots, \gamma_K\right)$ and $\Sigma=\operatorname{diag}\left(\sigma_1, \ldots, \sigma_K\right)$ are such that $\sigma_1 \geq \ldots \geq \sigma_K \geq 0$ and $\Gamma^2 + \Sigma^2 = I_K$.
    \end{theorem}

    These canonical angles can also be written in terms of the projection matrices on to the two respective subspaces.
    The proposition below is precisely theorem 5.5 from \citet{stewart_matrix_1990}.
    \begin{corollary}\label{cor:canonical-angles-to-projections}
        Let $\mathcal{X}$ and $\mathcal{Y}$ be $K$-dimensional subspaces of $\R^N$ with $2K \leq N$; let $\sigma_1 \geq \sigma_2 \geq \cdots \geq \sigma_K$ be the sines of the canonical angles between $\mathcal{X}$ and $\mathcal{Y}$.
        Then 
        \begin{itemize}
            \item The singular values of $P_{\mathcal{X}}\left(I-P_{\mathcal{Y}}\right)$ are $\sigma_1, \sigma_2, \ldots, \sigma_K, 0, \ldots, 0$.
            \item The singular values of $P_{\mathcal{X}}-P_{\mathcal{Y}}$ are $\sigma_1, \sigma_1, \sigma_2, \sigma_2, \ldots, \sigma_K, \sigma_K, 0, \ldots, 0$.
        \end{itemize}
        Note that in each case there are a total of $N$ singular values because the operators send $\R^N \to \R^N$.
    \end{corollary}

\newpage
\section{CCA foundations: existence, uniqueness and equivalent formulations}\label{app:cca-foundations}
    Standard mathematical treatments of CCA, such as \citet{anderson_introduction_2003} typically assume that $\Sigxx, \Sigyy$ are invertible.
    Though this is a reasonable assumption for a population distribution, this means that the analysis cannot apply to sample CCA in high dimensions when $\Cxx, \Cyy$ are low rank $N$.

    We give a different introduction of CCA that focuses on canonical variates, and naturally applies both to population CCA and to sample CCA, whether or not these within-view matrices are invertible.
    We also carefully specify the appropriate notion of uniqueness for CCA, which we could not find in previous texts; this helps provide intuition for what types of statistical estimation are most natural, see \Cref{rem:uniqueness-two-types}.
    
\subsection{A general definition of CCA}
    \begin{definition}[Full canonical decomposition, general definition]\label{def:canonical-decomposition-most-general}
        Let $(x_1,\dots,x_p),$ $(y_1, \dots, y_q)$ be two collections of vectors all from some common vector space $\mathcal{H}$ with inner product $\langle \cdot, \cdot \rangle$ and corresponding norm $\norm{\cdot}$.
        Let the subspaces of vectors they generate be $\mathcal{X} = \cspan{x_j: j=1,\dots,p}$, $\mathcal{Y} = \cspan{y_j: j=1,\dots,q}$ and have dimensions $R_x, R_y$ respectively.
        For $u \in \R^p, v\in \R^q$ introduce the shorthand
        \begin{align*}
            \genX^\top u = \sum_{j=1}^p u_j x_j, \quad \genY^\top v = \sum_{j=1}^q v_j y_j \; .
        \end{align*}

        Then we say that a pair of sequences of vectors $(u_k)_{k=1}^{R_x}, (v_k)_{k=1}^{R_y}$ give a full basis of canonical directions for $(x_j)_j, (y_j)_j$ if for $k \leq \min(R_x,R_y)$ the directions $u_k, v_k$ solve
        \begin{equation}\label{eq:cca-variate-sequential-optim-def}
    	\begin{split}
        	\underset{u \in \mathbb{R}^{p}, v \in \mathbb{R}^{q}}{\operatorname{maximize}}\,
        	& \langle \genX^\top u, \genY^\top v \rangle \\
        	\text{subject to  }& \norm{\genX^\top u}^2 \leq 1,\: \norm{\genY^\top v}^2 \leq 1,\\ 
        	& \langle \genX^\top u, \genX^\top u_l \rangle = \langle \genY^\top v, \genY^\top v_l \rangle =0 \text{ for } 1 \leq l \leq k-1 .	
        	\end{split}
    	\end{equation}
        and the remaining directions are such that $(\genX^\top u_k)_{k=1}^{R_x}, (\genY^\top v_k)_{k=1}^{R_y}$ form orthonormal bases for $\mathcal{X}, \mathcal{Y}$ respectively.

        If so, we may concatenate the successive vectors into matrices $U \in \R^{p \times R_x}, V \in \R^{q \times R_y}$ and equivalently write any of
        \begin{align*}
            \left((u_k)_{k=1}^{R_x}, (v_k)_{k=1}^{R_y}\right) \equiv (U, V)  \in \CCA(\genX,\genY) \equiv \CCA\left( (x_j)_{j=1}^p, (y_j)_{j=1}^q \right) \;.
        \end{align*}
    \end{definition}

    \begin{remark}
        Note that one could extend the $(u_k)_k, (v_k)_k$ to full bases for $\R^p, \R^q$ respectively, but in this case it is not obvious what the correct notion of orthonormality is. We found the formulation above most convenient to work with, and will not dwell further on this choice.
    \end{remark}      
    
\subsubsection{Existence, and uniqueness in variate space}
    \begin{proposition}[Equivalence to SVD of covariance operator]\label{prop:general-CCA-equivalent-to-SVD-covariance-operator}
        Adopt all setup and notation from \Cref{def:canonical-decomposition-most-general}.
        Consider the inner product $\langle \cdot, \cdot \rangle$ as a bi-linear form $\mathcal{H} \times \mathcal{H} \mapsto \R$.
        Then canonical correlation analysis is equivalent to a singular value decomposition of this (inner product) bi-linear form.
        Indeed, the canonical correlations are precisely the singular values of $\langle \cdot, \cdot \rangle$ and
        \begin{align*}
            (U,V) \in \CCA(\genX,\genY) \iff (\genX^\top U, \genY^\top V) \in \SVD(\langle \cdot, \cdot \rangle) \; .
        \end{align*}
    \end{proposition}
    \begin{proof}
        This immediately follows by equating the sequential optimisation definition of CCA given above \Cref{eq:cca-variate-sequential-optim-def} with the equivalent sequential optimisation formulation of SVD \Cref{thm:svd-variational-app}.
    \end{proof}

    This immediately implies that the canonical correlations are unique.
    The canonical variates inherit a notion of uniqueness from the SVD formulation; this uniqueness is up to orthogonal transformations within subspaces of the same canonical correlation, as we make precise with the next result.

    \begin{corollary}[Existence and uniqueness]\label{cor:cca-existence-uniqueness}
            The canonical correlations $(\rho_k)_k$ are unique, while the canonical variates $\genX^\top U, \genY^\top V$ inherit the notion of uniqueness from the corresponding SVD.

            In particular, if $U', V'$ were alternative sets of canonical directions satisfying \Cref{def:canonical-decomposition-most-general}, and $\rho = \rho_{k} = \dots = \rho_{k+\delta - 1} > 0$ have multiplicity $\delta$; then there exists some orthogonal matrix $O \in \R^{\delta \times \delta}$ such that for all $\epsilon \in [\delta]$ we have
            \begin{align*}
                (\genX^\top u)'_{(k-1) + \epsilon} = \sum_{\eta=1}^\delta O_{\eta, \epsilon} \genX^\top u_{(k-1) + \eta}, 
                \quad \text{and} \quad
                \genY^\top v'_{(k-1) + \epsilon} = \sum_{\eta=1}^\delta O_{\eta, \epsilon} \genY^\top v_{(k-1) + \eta} 
                \;.
            \end{align*}
            By contrast, for the zero singular value, we can only say that 
            \begin{align*}
                \cspan{\genX^\top u'_{K+1}, \dots, \genX^\top u'_p} &= \cspan{\genX^\top u_{K+1}, \dots, \genX^\top u_p} = \mathcal{X} \cap \mathcal{Y}^\perp \text{ , and }  \\
                \cspan{\genY^\top v'_{K+1}, \dots, \genY^\top v'_p} &= \cspan{\genY^\top v_{K+1}, \dots, \genY^\top v_p} = \mathcal{Y} \cap \mathcal{X}^\perp \;.
            \end{align*}
    \end{corollary}
    \begin{proof}
        This follows immediately by combining the SVD formulation of CCA \Cref{prop:general-CCA-equivalent-to-SVD-covariance-operator} with the uniqueness result for SVD \Cref{cor:svd-existence-uniqueness-app}.
    \end{proof}

    \begin{remark}[Canonical bases]
        By comparing \Cref{prop:general-CCA-equivalent-to-SVD-covariance-operator} with \Cref{def:canonical-bases} we see that CCA is equivalent to picking canonical bases in variate space.
        This provides another way to think about the break down of sample CCA when $N \leq p$, and specifically when the columns of $\X$ span $\R^N$; then $\cspan{\Y} \subset \R^N = \cspan{\X}$ so any orthonormal basis for $\cspan{\Y}$ can be extended to give a pair of canonical bases, and all the canonical angles are zero.
    \end{remark}

    \begin{remark}[Additional non-uniqueness in direction space]
        Consider the operators $\mathcal{E}_x: u \mapsto \genX^\top u, \mathcal{E}_y: v \mapsto \genY^\top v$.
    Then by linearity of these operators, we have
    \begin{align}\label{eq:additional-non-uniqueness-general}
        \genX^\top u' = \genX^\top u \iff u' = u + \nu_x, \quad \text{where} \quad \nu_x \in \ker(\mathcal{E}_x), \quad \text{i.e.,} \quad \genX^\top \nu_x = 0 \; .
    \end{align}

    We can plug this into \Cref{cor:cca-existence-uniqueness} to read off the corresponding notions of uniqueness for the canonical directions.
    In the subsection below we examine what the kernel of these operators look like in more detail.
    \end{remark}

\subsubsection{Application to population and sample CCA}
    The formulation above immediately applies to both population CCA and sample CCA.
    
    For population CCA take $\mathcal{H} \subset \mathcal{L}^2$ to be some space of square integrable random variables with mean zero, and the successive $x_j \defeq X_j, y_j \defeq Y_j$ for the components of the random vectors $X,Y$ respectively; then the inner product is precisely the covariance operator.
    The additional non-uniqueness in direction space from \Cref{eq:additional-non-uniqueness-general} is up to vectors in the kernel of $\Sigxx, \Sigyy$ respectively because
    \begin{equation}\label{eq:population-direction-non-uniqueness}
        X^\top u = X^\top \tilde{u} \text{ a.s.} \quad \iff \quad \mathbb{E} \|X^\top (u-\tilde{u}) \|^2 = 0 \quad \iff \quad (u-\tilde{u}) \in \operatorname{Ker}(\Sigxx) \; ,
    \end{equation}
    and similarly for the $Y$-view.
    
    For sample CCA take $\mathcal{H} = \{\ve \in \R^N: \ve^\top \ones{n} = 0\}$ to be the space of observed sample vectors with mean zero, and the successive $x_j \defeq \X_j, y_j \defeq \Y_j$ for the sample vectors corresponding to successive components of the data matrices $\X,\Y$ respectively.
    Now the inner product is precisely the sample covariance operator.
    In this case, one can interpret the non-uniqueness from \Cref{eq:additional-non-uniqueness-general} as being up to vectors in the kernels of the respective data matrices, i.e. $\nu_x \in \ker(\X)$ and similarly for the $Y$-view.

    \begin{remark}[Two types of uniqueness: applied relevance]\label{rem:uniqueness-two-types}
        Separating out these two different forms of non-uniqueness helps explain our empirical observations and methodological suggestions.
        \begin{itemize}
            \item The additional non-uniqueness of directions due to low-rank within-view covariance matrices helps explain why we observe better estimation in variate space than direction space. This is one reason why we consider analysing estimated variates and loadings rather than directions. On a more philosophical level, one could argue from the definition in \Cref{eq:cca-variate-sequential-optim-def} and result of \Cref{prop:general-CCA-equivalent-to-SVD-covariance-operator} that variates are more fundamental than the directions.
            \item The fundamental non-uniqueness up to orthogonal transformations within variate subspaces helps explain why we observe better estimation with subspace losses rather than considering separate subspaces. This motivates the need for registration when comparing different estimates.
        \end{itemize}                
        The statements above are deliberately blunt, but help convey our key intuitions.
        When analysing estimates from some regularised CCA method in practice, the situation is evidently more complicated.
        For example, when the within-view covariances are full rank but poorly conditioned, the canonical directions are technically unique, but we may expect estimated directions to be less stable than estimated variates.
        It may also be helpful conceptually to consider more subtle behaviour in the context of interactions between these two different forms of non-uniqueness.
    \end{remark}
    
\subsection{Further properties}
    
\subsubsection{Matrix SVD representation for CCA}\label{app:matrix-SVD-for-CCA}

    
    \begin{theorem}\label{thm:svd-equivalence-app}
        Suppose $\Sigxx,\Sigyy$ are both invertible.
        Then
        \begin{align*}
            (U_K, V_K) \in \CCA_K(X,Y)
            \iff
            (\Sigxx^\half U_K, \Sigyy^\half V_K) \in \operatorname{SVD}(\Sigxx^{-1/2} \Sigxy\Sigyy^{-1/2})
        \end{align*}
        and the canonical correlations are precisely the singular values.
    \end{theorem}
    \begin{proof}
        To help us, recall our notation for the SVD target matrix from \Cref{eq:SVD-target-T-def} in the main text:
        \begin{equation}
        	T \defeq \Sigxx^{-1/2} \Sigxy\Sigyy^{-1/2} \; .
        \end{equation}
        We also now introduce notation for the transformed weights $\alpha_k \defeq \Sigma_x^{1/2}u_k$ and $\beta_k \defeq \Sigma_y^{1/2}v_k$. 

        To prove the claim, simply observe
	\begin{eqnarray*}
		\operatorname{Cov}\left(u^\top X, v^\top Y\right)    &= \:u^\top \Sigxy v         &=\: \alpha^\top T \beta \\
		\operatorname{Cov}\left(u^\top X, u_{j}^\top X\right)&= \:u^\top \Sigxx u_{j}^\top &=\: \alpha^\top \alpha_j \\
		\operatorname{Cov}\left(v^\top Y, v_{j}^\top Y\right)&= \:v^\top \Sigyy v_{j}^\top &=\: \beta^\top  \beta_j.
	\end{eqnarray*}
        
        So the optimisation definition of CCA \Cref{eq:def-via-Sig} is indeed equivalent to the variational characterisation in Theorem \ref{thm:svd-variational-app} for finding the singular values of $T$.
    \end{proof}

    Note that this theorem immediately implies 
	\begin{equation*}
		\Sigxy = \Sigxx^{1/2}\, T\, \Sigyy^{1/2} = \Sigxx \left(\sum_{k=1}^K \rho_k u_k v_k^\top \right) \Sigyy
	\end{equation*}
    as we claimed in the main text.

\subsubsection{Reconstruction and latent variables}\label{sec:dimension-reduction-reconstruction}
    One key motivation for using PCA is its interpretation in terms of optimal reconstruction from low-dimensional transformed variables, or principal component \textit{scores}.
    The scores can therefore be thought of as latent variables.
    This latent variable interpretation is particularly clean when the data is assumed to be multi-variate Gaussian; in this setting it is well known that PCA can be recovered from a simple Gaussian linear model \citep{bishop_pattern_2006, roweis_unifying_1999}.

    CCA can be interpreted analogously in terms of reconstruction from the canonical variates.
    As we now show, it is precisely the canonical loading vectors which determine the reconstruction\footnote{Hence their name, drawing off the factor analysis literature \citep{gorsuch2014factor}.}.
    So to interpret CCA via reconstruction it is unnecessary to estimate weights, and sufficient to estimate variates and loadings.

    A formula for such reconstruction follows quickly from the orthogonality of canonical variates.
    Suppose we obtain bases $\{u_1,\dots,u_{R_x}\},\{v_1,\dots,v_{R_y}\}$ for $\range{\Sigxx}, \range{\Sigyy}$ respectively from a full CCA. 
    
    Then if we express $X,Y$ in the bases $\{\Sigxx u_1,\dots,\Sigxx u_{R_x}\}$,  $\{\Sigyy v_1,\dots,\Sigyy v_{R_y}\}$ the coefficients (determined by orthogonality) are precisely the canonical variates; this gives the expansion
    \begin{equation}\label{eq:reconstruction-base}
        X = \sum_{k=1}^{R_x} \Sigxx u_k \: \langle  u_k, X \rangle \;, \qquad 
        Y = \sum_{k=1}^{R_y} \Sigyy v_k \: \langle  v_k, Y \rangle \;.
    \end{equation}

    We can now use the correlation structure of the canonical variates to come up with more interpretable latent variable expressions; we formalise one possibility below. 
    Note that this result holds without any distributional assumptions, but that the latent variables are not necessarily independent, only uncorrelated.

    \begin{proposition}
        There exist orthonormal scalar random variables $(\xi_k)_{k=1}^{\max(R_x,R_y)},(\epsilon_k)_{k=1}^{R_y}$ such that
        \begin{equation}\label{eq:reconstruction-asymmetric}
            X = \sum_{k=1}^{R_x} \Sigxx u_k \: \xi_k, \qquad 
            Y = \sum_{k=1}^{R_y} \Sigyy v_k \: \left(\rho_k \xi_k + (1-\rho_k^2)^{1/2} \epsilon_k \right)\; .
        \end{equation}
    \end{proposition}
    \begin{proof}
        Write $\xi_k = u_k^\top X$ and define $\tilde{\epsilon}_k = v_k^\top Y - \xi_k \Cov(\xi_k, v_k^\top Y) = v_k^\top Y - \rho_k \xi_k$.
        Then $\Var(\tilde{\epsilon}_k) = 1 - 2 \rho_k^2 + \rho_k^2 = 1-\rho_k^2$.
        So define $\epsilon_k = \tilde{\epsilon}_k / (1-\rho_k^2)^\half$.
        Then the $(\xi_k)_k, (\epsilon_k)_k$ are orthonormal by construction.
        Moreover, rearranging the previous equations gives $v_k^\top Y = \rho_k \xi_k + (1-\rho_k^2)^{1/2} \epsilon_k$ so \Cref{eq:reconstruction-asymmetric} follows from \Cref{eq:reconstruction-base}.
    \end{proof}

    We conclude by completing the analogy to PCA.
    The expressions for reconstruction above also satisfy a notion of optimality, in terms of between-view covariances, see \Cref{app:biplots}.
    In addition, CCA can be recovered from a simple Gaussian linear model \citep{bach_probabilistic_2005}, as we present in \Cref{sec:probabilistic-cca}; note that, as in \Cref{eq:reconstruction-base}, the canonical variates play a role as latent variables, and reconstruction is via the canonical loading vectors.


\newpage
\section{Registration and biplots}
\subsection{Registration}\label{app:registration}
\begin{lemma}\label{lem:orthog-registration}
    Let $Z_0 \in \R^{n \times k_0}, Z_1 \in \R^{n \times k_2}$ have orthonormal columns, where $1 \leq k_0 \leq k_1$.
    Let $Z_1^\top Z_0$ have singular value decomposition $U D V^\top$. 
    Then $S^* \defeq U V^\top$ satisfies
    \begin{align*}
        S^* \in \underset{S \in \R^{k_1 \times k_0} \text{ s.t. } S^\top S = I_{k_0}}{\argmin} \| Z_1 S - Z_0 \|_F^2 \; .
    \end{align*}
\end{lemma}
\begin{proof}
    For arbitrary $S \in \R^{k_1 \times k_0}$ with $S^\top S = I_{k_0}$ we have
    \begin{align*}
        \| Z_1 S - Z_0 \|_F^2
        &= \| Z_1 S \|_F^2 + \| Z_0 \|_F^2 - 2 \langle Z_1 S, Z_0 \rangle_F \\
        &= \tr S^\top Z_1^\top Z_1 S + \tr Z_0^\top Z_0 - 2 \tr S^\top Z_1^\top Z_0 \\
        &= \tr S^\top I_{k_1} S + \tr I_{k_0} - 2 \langle S Z_1^\top Z_0 \rangle_F \\
        &= 2 k_0 - 2 \langle S Z_1^\top Z_0 \rangle_F \; .
    \end{align*}

    Next, by von Neumann's trace inequality \citep{carlsson_von_2021}, we have
    \begin{align*}
        \langle S Z_1^\top Z_0 \rangle_F 
        & \leq \sum_{i=1}^{k_0} \sigma_i(S) \sigma_i(Z_1^\top Z_0) \\
        & = \sum_{i=1}^{k_0} \sigma_i(Z_1^\top Z_0)
    \end{align*}
    where we use that the singular values of a matrix with orthonormal columns are all 1.

    Finally we note that $S^*$ as defined in the lemma statement does satisfy this equality; so maximises $\langle S Z_1^\top Z_0 \rangle_F $ and minimises our original objective.
\end{proof}

\begin{lemma}\label{lem:arbitrary-registration}
    Let $Z_0 \in \R^{n \times k_0}, Z_1 \in \R^{n \times k_2}$ be such that $Z_1$ is of full rank. Then
    \begin{align*}
        \left(Z_1^\top Z_1\right)^{-1} Z_1^\top Z_0 = \underset{M \in \R^{k_1 \times k_0}}{\argmin} \| Z_1 M - Z_0 \|_F^2 \; .
    \end{align*}
    Moreover, if $Z_0$ has orthonormal columns then the minimum value is precisely $\sin^2\Theta(Z_0,Z_1)$.
\end{lemma}
\begin{proof}
    Expanding out the Frobenius norm and completing the square gives
    \begin{align*}
        \| Z_1 M - Z_0 \|_F^2
        &= \tr \left\{ M^\top Z_1^\top Z_1 M - 2 M^\top Z_1^\top Z_0 + Z_0^\top Z_0 \right\} \\
        &= \tr \left(M - \left(Z_1^\top Z_1\right)^{-1}Z_1^\top Z_0\right)^\top \left(Z_1^\top Z_1\right) \left(M - \left(Z_1^\top Z_1\right)^{-1}Z_1^\top Z_0\right) \\
        & \qquad 
        - \tr \left(Z_1^\top Z_0\right)^\top \left(Z_1^\top Z_1\right)^{-1} \left(Z_1^\top Z_0\right)
        + \tr Z_0^\top Z_0 \; .
    \end{align*}

    So by full rankness of $Z_1$, the quadratic form is strictly convex and the unique minimum is $M^* \defeq \left(Z_1^\top Z_1\right)^{-1} Z_1^\top Z_0$ as stated.

    Moreover, the minimum value obtained is
    \begin{align*}
        \norm{Z_1 M^* - Z_0}_F^2
        &= \tr \left\{ Z_0 Z_0^\top (I - \mathcal{P}_{Z_1}) \right\} \; .
    \end{align*}

    So when $Z_0$ has orthonormal columns, this minimum value is precisely
    \begin{align*}
        \tr \mathcal{P}_{Z_0} (I - \mathcal{P}_{Z_1})
        = \norm{\mathcal{P}_{Z_0} (I - \mathcal{P}_{Z_1})}_F^2 = \sin^2\Theta(Z_0,Z_1)
    \end{align*}
    by idempotency of projection matrices and then \Cref{cor:canonical-angles-to-projections}.
\end{proof}

\subsection{Biplot background}\label{app:biplots}
\subsubsection{Original definition}    
    We now provide further background for the idea of biplots introduced in \Cref{sec:comparison-biplots}.
    To ease notation, \citet{ter_braak_interpreting_1990} assumed that all the original variables $X_i, Y_i$ have been standardised to have variance 1.
    We will work in this setting for the rest of this subsection.
    We can replace the correlations in \Cref{def:pop-biplot-corr} with covariances.    
    The population version of the original definition of biplots in \citet{ter_braak_interpreting_1990} is as follows:

    \begin{definition}[Population Covariance Biplot, adapted from \citet{ter_braak_interpreting_1990}]\label{def:pop-biplot-original}
        Let $U_K,V_K$ be matrices of canonical directions. Respectively plot the $i\uth$ variable in view 1 and $j\uth$ variable in view 2 with coordinates
        \begin{equation*}
            \left((\Sigxx U)_{ik}\right)_{k=1}^K\; ,  \quad 
            \left((\Sigyy V R)_{jk}\right)_{k=1}^K\; .
        \end{equation*}
    \end{definition}

    This is equivalent to \Cref{def:pop-biplot-corr} under our assumption of variables with variance 1.
    Indeed, \Cref{eq:sigxy-from-T-and-URV} implies 
    $$U^\top \Sigxy = U^\top \Sigxx U R V^\top \Sigyy = R V^\top \Sigyy$$
    so in fact
    \begin{align}\label{eq:change-of-basis}
        \rho_k \Cov(Y_j, v_k^\top Y)
        = \left(\Sigyy V R\right)_{jk}
        = \left(\Sigyx U \right)_{jk} 
        = \Cov(Y_j,u_k^\top X)
    \end{align}
    as required.

    \begin{remark}
        \citet{ter_braak_interpreting_1990} also suggest plotting in the symmetric coordinates $\left((\Sigxx U R^\half)_{ik}\right)_{k=1}^K$ and $\left((\Sigyy V R^\half)_{jk} \right)_{k=1}^K$ but this makes the interpretation somewhat more difficult so we do not pursue the idea further.
    \end{remark}

\subsubsection{Optimality of Low rank approximation}\label{sec:biplot-optimality-eckhart-young}   
    Proposition \ref{prop:biplots} states a notion of optimality for biplots; this is a simple consequence of the powerful Eckhart--Young theorem, which is standard in matrix analysis \citep{stewart_matrix_1990,bhatia_matrix_1997}, stated below for completeness.
    \begin{theorem}[Eckhart--Young]\label{thm:eckhart-young}
        Let $T \in \R^{p\times q}$. Then $M$ minimises $\|T-M\|_F$ over matrices $M$ of rank at most $K$ if and only if $M = A_K R_K B_K^\top$ where $(A_K,R_K,B_K)$ is some top-$K$ SVD of the target $T$. 
    \end{theorem}
    \begin{remark}
        In fact Eckhart--Young holds for any unitarily invariant norm.
    \end{remark}
    
    \begin{restatable}[Optimality of biplots] {proposition}{propbiplots}
        \label{prop:biplots}
        Suppose $M^*$ maximises the following loss function over $p\times q$ matrices $M$ of rank at most $k$
        \begin{equation}\label{eq:eckhart-young-for-cca}
            \|\Sigxx^\mhalf (\Sigxy - M) \Sigyy^\mhalf \|_F = \|\Sigxx^\mhalf \Sigxy \Sigxx^\mhalf - \Sigxx^\mhalf M \Sigyy^\mhalf\|_F 
        \end{equation}
        then 
        \begin{equation}
            M^* = \Sigxx U_K R_K V_K^\top \Sigyy
        \end{equation}
        where $(U_K,R_K,V_K)$ is some solution to the top-$K$ CCA problem.
    \end{restatable}
    \begin{proof}
        As before, write $T=\Sigxx^\mhalf \Sigxy \Sigyy^\mhalf$. Then writing $\tilde{M} = \Sigxx^\mhalf M \Sigyy^\mhalf$ the objective is just $\|T-\tilde{M}\|$. So by Eckhart--Young (Theorem \ref{thm:eckhart-young}), we have $\tilde{M} = A_K R_K B_K^\top$ where $(A_K,R_K,B_K)$ are a top-$K$ SVD for $T$. Now define $U_K = \Sigxx^\mhalf A_K, V_K = \Sigyy^\mhalf B_K$; then by the arguments of Section \ref{app:matrix-SVD-for-CCA}, this happens precisely when $(U_K,R_K,V_K)$ is some solution to the top-$K$ CCA problem.
    \end{proof}

\newpage 

\newpage
\section{Multiple correlations}\label{app:multiple-correlations}
\subsection{Valid functions for aggregating correlations}
We now present some further background motivating how to aggregate successive correlations to provide a scalar objective.
We show that a broad class of aggregation methods do in fact recover CCA, and are therefore well motivated.
Oddly, we are not aware of similar analysis or discussion in the existing literature.
\Cref{lem:sums-convex-functions-matrix,prop:sum-convex-functions-CCA} may appear to have long statements and proofs; we reassure the reader that in each case the main content is given in the first half of the statements and proofs, while the second half contains more careful (repetitive) analysis of the equality cases.

\begin{lemma}[Interlacing for SVD]\label{lem:interlacing-for-svd}
    Let $A \in \R^{m \times n}$ be partitioned as
    \begin{align*}
        A = \begin{pmatrix}
                A_{11} & A_{12} \\ A_{21} & A_{22}
        \end{pmatrix}
    \end{align*}
    where $A_{11} \in \R^{m' \times n'}$ and $m' \in [m], n' \in [n]$.
    Then
    \begin{align*}
        \sigma_i(A_{11}) \leq \sigma_i(A) \quad \forall i =1,\dots,n \; .
    \end{align*}
\end{lemma}
\begin{proof}
    Apply Theorem 4.4 of \citet{stewart_matrix_1990} twice.
\end{proof}

We now introduce one additional piece of notation that will be helpful for the remainder of this section.
Let $\CCAcorr{k}(X,Y)$ denote the $k^\text{th}$ canonical correlation of the pair of random variables $(X,Y)$.
(Note that then $\CCAcorrto{K}(X,Y)$ as defined in \Cref{def:cc-vector} is precisely $\left(\CCAcorr{k}(X,Y)\right)_{k=1}^K$).

\begin{lemma}[Interlacing for CCA]\label{lem:interlacing-for-cca}
    Let $X,Y$ be random vectors taking values in $\R^p,\R^q$ respectively.
    Suppose we have two arbitrary matrices $\hat{U} \in \R^{p \times K},\hat{V} \in \R^{q\times K}$.
    Then we have the element-wise inequalities
    \begin{equation}\label{eq:cca-interlacing}
        \CCAcorr{K}(\hat{U}^\top X, \hat{V}^\top Y) \leq \CCAcorr{K}(X,Y)  
    \end{equation}
\end{lemma}
\begin{proof}
    Take some orthonormal bases $e_1,\dots, e_K$ and $f_1,\dots,f_K$ for $\spann(\hat{U}^\top X)$ and $\spann(\hat{V}^\top Y)$ respectively.
    Extend these to orthonormal bases $e_1,\dots,e_p$ and $f_1,\dots,f_q$ for $\spann(X),\spann(Y)$ respectively.
    Then from the discussion in Section \ref{app:cca-foundations}, $\CCA(X,Y)$ is precisely the vector of singular values of the matrix $E^\top F$ whose entries are $(E^\top F)_{kl} = \langle e_k, f_l \rangle$.
    
    Moreover, $\CCAcorr{}(\hat{U}^\top X, \hat{V}^\top Y)$ is precisely the vector of singular values of the top-left $K\times K$ sub-matrix of $E^\top F$.
    So by \cref{lem:interlacing-for-svd}, the interlacing theorem for singular values, we do indeed obtain (\ref{eq:cca-interlacing}).
\end{proof}

\begin{lemma}[Sums of increasing convex functions]\label{lem:sums-convex-functions-matrix}
    Let $f : \R \to \R$ be a convex function such that $f(x) \geq f(-x) ~ \forall x > 0$.
    Let $A \in \R^{K\times K}$ be a square matrix.
    Then
    \begin{align}\label{eq:sums-of-convex-lemma}
        \sum_{k=1}^K f(A_{kk}) \leq \sum_{k=1}^K f(\sigma_k(A)) \; .
    \end{align}
    \eqcase{
    Moreover, if $f$ is strictly convex then equality implies that there exists some permutation $\pi \in S_K$ such that $|A_{kk}| = \sigma_{\pi(k)}(A)$.
    Meanwhile, if $f$ is linear and non-constant then equality implies that there is an orthogonal matrix $O$ such that $O^\top A O = \diag\sigma(A)$ is a valid SVD of $A$.
    }
\end{lemma}
\begin{proof}
    Fix some $f,A$.
    Define the function $g: \R^K \to \R, x \mapsto \sum_{k=1}^K f(x_k)$.
    Then $g$ can be viewed as the composition of convex functions and so is itself convex.
    
    Write $a = (A_{kk})_{k=1}^K$ for the vector of diagonal entries of $A$.
    Then the left hand side of (\ref{eq:sums-of-convex-lemma}) is precisely $g(a)$.
    
    Now take some $b \in \partial g(a)$, the sub-differential of $g$ at $a$.
    Then writing $m=g(a)- a^\top b$, we have by convexity that $\forall x \in \R^K$
    \begin{align}\label{eq:convexity-application}
         m + \sum_k x_k b_k = m + x^\top b \leq g(x) = \sum_k f(x_k) \; .
    \end{align}
    
    Now let $B = \diag(b)$ be the diagonal matrix with entries $B_{kk} = b_k$.
    Then the singular values of $B$ are precisely $|b_{\pi(1)}| \geq \dots \geq |b_{\pi(K)}|$ where $\pi$ is some permutation of $\{1,\dots,K\}$ that arranges the elements of $b$ into descending order of modulus.
    
    Then by von Neumann's trace inequality \citep{carlsson_von_2021} and re-indexing the sum we have
    \begin{align}\label{eq:von-Neumann-application}
        a^\top b = \tr(A B) 
        \leq \sum_{k=1}^K \sigma_k(A) \sigma_k(B)
        = \sum_{k=1}^K \sigma_k(A) |b_{\pi(k)}|
        = \sum_{l=1}^K \left(\sigma_{\pi^{-1}(l)}(A) \, \operatorname{sign}(b_l)\right) b_l.
    \end{align}
    
    Combining the above
    \begin{align}
        \sum_{k=1}^K f(A_{kk})
        &= g(a) \nonumber\\
        &= m + a^\top b \nonumber\\
        &= m + \tr(A B) \nonumber\\
        &\leq m + \sum_{l=1}^K \left(\sigma_{\pi^{-1}(l)}(A) \, \operatorname{sign}(b_l)\right) b_l \label{eq:sums-von-Neumann}\\
        &\leq \sum_{l=1}^K f\left( \sigma_{\pi^{-1}(l)}(A) \, \operatorname{sign}(b_l)\right) \label{eq:sums-convex-hyperplane}\\
        &\leq \sum_{l=1}^K f\left(\sigma_{\pi^{-1}(l)}(A)\right) \label{eq:sums-positive-better}\\
        &= \sum_{k=1}^K f\left(\sigma_{k}(A)\right) \nonumber
    \end{align}
    where the first inequality was from (\ref{eq:convexity-application}), the second inequality from (\ref{eq:von-Neumann-application}), and the third from the assumption that $f(x) \geq f(-x)~ \forall x > 0$.
    
    \eqcase{
    For the moreover conclusions, note that in either case, equality in the theorem statement implies equality in all the inequalities in the chain of inequalities above.
    First consider the case where $f$ is strictly convex.
    Then equality holds in \Cref{eq:convexity-application} if and only if $x = a$.
    This was applied at \Cref{eq:sums-convex-hyperplane} to the vector $x$ with entries $x_l =\sigma_{\pi^{-1}(l)}(A) \, \operatorname{sign}(b_l)$.
    Equating with $a_l = A_{ll}$ and taking the absolute value gives the claim.    

    The condition that $f$ is linear and non-constant, combined with $f(x) \geq f(-x)$ for $x>0$ means that $f$ must have strictly positive slope.
    By subtracting the constant and rescaling we may therefore assume that $f(x) = x$.
    In this case, $\partial g(a) = \{ \ones{K} \}$ so we have $b=\ones{K}$ and $B$ is precisely the identity matrix.
    Singular value decompositions of $B$ are precisely $I = O I O^\top$ for an orthogonal matrix $O \in \R^{K \times K}$.
    But then the equality case of von Neumann's trace inequality from \citet{carlsson_von_2021} states that $A$ and $B$ `share singular vectors', which in this case implies that there is some orthogonal matrix $O$ such that $O^\top A O = \diag\sigma(A)$ is a valid SVD.
    }
\end{proof}

Finally we combine the results above to show that CCA can indeed be defined by maximising sums of convex functions of canonical correlations.

\begin{proposition}\label{prop:sum-convex-functions-CCA}
   Let $f : \R \to \R$ be a convex function such that $f(x) \geq f(-x) ~\forall x > 0$.
   Let $X,Y$ be random vectors taking values in $\R^p,\R^q$ respectively.
   For $K \in \{1,\dots,\min(p,q)\}$ define functions
   \begin{align*}
       g_K : \R^{p\times K} \times \R^{q\times K} \to \R, \: (U,V) \mapsto \sum_{k=1}^K f\left(\rho^\text{oracle}(u_k,v_k)\right)
   \end{align*}
   where $U,V$ have columns $u_k,v_k$ respectively.
   
   Then
   \begin{align}\label{eq:sum-of-convex-objective}
       \dCCA(X,Y) \subset \underset{\substack{U \in \R^{p\times K},V \in \R^{q\times K}  \\ U^\top \Sigxx U = V^\top \Sigyy V = I_K}}{\argmax}{\:\sum_{k=1}^K f\left((\rho^\text{oracle}(u_k,v_k)\right)}
       \eqdef \mathcal{G}_K \; .
   \end{align}

   \eqcase{
   Moreover, if $f$ is linear and non-constant then $\mathcal{G}_K$ consists precisely of matrices $U,V$ whose column span forms a top-$K$ subspace for the CCA problem. i.e. 
   \begin{align*}
       \mathcal{G}_K = \left\{(U O, V O): \: O \in \mathcal{O}(K), (U,V) \in \dCCA_K(X,Y)\right\}. 
   \end{align*}
   
   If however, $f$ is strictly convex and $f(x) = f(-x) \forall x \in \R$, then $\mathcal{G}_K$ consists precisely of matrices $U,V$ whose columns can be reordered and possibly negated to form a solution to the $\dCCA$ problem, i.e. 
    \begin{equation*}
    \begin{aligned}
        \mathcal{G}_K = \left\{(U P, V P S): \right. & \left. (U,V) \in \CCA_K(X,Y), \right. \\
        & \left. P \text{ is a permutation matrix and } S \text{ is diagonal with entries in } \pm 1 \right\} \; .
    \end{aligned}
    \end{equation*}
    }

\end{proposition}
\vspace{10pt}

\begin{proof}
    First note that if $U^*,V^* \in \dCCA_K(X,Y)$ then $g_K(U^*,V^*) = \sum_{k=1}^K f(\rho_k^*)$ where $\rho_k^* \defeq \CCAcorr{k}(X,Y)$. 
    So it is sufficient to show that this is the maximum possible value of $g_K$ over feasible $(U,V) \in \R^{p\times K} \times \R^{q\times K}$, i.e. those for which $U^\top \Sigxx U = I, V^\top \Sigyy V = I$.

    So now take arbitrary feasible $U,V$.
    Define the matrix 
    \begin{align*}
        M = U^\top \Sigxy V \in \R^{K\times K} \; .
    \end{align*}
    Then by the orthonormality constraint, in fact we have $M_{kk} = \Corr(u_k^\top X, v_k^\top Y) = \rho^{\text{oracle}}(u_k,v_k)$.
    So the objective in (\ref{eq:sum-of-convex-objective}) is precisely $\sum_{k=1}^K f(M_{kk})$.
    
    The orthonormality also means that the the canonical correlations $\CCAcorr{k}(U^\top X, V^\top Y)$ are precisely the singular values $\sigma_k(M)$ of $M$.
    So by Lemma \ref{lem:interlacing-for-cca} we have that $\sigma_k(M) \leq \rho_k^*$.
    
    So by Lemma \ref{lem:sums-convex-functions-matrix} we have 
    \begin{align*}
        g_K(U,V) = \sum_{k=1}^K f(M_{kk}) \leq \sum_{k=1}^K f(\sigma_k(M)) \leq \sum_{k=1}^K f(\rho_k^*)
    \end{align*}
    as required.

    \eqcase{
    Suppose that $f$ is linear and non-constant, and $(U,V)$ attain the maximal value of $\sum_{k=1}^K f(\rho_k^*)$.
    Then both inequalities in the final equation must be equalities.
    The latter gives $\sigma(M) = \rho$.
    The former implies that we must have $O^\top M O = \diag(\sigma(M))$ for some orthogonal matrix $O \in \R^{K\times K}$, by the corresponding equality case of \Cref{lem:sums-convex-functions-matrix}.
    Write $U^* = U O, V^* = V 0$.
    Then substituting the two results above gives ${U^*}^\top \Sigxy V^* = O^\top M O = \diag(\rho)$ and we still have the required orthonormality ${U^*}^\top \Sigxx U^* = {V^*}^\top \Sigyy V^* = O^\top O = I_K$.
    Therefore $U^*,V^*$ are indeed top-$K$ CCA matrices, as required.
    For the reverse inclusion, again by removing the intercept and re-scaling we may WLOG assume that $f(x) = x$.
    Take any $U = U^* O, V = V^* O$ with $(U^*, V^*) \in \dCCA_K(X,Y)$.
    Then $\sum_{k=1}^K \rho^\text{oracle}(u_k,v_k) = \tr U^\top \Sigxy V = \tr O^\top {U^*}^\top \Sigxy V^* O = \tr{U^*}^\top \Sigxy V^* = \sum_{k=1}^K \rho_k^*$ does indeed attain the maximum value.

    Finally, consider $f$ strictly convex with $f(x) = f(-x) \forall x \in \R$.
    Suppose $(U,V)$ attain the maximal value of $\sum_{k=1}^K f(\rho_k^*)$.
    Then again we must have $\sigma(M) = \rho$.
    We must also have equality in the application of \Cref{lem:sums-convex-functions-matrix} and therefore there is some permutation $\pi$ such that $|M_{kk}| = \sigma_{\pi(k)}(M)$.
    So combining we get $|M_{kk}| = \rho_{\pi(k)}$.
    
    Now define new matrices $U^*, V^*$ by their columns
    \begin{align*}
        U^*_k = U_{\pi^{-1}(k)}, \: V^*_k = V_{\pi^{-1}(k)} \operatorname{sign}(M_{\pi^{-1}(k)\pi^{-1}(k)}) \; .
    \end{align*}
    Then $U^*, V^*$ are both feasible, and satisfy ${U_k^*}^\top \Sigxy V^*_k = |M_{\pi^{-1}(k)\pi^{-1}(k)}| = \rho_{\pi\pi^{-1}(k)}^* = \rho_k^*$ so are indeed top-$K$ CCA matrices.
    It is straightforward to rewrite $U^*,V^*$ into the matrix form in the proposition statement.
    For the reverse inclusion, take any $U^*, V^*$ in this matrix form, then by direct computation one can verify that the maximal value is obtained.
    }
\end{proof}

\newpage
\section{Graphical CCA}\label{app:graphical-cca}
    
\subsection{Sparse precision gives sparse directions}\label{app:sparse-precision-to-directions}
In this section we will write $\Omega=\Sigma^{-1}$ for the precision matrix of the concatenated vector $\begin{pmatrix} X \\ Y\end{pmatrix} \in \R^{p+q}$ and decompose it in block form via
    $\Omega = \left(\begin{array}{ll} \Omega_{xx} & \Omega_{xy} \\ \Omega_{yx} & \Omega_{yy} \end{array}\right)
    \def\arraystretch{1.3}$.

    \sparseDirections*
    \begin{proof}
    We will use the Schur complement identity that $\Omega_{xx}^{-1} \Omega_{xy} = - \Sigxy\Sigyy^{-1}$.
    Starting with a rearrangement of \Cref{eq:sigxy-from-T-and-URV} then substituting in this identity gives
    \begin{align*}
        U D V^\top 
        &= \Sigxx^{-1} \Sigxy \Sigyy^{-1} \\
        &= \Sigxx^{-1} \Omega_{xx}^{-1} \Omega_{xy} \\
        &= (\Omega_{xx} - \Omega_{xy}\Omega_{yy}^{-1} \Omega_{yx}) \Omega_{xx}^{-1} \Omega_{xy} \\
        &= \Omega_{xy} ( I - \Omega_{yy}^{-1} \Omega_{yx} \Omega_{xx}^{-1} \Omega_{xy}). 
    \end{align*}
    
    Hence
    \begin{align*}
        U = (U D V^\top) \Sigyy V D^{-1} = \Omega_{xy} ( I - \Omega_{yy}^{-1} \Omega_{yx} \Omega_{xx}^{-1} \Omega_{xy}) \Sigyy V D^{-1} 
    \end{align*}
    and so $U$ inherits zero-rows from $\Omega_{xy}$: for any $a \in A$, 
    \begin{align*}
        U_{a,:} = ({\Omega_{xy}})_{a,:} ( I - \Omega_{yy}^{-1} \Omega_{yx} \Omega_{xx}^{-1} \Omega_{xy}) \Sigyy V D^{-1} = 0 \; ,
    \end{align*}
    as required.
\end{proof}
    
\subsection{Relationship to partial correlations}\label{app:partial-correlation-to-sparse-directions}
Partial correlations provide an alternative lens through which to investigate the relationship between sparse precision and sparse canonical directions.
These are often used in the statistics literature, and may help the reader develop complementary intuition.
We first give a definition and well known result relating partial correlations to the precision matrix\footnote{Some authors even define partial correlations this way, e.g. \citet{lauritzen_graphical_1996}.}, before applying this result to CCA.

\begin{definition}[Partial Correlation, \citep{prokhorov_partial_nodate}]
    Let $X, Y$ be scalar random variables, and $Z$ be a random vector in $\R^d$, each with mean zero.
    Let the coefficients for the optimal linear regression fit, and the corresponding residuals be
    \begin{align*}
        \beta_X &= \argmin_{\beta \in \R^{d}} \E \norm{X - \beta^\top Z}^2, & R_X &= X - \beta_X^\top Z, \\
        \beta_Y &= \argmin_{\beta \in \R^{d}} \E \norm{Y - \beta^\top Z}^2, & R_Y &= Y - \beta_Y^\top Z \; .
    \end{align*}
    Then the partial correlation is defined by
    \begin{align}
        \Corr\left(X, Y | Z\right) = \frac{\Cov(R_X, R_Y)}{\Var(R_X)^\half \Var(R_Y)^\half} \; .
    \end{align}

    For $X, Y, Z$ with arbitrary mean, one can define $\Corr\left(X, Y | Z\right) = \Corr\left(X - \E X, Y - \E Y | Z - \E Z\right)$.

    One can also extend this immediately to arbitrary random vectors $X \in \R^p, Y \in \R^q$ to give the matrix of partial correlations
    \begin{align*}
        \Corr\left( X, Y \mid Z \right) = \left( \Corr\left(X_i, Y_j \mid Z\right)\right)_{i \in [p], j \in [q]} \; .
    \end{align*}
\end{definition}

It is well known that partial correlations can be expressed in terms of the joint precision matrix, as made precise in the following lemma.

\begin{lemma}[Partial correlations from precision matrix]\label{lem:partial-correlation-from-precision}
    Let $Z$ be a full-rank random vector in $\R^d$ with precision matrix $\Omega$.
    Then for any $i,j \in [d]$ distinct
    \begin{align*}
        \Corr\left(Z_i, Z_j \mid Z_{(-ij)}\right) = \frac{\Omega_{ij}}{\Omega_{ii}^\half \Omega_{jj}^\half} \; .
    \end{align*}
\end{lemma}

\begin{proof}[Alternative proof of \Cref{prop:sparse-directions-from-sparse-precision} using variational characterisation]
    Write $B = [p] \setminus A$.
    By \Cref{lem:partial-correlation-from-precision}, zeroes in the precision matrix imply zero partial correlations.
    Therefore our main assumption becomes $\Corr\left(X_A; Y \mid X_B\right) = 0$.

    Now define the (multivariate) linear regression coefficients and residuals
    \begin{align*}
        \Gamma_A &= \argmin_{\Gamma \in \R^{|B| \times |A|}} \E \norm{X_A - \Gamma^\top X_B}^2, & 
        R_A &= X_A - \Gamma_A^\top X_B, \\
        \Gamma_Y &= \argmin_{\Gamma \in \R^{|B| \times q}} \E \norm{Y - \Gamma^\top X_B}^2, & 
        R_Y &= Y - \Gamma_Y^\top X_B \; .
    \end{align*}
    By properties of linear regression we have $\Cov(R_A, X_B) = \Cov(R_Y, X_B) = 0$.
    We can also now write the main assumption as $\Cov(R_A, R_Y) = 0$.    Now take an arbitrary pair $u \in \R^p, v \in \R^q$.
    Then from the definition of $R_A$ we have $u_A^\top X_A = u_A^\top \Gamma_A^\top X_B + u_A^\top R_A$, and therefore
    \begin{align}\label{eq:u-ua-ub-expansion}
        u^\top X = u_A^\top X_A + u_B^\top X_B = (u_B + \Gamma_A u_A)^\top X_B + u_A^\top R_A \; .
    \end{align}
    So now define $\tilde{u} \in \R^p$ by $\tilde{u}_B = u_B + \Gamma_A u_A, \tilde{u}_A = 0$.
    Then since $\Cov(R_A, X_B) = 0$ we have
    \begin{align}\label{eq:smaller-variance}
        \Var(u^\top X) = \Var(\tilde{u}_B^\top X_B) + \Var(u_A^\top R_A) \geq \Var(\tilde{u}_B^\top X_B) = \Var(\tilde{u}^\top X)
    \end{align}
    with equality if and only if $u_A^\top R_A = 0$, i.e. that $u^\top X = \tilde{u}^\top X$.

    We can also use the definition of residuals and orthogonality of $R_A$ to both $X_B$ and $R_Y$  to compute
    \begin{align*}
        \Cov(u^\top X, v^\top Y) - \Cov(\tilde{u}^\top X, v^\top Y)
        &= \Cov(u_A^\top R_A, v^\top Y) \\
        &= \Cov(u_A^\top R_A, v^\top \Gamma^\top X_B + v^\top R_Y) = 0. 
    \end{align*}
    and conclude that the projected variables have the same covariances with $v^\top Y$ as the original variables.
    They also have the same covariances with $X_B$ as the original variables by a similar computation
    \begin{align}\label{eq:orthogonality-to-B-variables-preserved}
        \Cov(u^\top X, X_B) - \Cov(\tilde{u}^\top X, X_B) = \Cov(u_A^\top R_A, X_B) = 0 \; .
    \end{align}


    Now let $(X^\top u_k, Y^\top v_k)_{k=1}^K$ be some sequence of canonical variates, as defined by the program of \Cref{eq:cca-variate-sequential-optim-def}.
    Inductively apply the construction above to each pair of directions $u_k, v_k$ to yield the sequence of vectors $\tilde{u}_k$ in $\R^p$.
    By the induction hypothesis, $u_l^\top X$ can be written only in terms of variables in $B$ for each $l<k$; so by \Cref{eq:orthogonality-to-B-variables-preserved}, $\tilde{u}_k^\top X$ satisfies the required orthogonality to these previous canonical variates.
    Since the canonical variates maximise the correlation $\Corr(u_k^\top X, v_k^\top Y)$, we must have equality in \Cref{eq:smaller-variance} and therefore that $u_k^\top X = \tilde{u}_k^\top X$.
    Finally, we note that $\tilde{u}_k$ has support in $B$, as required. 
\end{proof}


\subsection{Lipschitz Target Matrix}\label{app:lipschitz-plug-in}
First we recall the assumption from the main text.

\wellConditionedWithinViewVariances* 

We actually prove a slightly stronger version of the `Lipschitz plug-in' theorem from the main text, with an additional bound that will be useful in \Cref{sec:convert-to-variate-bounds}.
The proof spans the remainder of this section.

\begin{proposition}[Lipschitz plug-in]\label{prop:lipschitz-plug-in-app-longer}
        Assume Assumption \ref{ass:well-conditioned-within-view}$(M,m)$.
        Then there exists constants $C_1, C_2 \in (0, \infty)$ only depending on $M,m,\rho_1$ such that whenever $\norm{\Omega - \hat{\Omega}}_\textrm{op} \leq \frac{1}{4M}$ we have 
        \begin{equation}
            \norm{\hat{T} - T}_\textrm{op} \leq C_1 \norm{\Omega - \hat{\Omega}}_\textrm{op},
            \quad \text{and} \quad
            \norm{\hat{\Sigma}_{xx}^{-1} - \Sigxx^{{-1}}}_\textrm{op} \leq C_2 \norm{\Omega - \hat{\Omega}}_\textrm{op}\; .
        \end{equation}
\end{proposition}
    
\begin{proof}
The main idea is to show that the function $f: \Omega \mapsto T$ is Lipschitz by decomposing it as sums, products and compositions of Lipschitz functions. 

We shall write $\|\Omega - \hat{\Omega}\| = \epsilon$ for the remainder of the proof; so we seek to bound $\|\hat{T}-T\| \leq C_1 \epsilon$, and will track the constants in front of epsilon over the course of the proof. 
Note also that the only matrix norm needed for this proof is the operator norm, so we drop the subscript-$\text{op}$ from $\norm{\cdot}_\text{op}$ for the remainder of the proof.
All vector norms refer to the Euclidean norm.
In this section, we will sometimes denote diagonal blocks with a single subscript to avoid clutter (e.g., writing $\Sigma_x$ for $\Sigxx$). 

\subsubsection{Identities from Schur complements}
Firstly, let's introduce notation for the following Schur complements:
\begin{align*}
    S_x \defeq  \Omega_{x} - \Omega_{xy} \Omega_{y}^{-1}\Omega_{yx} = \Sigma_{x}^{-1} \\
    S_y \defeq  \Omega_{y} - \Omega_{yx} \Omega_{x}^{-1}\Omega_{xy} = \Sigma_{y}^{-1}\\[10pt]
    \Psi_x \defeq  \Sigma_{x} - \Sigxy \Sigma_{y}^{-1} \Sigyx = \Omega_{x}^{-1}\\
    \Psi_y \defeq  \Sigma_{y} - \Sigyx \Sigma_{x}^{-1} \Sigxy = \Omega_{y}^{-1}.\\
\end{align*}

Next recall the identity 
\begin{equation}\label{eq:xy-equivs}
    \Sigxy = \Sigma_x^{1/2} T \Sigma_y^{1/2} = \Sigma_x^{1/2} A D B^\top \Sigma_y^{1/2} = \Sigma_x U D V^\top \Sigma_y \; .
\end{equation}
We can now rewrite
\begin{equation}\label{eq:psi-x}
    \Psi_x = \Sigma_x - \Sigma_x^{1/2} T T^\top \Sigma_x^{1/2} = \Sigma_x^{1/2}(I-TT^\top) \Sigma_x^{1/2} = \Sigma_x^{1/2}(I-A D^2 A^\top) \Sigma_x^{1/2} \; .
\end{equation}
Starting with identity for off-diagonal inverse blocks from Schur's formula we also get
\begin{align}
    \Omega_{xy} 
    &= - \Psi_x^{-1} \Sigxy \Sigma_y^{-1} \nonumber\\
    &= - \Sigma_x^{-1/2}(I - T T^\top)^{-1} \Sigma_x^{-1/2} \cdot \Sigma_x^{1/2} T \Sigma_y^{1/2} \cdot \Sigma_y^{-1} \nonumber\\
    &= - \Sigma_x^{-1/2}(I - T T^\top)^{-1} T \Sigma_y^{-1/2} \nonumber\\
    &= - \Sigma_x^{-1/2} A \frac{D}{(I - D^2)} B^\top \Sigma_y^{-1/2} \label{eq:omegaxy}
\end{align}
where we understand that $D$ is diagonal with entries in $(0,1)$ so the inverse exists and we stretch notation to write $\frac{D}{(I - D^2)} = D (I-D^2)^{-1}$.

\subsubsection{Initial operator norm bounds}
Firstly note that $\|S_x\| = \|\Sigma_x^{-1}\|\leq m$, and similarly $\|S_y\| \leq m$. Furthermore, $0 \preceq \Psi_x \preceq \Sigma_x$, therefore 
\begin{equation}
    \|\Psi_x\| \leq \|\Sigma_x\| \leq M\;,
\end{equation}
and similarly, $\|\Psi_y\| \leq M$. Next by (\ref{eq:psi-x})
\begin{equation}
    \|\Omega_x\| \leq \|\Sigma_x^{-1/2}\| \|A(I-D^2)^{-1}A^\top\| \|\Sigma_x^{-1/2}\| \leq \frac{1}{1-\rho_1^2} m\;,
\end{equation}
and similarly $\| \Omega_y \| \leq \frac{1}{1-\rho_1^2} m$.
Next by (\ref{eq:xy-equivs}) we have
\begin{equation}
    \|\Sigxy\| \leq \|\Sigma_x^{1/2}\| \|A D B^\top\| \|\Sigma_y^{1/2}\| \leq \rho_1 M\;,
\end{equation}
and by (\ref{eq:omegaxy}) we have
\begin{equation}\label{eq:omegaxy-bound}
    \|\Omega_{xy}\| \leq \|\Sigma_x^{-1/2}\|  \left\|A \frac{D}{(I - D^2)} B^\top\right\| \|\Sigma_y^{-1/2}\| \leq \frac{\rho_1}{1-\rho_1^2} m = \gamma_1 m\;,
\end{equation}
where we define $\gamma_1 = \frac{\rho_1}{1-\rho_1^2}$ for later convenience.

\subsubsection{Some useful lemmas}
\begin{lemma}\label{lem:full-op-norm}
    $ \|\Sigma \| \leq 2 M $.
\end{lemma}
\begin{proof}
    Consider a vector $\begin{pmatrix} u \\ v \end{pmatrix} \in \R^{p+q}$. Then 
    \begin{equation*}
    \begin{pmatrix} u \\ v \end{pmatrix}^\top \Sigma \begin{pmatrix} u \\ v \end{pmatrix}
    \leq \begin{pmatrix} u \\ v \end{pmatrix}^\top \Sigma \begin{pmatrix} u \\ v \end{pmatrix} + \begin{pmatrix} u \\ \shortminus v \end{pmatrix}^\top \Sigma \begin{pmatrix} u \\ \shortminus v \end{pmatrix}
    = 2\left( u^\top \Sigma_x u + v^\top \Sigma_y v\right) \leq 2 M \left\|\begin{pmatrix} u \\ v \end{pmatrix}\right\|^2\;. \hfill\qedhere
    \end{equation*}
\end{proof}

\begin{lemma}\label{lem:restriction}
Restriction is 1-Lipschitz with respect to the operator norm; in particular, $\|\hat{\Sigma}_{xy} - \Sigxy\| \leq \|\hat{\Sigma} - \Sigma \|$.
\end{lemma}
\begin{proof}
    Consider an arbitrary block-structured matrix $\begin{pmatrix} A & B \\ C & D \end{pmatrix} \in \R^{(p+q)\times(p+q)}$ and $u\in\R^p$. We have
    \begin{equation*}
        \|Cu \| \leq \left\|\begin{pmatrix} A \\ C \end{pmatrix} u\right\| = \left\| \begin{pmatrix} A & B \\ C & D \end{pmatrix} \begin{pmatrix} u \\ 0 \end{pmatrix}\right\| \leq \left\|\begin{pmatrix} A & B \\ C & D \end{pmatrix}\right\| \|u\|\; .
    \end{equation*}
    So indeed we have $\|C\| \leq \left\|\begin{pmatrix} A & B \\ C & D \end{pmatrix}\right\|$ and the result follows by linearity of restriction.
\end{proof}

\begin{lemma}\label{lem:inv-lips}
Let $A,B \in \R^{n \times n}$  with $\|B^{-1}\| \leq M$ and $\|A - B\| \leq \epsilon$. Then $\|A^{-1} - B^{-1}\| \leq \frac{M^2 \epsilon}{1-M\epsilon}$.
Furthermore $\|A^{-1}\| \leq \frac{M}{1-M\epsilon}$.
\end{lemma}
\begin{proof}
    Write $\delta = \|A^{-1} - B^{-1}\|$. Then
    \begin{equation*}
    \begin{split}
        \delta = \|A^{-1} - B^{-1}\| = \|A^{-1}(B-A)B^{-1}\| \leq \|A^{-1}\| &\|A-B\| \|B^{-1}\| \\
        \leq (\|B^{-1}\|+\|A^{-1} - B^{-1}\|) & \|A-B\| \|B^{-1}\| \leq (M + \delta) \epsilon M\;. 
        \end{split}
    \end{equation*}
    So rearranging gives $\delta \leq \frac{M^2 \epsilon}{1-M\epsilon}$ and so also 
    \begin{equation*}
        \|A^{-1}\| \leq \|B^{-1}\|+\|A^{-1} - B^{-1}\| \leq M + \delta \leq \frac{M(1-M\epsilon) + M^2 \epsilon}{(1-M\epsilon)} = \frac{M}{1-M\epsilon}\;. \hfill\qedhere
    \end{equation*}
\end{proof}

\begin{lemma}\label{lem:sqrt-lips}
    Let $A,B \in \R^{n \times n}$ be positive definite then 
    \begin{equation}
        \| \sqrt{A} - \sqrt{B} \| \leq \frac{\|A-B\|}{\lambda_\text{min}(\sqrt{A}+\sqrt{B})}
        \leq \frac{\|A-B\|}{\sqrt{\lambda_\text{min}(A)}+\sqrt{\lambda_\text{min}(B)}}\; .
    \end{equation}
\end{lemma}
\begin{proof}
    Let $x\in\R^n$ be a unit-norm eigenvector of $\sqrt{A} - \sqrt{B}$ with eigenvalue $\|\sqrt{A} - \sqrt{B}\| = |\mu|$; by switching $A,B$, we can assume WLOG that $\mu\geq 0$ then
    \begin{align*}
        \|A-B\| &\geq 
        x^\top (A-B) x \\ 
        &= x^\top\{ (\sqrt{A} - \sqrt{B})\sqrt{A} + (\sqrt{A}-\sqrt{B})\sqrt{B}\} x \\
        &= \mu x^\top (\sqrt{A}+\sqrt{B}) x \\
        & \geq \mu \lambda_\text{min}(\sqrt{A}+\sqrt{B}) \\
        & \geq \mu \left(\lambda_\text{min}(\sqrt{A})+
        \lambda_\text{min}(\sqrt{B}) \right)\;,
    \end{align*}
    where the first equality follows from rearranging, and the final two inequalities follow from the variational characterisation of eigenvalues.
    \end{proof}

\subsubsection{Further operator norm bounds}
For clarity we shall put these in bullet form:
\begin{itemize}
    \item Recall that \Cref{eq:psi-x} gave $\|\Omega_x^{-1}\| = \|\Psi_x\| \leq M$; and because restriction is Lipschitz we have $\|\hat{\Omega}_x - \Omega_x\| \leq \|\hat{\Omega}-\Omega\| \defeq \epsilon$. Applying these to Lemma \ref{lem:inv-lips} gives
    \begin{equation}\label{eq:mod-hat-om-x}
        \|\hat{\Omega}_x\| \leq \frac{M}{1-M \epsilon} 
        ,\quad\text{and}\quad
        \|\hat{\Omega}_x^{-1} - \Omega_x^{-1}\| \leq \frac{M^2}{1-M \epsilon} \epsilon
    \end{equation}
    and similarly
    \begin{equation}\label{eq:mod-hat-om-y}
        \|\hat{\Omega}_y\| \leq \frac{M}{1-M \epsilon},\quad\text{and}\quad
        \|\hat{\Omega}_y^{-1} - \Omega_y^{-1}\| \leq \frac{M^2}{1-M \epsilon} \epsilon \; .
    \end{equation}
    
    \item By combining the triangle inequality with Lemma \ref{lem:restriction} and (\ref{eq:omegaxy-bound}) we obtain
    \begin{equation}
        \|\hat{\Omega}_{xy} \| \leq \epsilon + \|\Omega_{xy}\| \leq \epsilon + \gamma_1 m \;.
    \end{equation}
    
    \item Combining Lemmas \ref{lem:full-op-norm} and \ref{lem:inv-lips} gives
    \begin{equation}
        \|\hat{\Sigma} - \Sigma\| \leq \frac{(2M)^2}{1-2M\epsilon}\epsilon =\vcentcolon \kappa \epsilon \;.
    \end{equation}
    where we defined $\kappa = \kappa(\epsilon) = \frac{(2M)^2}{1-2M\epsilon} $. Then by Lemma \ref{lem:restriction} also \begin{equation}
        \|\hat{\Sigma}_{xy} - \Sigxy\| \leq \kappa \epsilon, \quad\text{and}\quad \norm{\hat{\Sigma}_{xx} - \Sigxx} \leq \kappa \epsilon \;.
    \end{equation}
\end{itemize}

\subsubsection{Estimating \texorpdfstring{$S_x^{1/2}$}{sqrt Sx}}
Finally we can start to approach bounding our target quantity.
\newcommand{\hOm}{\hat{\Omega}} 
\begin{align*}
    \|\hat{S}_x - S_x \|
    &= \| (\hat{\Omega}_x - \Omega_x) - (\hat{\Omega}_{xy}\hat{\Omega}_{y}^{-1} \hat{\Omega}_{yx} - \Omega_{xy}\Omega_{y}^{-1}\Omega_{yx}) \| \\[6pt]
    &\leq \|\hat{\Omega}_x - \Omega_x\| 
    + \|(\hat{\Omega}_{xy} - \Omega_{xy})\hOm_y^{-1} \hOm_{yx} + \Omega_{xy}(\hOm_y^{-1}-\Omega_y^{-1})\hOm_{yx} + \Omega_{yx}\Omega_{y}^{-1}(\hOm_{yx} - \Omega_{yx})\| \\[6pt]
    &\leq \epsilon + \epsilon \cdot \frac{M}{1-M\epsilon} \cdot (\epsilon + \gamma_1 m) + \gamma_1 m \cdot \frac{M^2 \epsilon}{1-M\epsilon} \cdot (\epsilon + \gamma_1 m) + \gamma_1 m \cdot \frac{M}{1-M\epsilon} \cdot \epsilon \\[6pt]
    &= \epsilon \Bigl\{ 1 + \frac{\gamma_1 M m}{1-M\epsilon} (1+\gamma_1 m M + 1)\Bigr\} + \epsilon^2 \frac{M}{1-M\epsilon} (1+\gamma_1 M m) \;.
\end{align*}
Now note that $\epsilon \leq \frac{1}{4M}$ means that $\frac{1}{1-M\epsilon}\leq \frac{4}{3}$ and so if we write $\tau_1 = \gamma_1 M m$ we can simplify this bound to
\begin{align*}
    \|\hat{S}_x - S_x \|
    &\leq \epsilon \cdot \Bigl\{
    1 + \frac{1}{1-M\epsilon} \left( 2 \tau_1 + \tau_1^2 + \epsilon M (1 + \tau_1)\right)
    \Bigr\} \\
    &\leq \epsilon \cdot \frac{4}{3} \Bigl\{
    \frac{3}{4} + 2 \tau_1 + \tau_1^2 + \tfrac{1}{4} (1 + \tau_1)
    \Bigr\} \\
    &\leq \epsilon \cdot \frac{4}{3} \left(\tau_1 + \tfrac{9}{8}\right)^2 =\vcentcolon L \epsilon \;.
\end{align*}
where we defined $L \defeq \frac{4}{3} \left(\tau_1 + \tfrac{9}{8}\right)^2$. 
Next apply Lemma \ref{lem:sqrt-lips}, and note $1/\lambda_\text{min}(S_x) = \norm{S_x^{-1}} = \norm{\Sigxx} \leq M$ to get
\begin{equation}\label{eq:s-to-half-error}
    \|\hat{S}_x^{1/2} - S_x^{1/2}\| \leq \frac{\| \hat{S}_x - S_x\|}{\sqrt{\lambda_\text{min}(S_x)}} \leq  M^\half L \epsilon \;. 
\end{equation}
Therefore we can take $C_2(M,m,\rho_1) = M^\half L$.
We similarly obtain $\|\hat{S}_y^{1/2} - S_y^{1/2} \| \leq M^\half L \epsilon$.

Finally, we also note that by the triangle inequality we have $\| \hat{S}_x\| \leq \| S_x \| + \| \hat{S}_x - S_x\| \leq m + L\epsilon$. Then by completing the square
\begin{equation}
    \| \hat{S}_x^{1/2} \| \leq \sqrt{m + L \epsilon} \leq \sqrt{m + L\epsilon + \frac{L^2\epsilon^2}{4m}} = \sqrt{m} + \frac{L\epsilon}{2\sqrt{m}}
\end{equation}
and again the same bound holds for $\|\hat{S}_y^{1/2}\|$ by symmetry. Note, this is substantially tighter than would have been obtained from applying the triangle inequality via \cref{eq:s-to-half-error}, since $1/m \leq \lambda_\text{min}(\Sigxx) \leq \lambda_\text{max}(\Sigxx)\leq M$ and in general we expect $1/m \ll M$.

\subsubsection{Final target bound}
Finally plug this in to the expression for $T$ to obtain:
\begin{align*}
    \|\hat{T}-T\|
    &= \|(\hat{S}_x^{1/2} - S_x^{1/2})\hat{\Sigma}_{xy}\hat{S}_{y}^{1/2} + S_x^{1/2}(\hat{\Sigma}_{xy} -\Sigxy) \hat{S}_y^{1/2} + S_x^{1/2}\Sigxy(\hat{S}_y^{1/2}-S_y^{1/2})\| \\[6pt]
    &\leq M L  \epsilon \cdot (\rho_1 M + \kappa \epsilon) \cdot (\sqrt{m} +  \frac{L\epsilon}{2\sqrt{m}}) 
    + \sqrt{m} \cdot \kappa\epsilon \cdot (\sqrt{m} + \frac{L\epsilon}{2\sqrt{m}}) 
    + \sqrt{m} \cdot \rho_1 M \cdot M L  \epsilon \\[6pt]
    &= \epsilon \biggl\{ M L  \rho_1 M \sqrt{m} + m \kappa + \sqrt{m} \rho_1 M M L  \biggl\}  \\
    &\qquad + \epsilon^2 \biggl\{ M L  ( \rho_1 M \frac{L}{2\sqrt{m}} + \sqrt{m} \kappa ) + \sqrt{m} \kappa \frac{L}{2\sqrt{m}}  \biggr\}\\
    &\qquad + \epsilon^3 \cdot M L  \kappa \frac{L}{2\sqrt{m}}  \\[6pt]
    &\leq \epsilon \biggl\{ 2 \rho_1 \sqrt{m} M^2 L  + m \kappa \biggl\}  \\
    &\qquad + \epsilon \biggl\{ \rho_1 \frac{ML^2}{8\sqrt{m}} + \tfrac{1}{4} L \kappa \sqrt{m} + \tfrac{1}{2}L \cdot 2M ) + \sqrt{m} \kappa \frac{L}{2\sqrt{m}}  \biggr\}\\
    &\qquad + \epsilon \cdot \frac{ML^2}{8\sqrt{m}}  \\
    &\leq \epsilon \biggl\{ 2 \rho_1 \sqrt{m} M^2 L  + 8 m M^2 + \rho_1 \frac{ML^2}{8\sqrt{m}} +2 L M^2 \sqrt{m} + LM + 4 M^2 L
    \biggr\} \\
    &=\vcentcolon C_1 \epsilon \;.
\end{align*}
where we note $C_1$ does indeed only depend on $M,m,\text{ and } \rho_1$ and to leading order is $M^4 m^{5/2}$. Some of the results of \citet{gao_sparse_2016} rely on similar arguments; above we strive to make more explicit the high-order polynomial dependence of the Lipschitz constant on condition numbers $M$ and $m$.
\end{proof}

\subsection{Variate space bounds}\label{sec:convert-to-variate-bounds}
\subsubsection{Davis-Kahan variant}
The below is a special case of Theorem 4 from \citet{yu_useful_2015}; we change to our notation, only consider finding the top-$K$ subspace, and only consider operator norm bounds.
\begin{theorem}\label{thm:samworth}
Let $T, \hat{T} \in \mathbb{R}^{p \times q}$ have singular values $\rho_{1} \geq \ldots \geq \rho_{\min (p, q)}$ and $\hat{\rho}_{1} \geq \ldots \geq$ $\hat{\rho}_{\min (p, q)}$ respectively. 
Fix $K < \operatorname{rank}(T)$ and assume that $\rho_{K}^{2}-\rho_{K+1}^{2}>0$.
Let $V=\left(v_{1}, v_{2}, \ldots, v_{K}\right) \in \mathbb{R}^{q \times K}$ and $\hat{V}=\left(\hat{v}_{1}, \hat{v}_{2}, \ldots, \hat{v}_{K}\right) \in \mathbb{R}^{q \times K}$ have orthonormal columns satisfying $T v_{j}=\rho_{j} u_{j}$ and $\hat{T} \hat{v}_{j}=\hat{\rho}_{j} \hat{u}_{j}$ for $j=1,2, \ldots, K .$ Then
$$
\|\sin \Theta(\hat{V}, V)\|_{\mathrm{F}} 
\leq \frac{2 K^{1 / 2} \left(2 \rho_{1}+\|\hat{T}-T\|\right) \|\hat{T}-T\|}{\rho_{K}^{2}-\rho_{K+1}^{2}}\;.
$$
\end{theorem}
\vspace{10pt}

To apply this result, we will take $A=(\alpha_1,\ldots,\alpha_K) \in \mathbb{R}^{q\times K}$ and $\hat{A}=(\hat{\alpha}_1,\ldots,\hat{\alpha}_K) \in \mathbb{R}^{q \times K}$ to be the top-$K$ left singular vectors of $T,\hat{T}$ respectively, as in \Cref{app:matrix-SVD-for-CCA}. 
However, this $\sin\Theta$ loss between $\hat{A}$ and $A$ is not particularly intuitive.
We need a couple more steps to relate this to the variate space loss that we used in the main text.

\subsubsection{Converting to variate space losses}
By applying \Cref{lem:arbitrary-registration} and noting that the true canonical variates are such that $X^\top {U_{K}^*}$ has orthonormal components, we can rewrite the loss we want to control as
\begin{align}\label{eq:sin-theta-variate-to-bound}
    L(\hat{U}_{K}, U^*_{K}) \defeq \sin^2 \Theta(X^\top \hat{U}_{K}, X^\top U_{K}^*) = \inf_{W\in \R^{K\times K}} \mathbb{E}\|X^\top \hat{U}_{K} W - X^\top U_{K}^*\|^2 \;.
\end{align}
We now compare this to the loss from \citet{gao_sparse_2016}, which in our notation is defined as
\begin{equation}\label{eq:gao-loss-def}
    L^\text{Gao}(U_{K},\hat{U}_{K}) = \inf_{W\in \mathcal{O}(K)} \mathbb{E}\|X^\top \hat{U}_{K} W - X^\top U_{K}^*\|^2 \; .
\end{equation}
We see that $L(\hat{U}_{K}, U^*_{K}) \leq L^\text{Gao}(\hat{U}_{K}, U^*_{K})$ because this second infimum is over $\mathcal{O}(K)$, which is a subset of $\R^{K\times K}$.
We therefore bound this stronger loss $L^\text{Gao}$ and deduce the same bound must hold for our loss $L$.

We now work on controlling the term on the right hand side of \Cref{eq:gao-loss-def}.
From now we omit the subscript $K$ terms to declutter the notation.
For orthogonal $W$ we have
\begin{align*}
    \left(\mathbb{E}\|X^\top \hat{U} W - X^\top U^*\|^2 \right)^\half 
    &= \left(\mathbb{E} \operatorname{tr} (W^\top\hat{U}^\top - U^\top) X X^\top (\hat{U} W - U)\right)^\half  \\
    &= \left(\operatorname{tr} (\hat{U} W - U)^\top \Sigxx (\hat{U} W - U) \right)^\half \\
    &= \|\Sigxx^{1/2} (\hat{U} W - U) \|_F \\
    &\leq \| \hat{\Sigma}_{xx}^{1/2}\hat{U} W - \Sigxx^{1/2} U \|_F + \| \hat{\Sigma}_{xx}^{1/2}\hat{U} W - \Sigxx^{1/2} \hat{U} W \|_F \\
    &= \| \hat{\Sigma}_{xx}^{1/2}\hat{U} W - \Sigxx^{1/2} U \|_F + \| (\hat{\Sigma}_{xx}^{1/2} - \Sigxx^{1/2}) \hat{U} \|_F \;,
\end{align*}
where we first expand out using the trace-trick, then apply the triangle inequality, and note that the orthogonal $W$ can be dropped from the final Frobenius norm term.
We can further bound this final term by
\begin{align*}
    \| (\hat{\Sigma}_{xx}^{1/2} - \Sigxx^{1/2}) \hat{U} \|_F 
    &= \|(I - \Sigma_x^{1/2} \hat{\Sigma}_{xx}^{-1/2}) \hat{\Sigma}_{xx}^{1/2} \hat{U} \|_F \\
    &= \|(I - \Sigma_x^{1/2} \hat{\Sigma}_{xx}^{-1/2}) \hat{A} \|_F \\
    &\leq K^\half \,  \|I - \Sigma_x^{1/2} \hat{\Sigma}_{xx}^{-1/2} \|_\textrm{op}  \\
    &= K^\half \, \| \Sigma_x^{1/2}(\Sigxx^{-1/2} - \hat{\Sigma}_{xx}^{-1/2})\|_\textrm{op}  \\
    &\leq K^\half M^{1/2} \|\Sigxx^{-1/2} - \hat{\Sigma}_{xx}^{-1/2}\|_\textrm{op} \;, 
\end{align*}
where for the first inequality we use that $\hat{A}$ has orthonormal columns by construction.

Before we put everything together, we state the following standard lemma, e.g. see \citet{gao_minimax_2015, stewart_matrix_1990}.
    \begin{lemma}\label{lem:orthog-sin-theta}
        Suppose $A,\hat{A} \in \R^{p\times d}$, with $d<p$, have orthonormal columns. Then
        \begin{equation*}
            \inf_{W\in \mathcal{O}(K)} \|A - \hat{A} W \|_F \leq \sqrt{2} \| \sin\Theta(\hat{A},A)\|_F \;.
        \end{equation*}
    \end{lemma}

Now we can combining the previous bounds and apply \Cref{lem:orthog-sin-theta} to give
\begin{align}
    L^\text{Gao}(U,\hat{U})^{1/2}
    &\leq \inf_{W\in \mathcal{O}(K)} \left(\mathbb{E}\|X^\top \hat{U} W - X^\top U^*\|^2 \right)^\half \nonumber\\
    & \leq \inf_{W\in \mathcal{O}(K)} \Bigl\{\| \hat{\Sigma}_{xx}^{1/2}\hat{U} W - \Sigxx^{1/2} U \|_F + K^\half M^{1/2} \|\Sigxx^{-1/2} - \hat{\Sigma}_{xx}^{-1/2}\|\Bigr\} \nonumber\\ 
    &\leq \sqrt{2} \| \sin\Theta(\hat{A},A)\|_F + K^\half M^{1/2} \|\Sigxx^{-1/2} - \hat{\Sigma}_{xx}^{-1/2}\|_\textrm{op}\;.  \label{eq:gao-loss-bound-sin-theta-op-norm}
\end{align}

This is now in a format that we can handle with our results from \Cref{app:lipschitz-plug-in}.
Recall also that the same bound must apply to $L$, since $L^\text{Gao}$ bounds $L$.

\subsubsection{Combining results}
\begin{proposition}\label{prop:variate-bound-from-precision-bound-combined-app}
    Assume Assumption \ref{ass:well-conditioned-within-view}$(M,m)$.   
    Let $C_1$ be as in \Cref{prop:lipschitz-plug-in-app-longer}.
    Then there exists a constant $C_3 \in (0, \infty)$ only depending on $M,m,\rho_1$ such that whenever 
    \begin{equation}
        \epsilon \defeq \|\Omega - \hat{\Omega}\| \leq \min\left(\frac{1}{4M}, \frac{2 \rho_1}{C_1}\right)
    \end{equation}
    we have 
    \begin{equation}
       \norm{\sin \Theta(X^\top \hat{U}_{K}, X^\top U_{K}^*)}_F 
       \leq L^\text{Gao}(\hat{U}_{K}, U_{K}^*)^\half  \leq C_3 \frac{K^\half}{\rho_K^2 - \rho_{K+1}^2} \epsilon \;.
    \end{equation}
\end{proposition}
\begin{proof}
    Suppose
    \begin{equation*}
        \epsilon \defeq \|\Omega - \hat{\Omega}\| \leq \min\left(\frac{1}{4M}, \frac{2 \rho_1}{C_1}\right)
    \end{equation*}
    as in the statement of this result.
    Then by \Cref{prop:lipschitz-plug-in-app-longer}, there exist $C_1, C_2$ such that 
    \begin{equation*}
            \norm{\hat{T} - T}_\textrm{op} \leq C_1 \epsilon,
            \quad \text{and} \quad
            \norm{\hat{\Sigma}_{xx} - \Sigxx}_\textrm{op} \leq C_2 \epsilon\; .
    \end{equation*}

    Applying \Cref{thm:samworth} and using that $C_1 \epsilon \leq 2 \rho_1$ then yields
    \begin{align*}
        \|\sin \Theta(\hat{A}, A)\|_{\mathrm{F}} 
        \leq \frac{2 K^{1 / 2} \left(2 \rho_{1}+\|\hat{T}-T\|\right) \|\hat{T}-T\|}{\rho_{K}^{2}-\rho_{K+1}^{2}}
        \leq \frac{8 K^\half \rho_1}{\rho_{K}^{2}-\rho_{K+1}^{2}} C_1 \epsilon \; .
    \end{align*}

    Plugging this into \Cref{eq:gao-loss-bound-sin-theta-op-norm} and using that $\frac{1}{\rho_{K}^{2}-\rho_{K+1}^{2}} \geq 1$ gives
    \begin{align*}
        L^\text{Gao}(U,\hat{U})^{1/2}
        &\leq \sqrt{2} \| \sin\Theta(\hat{A},A)\|_F + K^\half M^\half \|\Sigxx^{-1/2} - \hat{\Sigma}_{xx}^{-1/2}\| \\
        &\leq \frac{8 \sqrt{2} K^\half \rho_1}{\rho_{K}^{2}-\rho_{K+1}^{2}} C_1 \epsilon + K^\half M^\half C_2 \epsilon  \\
        &\leq  (8\sqrt{2} \rho_1 C_1 + M^\half C_2 ) \frac{K^\half}{\rho_{K}^{2}-\rho_{K+1}^{2}} \epsilon \;,
    \end{align*}
    and so the desired bound holds by taking $C_3 = 8\sqrt{2} \rho_1 C_1 + M^\half C_2$.
\end{proof}

Note that applying instead to $B=(\beta_1,\ldots,\beta_K),\hat{B}=(\hat{\beta}_1,\ldots,\hat{\beta}_K)$ gives us an identical bound for the $Y$-direction estimates. These two bounds are identical due to the symmetry of Theorem \ref{thm:samworth}; we expect a stronger asymmetric statement could be attained from the results of \citet{cai_rate-optimal_2020}, but leave further exploration for future work. 

\subsection{Combining with existing bounds for graphical lasso}
    \subsubsection{Structural assumptions from \texorpdfstring{\citet{ravikumar_high-dimensional_2011}}{Ravikumar et al., 2011}}
    In the main text, we summarised the structural assumptions from \citet{ravikumar_high-dimensional_2011} in Assumption \ref{ass:glasso-structural-assumptions}, but did not define the terms formally.
    We now give formal definitions for each of the terms.
    We first restate the set of assumptions for reference.

    \structuralAssumptionsGlassoRavikumar*
    
     \textbf{$\ell_\infty$-operator-norm bound:}
     \begin{align*}
         \kappa_{\Sigma^{}} \defeq \|\Sigma^{}\|_{\mathrm{op}, \ell_\infty \to \ell_\infty} = \max_{i=1,\dots,\bar{p}} \sum_{j=1}^{\bar{p}} |\Sigma_{ij}^{}|\; .
     \end{align*}
     
     \textbf{Strict convexity:}
     Let $S(\Omega^{})$ be the set of nonzero index pairs for $\Omega^{}$, augmented with the index pairs of the diagonal entries. Then we define
     \begin{align*}         
         \Gamma^{}_{SS} &\defeq [{\Omega^{}}^{-1} \otimes {\Omega^{}}^{-1}]_{SS} \in \R^{(s+\bar{p})\times(s+\bar{p})} \;, \\
         \kappa_{\Gamma^{}} &\defeq \| \left(\Gamma^{}_{SS}\right)^{-1}\|_\mathrm{op} \; .
     \end{align*}    

     \textbf{Mutual incoherence:}
     There exists some $\alpha \in (0,1]$ such that 
     \begin{equation*}
         \max_{e\in S^c} \|\Gamma^{}_{eS}\left(\Gamma^{}_{SS}\right)^{-1}\|_1 \leq (1-\alpha) \;.
     \end{equation*}
     This may also be referred to as an \textit{irrepresentability} condition.

    \subsubsection{Adapting results from \texorpdfstring{\citet{ravikumar_high-dimensional_2011}}{Ravikumar et al., 2011}}
    The main result in \citet{ravikumar_high-dimensional_2011} is deliberately very general: this covers a broad class of distributions for $Z$ with a variety of different tail conditions, and bounds the elementwise sup-norm $\norm{\hat{\Omega} - \Omega}_\infty$.
    For our purposes, it is sufficient to consider exponential-type tails (original random variables are sub-Gaussian), and to obtain bounds on the operator norm $\norm{\hat{\Omega} - \Omega}_\text{op}$.
    The authors provide a number of helpful corollaries to help the reader obtain bounds in similar, specific situations of interest.
    In particular, Corollary 1 translates their general result to the situation of exponential tails.
    Then Corollary 3 translates their element-wise-sup-norm bound to an operator norm bound.

    By combining the analysis from these two special cases, we can read off the most special case of interest to us.

    \begin{proposition}[Synthesis of results from \citet{ravikumar_high-dimensional_2011}]\label{prop:ravikumar-synthesis-detail}
        Let Assumption \ref{ass:glasso-structural-assumptions}$\,(\kappa_{\Gamma^{}}, \kappa_{\Sigma^{}}, \alpha)$ hold.
        Suppose the variables $Z_i/\sqrt{\Sigma^{}_{ii}}$ are sub-Gaussian with common parameter $\sigma$ and the samples are drawn independently.
        Let
        \begin{align*}
            \gamma_* = 8\sqrt{2}(1+4\sigma^2) \max_{i}(\Sigma^{}_{ii})\;.
        \end{align*}
        Let $\tau$ be a user-specified parameter, which determines the $\ell_1$ penalty parameter via
        \begin{align*}
            \lambda_N 
            = \frac{8 \gamma_*}{\alpha} \sqrt{\frac{\tau \log \bar{p} + \log 4}{N}} \; .
        \end{align*}
        
        Suppose that the sample size $n$ satisfies the bound
        \begin{align*}
            N > \left\{6 \gamma_* \max\left\{\kappa_{\Sigma^{}} \kappa_{\Gamma^{}},\kappa_{\Sigma^{}}^3 \kappa_{\Gamma^{}}^2\right\} \right\} d^2 \left(1 + \frac{8}{\alpha}\right) (\tau \log \bar{p} + \log 4) \;.
        \end{align*}
        
        Then with probability greater than $1 - \bar{p}^{2 - \tau}$, the estimate $\hat{\Omega}$ using penalty parameter $\lambda_N$ satisfies the bound
        \begin{align*}
            \norm{\hat{\Omega} - \Omega^{}}_\text{op} 
            \leq
            \left\{2 \gamma_* \left(1 + \frac{8}{\alpha}\right) \kappa_{\Gamma^{}}\right\}
            \sqrt{\frac{\tau\log \bar{p} + \log 4}{N}} \; \min\{\sqrt{s+\bar{p}},d\} \;.
        \end{align*}        
    \end{proposition}

    \begin{proposition}[Simplified version of \Cref{prop:ravikumar-synthesis-detail}]\label{prop:ravikumar-synthesis-simplified}
        Let Assumption \ref{ass:glasso-structural-assumptions}$\,(\kappa_{\Gamma^{}}, \kappa_{\Sigma^{}}, \alpha)$ hold.
        Suppose the variables $Z_i/\sqrt{\Sigma^{}_{ii}}$ are sub-Gaussian with common parameter $\sigma$ and the samples are drawn independently.
        Let $\tau \geq 2$ be a user-specified parameter, and $\lambda_N$ be an appropriate penalty parameter determined by $\tau$. 
        Suppose $p\geq 2$.
        Then there exist constants $C_4, C_5$ only depending on $\left(\kappa_{\Gamma^{}}, \kappa_{\Sigma^{}}, \alpha, \sigma, \max_{i}(\Sigma_{ii}^{}) \right)$
        such that when the sample size $N$ satisfies
        \begin{align*}
            N > C_4 \, d^2 \tau \log \bar{p}
        \end{align*}
        we have the bound
        \begin{align*}
            \norm{\hat{\Omega} - \Omega^{}}_\text{op} 
            \leq
            C_5  \sqrt{\frac{\tau \min\{ {s+\bar{p}} ,d^2\} \log \bar{p}}{N}} \;.
        \end{align*}        
    \end{proposition}

    \subsubsection{Pulling everything together}
    By combining \Cref{prop:ravikumar-synthesis-simplified} with \Cref{prop:variate-bound-from-precision-bound-combined-app} we arrive at our result from the main text.
    \sketchGraphicalCCAGuarantee*

\newpage
\section{Background continued}\label{app:cca-broader-background}

CCA has been a central concept in multi-view learning and multivariate analysis since its proposal in \citet{hotelling_relations_1936}. There has been a vast amount of work around this concept, and textbook treatments of the classical methodology are available in \citet{anderson_introduction_2003}, for example.

\Cref{app:existing-further} gives a more thorough literature review on regularised CCA, which is important to place our approach in the context of existing methods. \citet{solari_sparse_2019} and \citet{uurtio_tutorial_2017} contain solid literature reviews on CCA, which touch on extensions and related methods. \Cref{app:alternative-formulations-and-extensions} complements these existing surveys by summarising selected recent developments; we emphasise work that has been particularly helpful for developing our intuitions for understanding CCA.

\subsection{High-dimensional CCA: further background}\label{app:existing-further}

\subsubsection{Unregularised CCA: assorted results}
    \Citet{ma_subspace_2020} provide non-asymptotic bounds for the (classical) sample CCA estimator in the $n > p+q$ regime. Their loss functions use $\sin\Theta$ distances in variate space, and are well motivated by some insightful discussion. We took inspiration from much of their presentation.
    
\subsubsection{sPLS: selected extensions}
    A number of extensions and variants to the \wit{} estimator we define have been proposed; these include considering more general penalties such as fused or group lasso penalties \citep{witten_penalized_2009}, extensions to data with more than two views, and a computational short-cut \citep{solari_sparse_2019}. However, fundamentally, all these variants differ from the standard CCA problem by an implicit assumption of identity-block covariance, as discussed in \Cref{sec:existing}.

\subsubsection{Sparse CCA: methodological suggestions}
    We now summarise a variety of recent works proposing sparse CCA methods building upon or distinct to the \suo{} \citep{suo_sparse_2017} method we considered in the main text.
    Though some of these works contain theoretical guarantees, the main contributions appear methodological.

\begin{itemize}
    \item \textbf{Earlier work:}
    The first work we are aware of not to require the identity-block-covariance assumption made by sPLS was that of \citet{hardoon_sparse_2011}; this uses a so-called ML-primal-dual framework, but requires a different restrictive assumption: namely positivity for one of the canonical directions.
    
    \item \textbf{Implementation of \citet{suo_sparse_2017}:}
    For each step of the ACS procedure, \citet{suo_sparse_2017} originally suggested to solve the resulting convex optimsation problem with a certain linearised alternating direction method of multipliers (LADMM) \citep{parikh_proximal_2014}. The more recent work \citet{xiu_tssnalm_2020} proposes an algorithm that still uses ACS but replaces LADMM with a semi-smooth Newton method; they claim this gives significantly increased performance.
    
    \item \textbf{Iterative Penalised Least Squares:}
    Another natural approach makes use of the alternating-projection formulation of CCA \citep{anderson_introduction_2003}; this was suggested as an exercise in \citet{hastie_statistical_2015} and analysed in \citet{mai_iterative_2019}, which we shall refer to as Iterative Penalised Least Squares (IPLS). Though relatively fast and easy to implement, this requires deflation to encourage orthogonality of successive estimated directions; we found this deflation caused problems with later directions when we tried to apply the method in high-dimensional settings.
    
    \item \textbf{Linear Programming:}
    \Citet{safo_integrative_2018} and its sister work \citet{safo_sparse_2018}, both use a ridge-regularised covariance estimate\footnote{They alternatively suggest a restrictive identity block-covariance approximation but this then resembles a PLS algorithm, so we focus on the ridge case.} to obtain initial canonical direction estimates, then find sparse vectors which are `near' these initial estimates. Formally, if the ridge estimates are $(\tilde{u}_k,\tilde{v}_k,\tilde{\rho}_k)$ then we define
    \begin{align*}
        \hat{u}_k^\text{DS} \in \argmin_{u \in \R^p} \mathcal{P}(u,\tau_x) &\quad\text{subject to}\quad \|\Cxy \tilde{v}_k - \tilde{\rho}_k \Cxx u \|_\infty \leq \tau_x \;, \\
        \hat{v}_k^\text{DS} \in \argmin_{v \in \R^p} \mathcal{P}(v,\tau_y) &\quad\text{subject to}\quad \|\Cyx \tilde{u}_k - \tilde{\rho}_k \Cyy v \|_\infty \leq \tau_y \;, 
    \end{align*}
    where $\mathcal{P}$ is some penalty function which may incorporate prior biological knowledge; for example $\mathcal{P}(u,\tau_x) = \|u\|_1$.
    We can think of this as first estimating the underlying variates with a ridge penalty and then finding sparse directions which give similar loadings.

    \item \textbf{Sparse Generalised Eigenvalue Problem (GEP): } \Citet{tan_sparse_2018} propose a Truncated Rayleigh Flow Method (Rifle) for the more general problem of solving a sample GEP whose eigenvectors are assumed to be sparse. Their convex-optimisation framework closely resembles that of \citet{gao_sparse_2016}. They provide interesting theory which use matrix perturbation results which exploit the geometry of GEPs from \citet{stewart_matrix_1990}.
\end{itemize}

\subsubsection{Sparse CCA: theoretical contributions}
We now summarise some recent works whose main contributions appear theoretical, though they do all also provide methodological suggestions.
They each contain very impressive mathematical analysis, which require assumptions of sparse directions. Following the discussion in the main text, we argue that it is important in practice to be cautious about these assumptions and contrast sCCA estimators with alternatives.

\begin{itemize}    
    \item \textbf{Minimax Rates: } \Citet{gao_minimax_2015} and the follow-up \citet{gao_sparse_2016} contain seminal theoretical work establishing minimax rates for sparse CCA under row-wise sparsity conditions. The former paper proposes a combinatorial estimator that is minimax but computationally intractable. The second paper proposes an estimator (COLAR), which uses a convex optimisation in matrix space to obtain a sparse-and-reduced-rank approximation to the matrix $U \Lambda V^\top$. Unfortunately, their analysis in this case requires extra assumptions (of knowing the true number of canonical directions). Though this is the most principled sparse CCA method proposed, it is computationally expensive (complexity $\mathcal{O}(p+q)^3$), and challenging to implement (2 stage ADMM). We did not spend time implementing the method, and would welcome a contribution to the open source python community and to \ccazoo{} in particular.
    
    \item \textbf{Support Recovery: }  \Citet{laha_support_2021} contains some sophisticated theoretical results on information theoretic and computational hardness of support recovery. They propose a Coordinate Thresholding algorithm for estimation of canonical directions, and a further algorithm, RecoverSupp, for support recovery (again via thresholding). Their method requires access to `good' precision matrix estimates, among other complex assumptions, including sparsity of the canonical directions which is essential to the premise of support recovery. 
    
    \item \textbf{Inference: } \Citet{laha_statistical_2022} provide more sophisticated theory giving asymptotically valid confidence intervals for sCCA. They use preliminary estimators and then a de-biasing step which uses a smooth objective function to characterise the first pair of canonical directions.
    The theory connects to the more general area of de-biased/de-sparsified inference in high-dimensional models. Their method only concerns the first pair of canonical directions, and again assumes access to `good' block precision estimates. The assumption of sparse canonical directions is of course essential to the validity of inference, and the algorithm is complicated to implement and use in practice.   
\end{itemize}

\subsection{Select extensions and alternative formulations of CCA}\label{app:alternative-formulations-and-extensions}
    We briefly highlight four avenues to reformulate or extend CCA; these have led to interesting empirical results, and we found that their different perspectives gave complementary intuition.
    
    Firstly, CCA can be viewed through the lens of `multi-view' learning; it is therefore natural to consider generalising CCA to handle 3 or more datasets. A large number of approaches have been proposed for multiview-CCA \citep{tenenhaus_regularized_2011,rodosthenous_integrating_2020, fu_scalable_2017} but there is no accepted `canonical' formulation and these techniques are much less widely used than standard two-view CCA.

    A second popular way to extend CCA is to consider learning highly correlated representations that are nonlinear transformations of the input variables (rather than linear transformations).
    Historically, this was first achieved by kernel CCA methods, which maximise correlation between the transformed variables over a class of nonlinear transformations in a certain Reproducing Kernel Hilbert Space (RKHS) \citep{bach_kernel_2002}. \cite{}
    More recently, various notions of DeepCCA \citep{andrew_deep_2013,chapman_efficient_2023} instead consider neural network transformations; these DeepCCA algorithms have had empirical success, now enjoy fast implementations \citep{chapman_efficient_2023}, and are closely related to recent advances in self-supervised learning (SSL) \citep{zbontar_barlow_2021}.
    
    Thirdly, there have been a number of probabilistic, or Bayesian approaches or interpretations of CCA; the excellent theoretical work of \citet{bach_probabilistic_2005} laid the foundations for later empirical work, such as \citet{klami_bayesian_2013}. These provide a natural and concrete link between CCA and the field of latent variable models \citep{roweis_unifying_1999}; this is relevant to our narrative of `loadings not weights', as we describe in the following subsection.

    Finally we mention abstract approaches that unify CCA with other dimension reduction methods.
    One interesting approach is that of `duality diagrams', which give a framework encompassing PCA, CCA, discriminant analysis and canonical correspondence analysis; these were developed by a French school of statisticians in the early 1970's. \Citet{holmes_multivariate_2008} gives an excellent (English language) summary of this literature, and provides references to many relevant works; of particular relevance is a multivariate analysis textbook of \cite{lebart_statistique_nodate}, that contains a detailed discussion of assessing the stability of dimension reduction techniques.
    The duality diagram framework is also closely linked to the Generalised Eigenvalue Problem (GEP) formulation of CCA, see \citet{chapman_efficient_2023} for a recent perspective.

\subsection{Probabilistic interpretations of CCA}
\subsubsection{Probabilistic CCA \texorpdfstring{\citep{bach_probabilistic_2005}}{Bach Jordan}}\label{sec:probabilistic-cca}
    We first translate the main probabilistic CCA result from \citet{bach_probabilistic_2005} into our notation and then give further interpretation.
    
        \begin{proposition}\label{prop:bach-jordan}
            Consider the model where $X$ and $Y$ are conditionally independent given a latent variable $Z$, with
        \begin{align*}
            Z &\sim \mathcal{N}(0,I_d),  &\min(p,q) \geq d \geq 1 \\
            X | Z &\sim \mathcal{N}(W_1 Z + \mu_1, \Psi_1),  &W_1 \in \R^{p\times d},\Psi_1 \succcurlyeq 0 \\
            Y | Z &\sim \mathcal{N}(W_2 Z + \mu_2, \Psi_2),  &W_2 \in \R^{p\times d},\Psi_2 \succcurlyeq 0 \;.
        \end{align*}
        Then the MLEs of the parameters $W_1,W_2,\Psi_1,\Psi_2,\mu_1,\mu_2$ are
        \begin{alignat*}{3}
            \hat{W}_1 &= \Cxx \hat{U}_d M_1 \qquad
            &\hat{\Psi}_1 &= \Cxx - \hat{W}_1 \hat{W}_1^\top \qquad
            &\hat{\mu}_1 &= \bar{\X} \\
            \hat{W}_2 &= \Cyy \hat{V}_d M_2 
            &\hat{\Psi}_2 &= \Cyy - \hat{W}_2 \hat{W}_2^\top 
            &\hat{\mu}_2 &= \bar{\Y}
        \end{alignat*}
        where $\hat{U}_d \in \R^{p \times d}, \hat{V}_d \in \R^{q \times d}$ are matrices whose columns are successive sample canonical directions and $M_1,M_2 \in \R^{d\times d}$ are arbitrary matrices with $M_1 M_2 ^\top = \hat{R}, \: \|M_1\|_\textrm{op}  \leq 1, \|M_2\|_\textrm{op} \leq 1$, where $\hat{R}$ is the diagonal matrix of sample canonical correlations.
        \end{proposition}

    The key observation here is that the MLEs $\hat{W}_1,\hat{W}_2$ are matrices of loadings and not weights; this is exactly analogous to the results in \Cref{sec:dimension-reduction-reconstruction}: weights are required to obtain canonical variates from input variables, but the loadings are the right basis with which to reconstruct the input variables from the canonical variates, treating these canonical variates as latent variables. This probabilistic result can be interpreted as a rigorous statement of `optimal' reconstruction, with respect to a natural linear Gaussian latent variable model.

\subsubsection{Mutual Information}\label{app:mutual-information}
    Mutual information provides another rigorous statement of `optimality' for CCA in the Gaussian setting: CCA defines low dimensional linear projections of maximal mutual information.
    Indeed there is a well known analytic formula expressing mutual information between two Gaussian random variables as a function of their canonical correlations \citep{borga_learning_1998}. We give a slightly different proof here for completeness, and to aid intuition.
    This motivates the alternative choice of metric for combining canonical correlations we proposed in \Cref{sec:oracle-correlations}. 
    
    \begin{proposition}\label{prop:mutual-info}
        For $X,Y$ jointly Gaussian with canonical correlations $(\rho_k)_{k=1}^K$, the mutual information
        \begin{equation}
            I(X;Y) = \frac{1}{2} \log\left(\frac{1}{\Pi_{k=1}^K (1-\rho_k^2)}\right)
            = \frac{1}{2} \sum_{k=1}^K \shortminus \log ( 1- \rho_k^2) \;.
        \end{equation}
    \end{proposition}
    \begin{proof}
        We state three well known properties of (continuous) mutual information \citep{cover_elements_2012}:
        \begin{itemize}
            \item Mutual information can be expressed in terms of differential entropies,
            \begin{align}\label{eq:mutual-info-differential-entropy}
                I(X;Y) = h(X) + h(Y) - h(X,Y)\;.
            \end{align}
            
            \item Mutual information is invariant to invertible linear transformations\footnote{In fact also invariant to more general invertible transformations, but this is harder to prove.}, to see this note that for any random vector $X$ and deterministic matrix $A$
            \begin{align*}
                h(A X) = h(X) + \log |A| \;,
            \end{align*}
            where $|\cdot|$ denotes the determinant, and so the $|A|$ terms cancel in \Cref{eq:mutual-info-differential-entropy}
            \item The differential entropy of a Gaussian random variable with covariance matrix $\Sigma$ is
            \begin{align*}
                \tfrac{1}{2} \log \left\{(2 \pi e)^n |\Sigma|\right\} \;.
            \end{align*}
        \end{itemize}

        To conclude, let $U,V$ be full bases of canonical directions for $(X,Y)$, and let $R=U^\top\Sigxy V = \text{diag}(\rho_k)_{k=1}^K$. Then since $U,V$ are invertible we see
        \begin{align*}
            I(X;Y) &= I(U^\top X; V^\top Y) \\
            &= \tfrac{1}{2} \log \left\{(2 \pi e)^p |I_p|\right\} + \tfrac{1}{2} \log\left\{ (2 \pi e)^q |I_q|\right\}
            - \tfrac{1}{2} \log\left\{ (2 \pi e)^{(p+q)} \left| \begin{pmatrix}
                I_p & R \\ R^\top & I_q
            \end{pmatrix}\right|\right\} \\
            &= - \tfrac{1}{2} \log\left\{ \left| \begin{pmatrix}
                I_p & R \\ R^\top & I_q
            \end{pmatrix}\right|\right\} \\
            &= - \tfrac{1}{2} \log \prod_{k=1}^K (1-\rho_k^2) \;,
        \end{align*}
        as claimed.
    \end{proof}
    
    This proof, together with CCA-interlacing \Cref{lem:interlacing-for-cca}, suggest that in the joint Gaussian setting, CCA can be characterised as finding subspaces of maximal mutual information.
    
    Note that we now have two straightforward ways to compute mutual information between Gaussian distributions with known covariance.
    First we can compute the canonical correlations and apply the result \Cref{prop:mutual-info}.
    Alternatively, we can use the intermediate result that for that joint Gaussian $X,Y$ with covariance $\Sigma$, the mutual information is
    \begin{equation}
        I(X;Y) = I^\text{Gauss}(\Sigma) \defeq \frac{1}{2} \log \left(\frac{|\Sigxx| |\Sigyy|}{|\Sigma|} \right) \;.
    \end{equation}

\newpage
\section{Toy synthetic data experiments}\label{app:toy-synthetic-experiments}
In this section, we present two further synthetic data experiments.
These consider two toy data generating mechanisms.
These experiments were produced for an earlier iteration of this manuscript, and use different error metrics to those from the main text.
These are defined by
\begin{align*}
    \texttt{loadu1} = \frac{\norm{\Sigxx \hat{u}_1 - \Sigxx u_1}}{\norm{\Sigxx u_1}}, \quad
    \texttt{weightu1} = \frac{\norm{\hat{u}_1 - u_1}}{\norm{u_1}}, \quad
    \texttt{rho1} = \rho^\text{oracle}(\hat{u}_1, \hat{v}_1) \;.
\end{align*}

\subsection{Data generating mechanisms}
\subsubsection{Canonical Pair model}\label{app:detail-canonical-pair-models}
	First we consider the canonical pair model, motivated by the decomposition of \Cref{app:matrix-SVD-for-CCA}, and first proposed in \citet{chen_sparse_2013}; this is a construction for non-trivial joint covariance matrices with sparse corresponding canonical directions; as opposed to being physically motivated%
	\footnote{Interestingly, one of their key observations was that high-dimensional sparse precision matrices can be well estimated, but they did not think to estimate the (full) joint precision, only the within-view blocks. Their particular choices of within-view covariance blocks leads to a certain sparsity in the joint precision matrix, which helps explain why our \glasso{} approach performs well.}.
	The construction has since been used in \citet{gao_sparse_2016,suo_sparse_2017,safo_sparse_2018,safo_integrative_2018,mai_iterative_2019,xiu_tssnalm_2020,rodosthenous_integrating_2020,laha_support_2021}.
	
	We generate multivariate Gaussian $\X \in \mathbb{R}^{N \times p}$ and $\Y \in \mathbb{R}^{N \times q}$ from 
	\begin{equation}\label{eq:joint-cov-form}
		\left(\begin{array}{l}
			x \\
			y
		\end{array}\right) \sim \mathrm{N}\left(\left(\begin{array}{l}
			0 \\
			0
		\end{array}\right),\left(\begin{array}{cc}
			\Sigxx & \Sigma_{x y} \\
			\Sigma_{y x} & \Sigyy
		\end{array}\right)\right) \; ,
	\end{equation}
	where we vary the form of $\Sigxx, \Sigyy$ but always generate
    \begin{equation}
		\Sigxy \defeq \Sigxx \left(\sum_{k=1}^{K} \rho_k u_k v_k^\top\right) \Sigyy
	\end{equation}
	with $\rho_k \in (0,1)$ and $u_k \in \R^m, v_k \in \R^p$ such that $u_k^\top \Sigxx u_j = v_k^\top \Sigyy v_j = \delta_{jk}$. From the arguments of \Cref{app:matrix-SVD-for-CCA}, we see that $\rho_k, u_k, v_k$ are precisely the corresponding canonical correlations and directions. 

    For this section, we consider constructing the covariances $\Sigxx, \Sigyy$ via $\Sigxx = (W_m)^{-1}$, $\Sigyy= (W_p)^{-1}$ where $W_m \in \R^{m\times m} $ is defined by 
    $$\left(W_m\right)_{ij} = \ind_{\{i=j\}} + 0.5 \times \ind_{\{|i-j|=1\}} + 0.4 \times \ind_{\{|i-j|=2\}} \;.$$
    This is referred to as the `sparse precision' model in \citet{suo_sparse_2017}.
    There is still plenty of flexibility in the choice of $\rho_k,u_k,v_k$.
    In this case, we took a single canonical pair model with $\rho_1 = 0.9$ (and $\rho_k = 0$ for $k > 1$).
    We constructed $u_1, v_1$ to only have non-zero entries in the first 5 components, with each such entry generated independently from a uniform distribution on $[-1,1]$ (a single pair $u_1, v_1$ were generated before any samples drawn, and used in all the runs).
    We will therefore denote this covariance structure as `\texttt{suo-sp-rand}'.
    The dimensions were taken as $p=q=150$, and we varied the number of samples $n$.

\subsubsection{Graphical Model}\label{sec:detail-graphical-models}
    We also consider covariance structures from a sparse graphical model.
    We use the `power law' construction for the covariance matrix corresponding to a sparse graphical model as implemented in the \texttt{gglasso} package \citep{schaipp_gglasso_2021}.
    Here the precision matrix is constructed using such that the interaction graph has nodes which satisfy a power law distribution, with some decay constant $\gamma$.
    Here we took $\gamma=3$.
    We denote this covariance structure as `\texttt{powerlaw}'.
    Again we worked in dimensions $p=q=150$, and we varied the number of samples $n$.
    
\subsection{Results}
    \begin{figure}[t]
         \centering
         \includegraphics[width=0.8\textwidth]{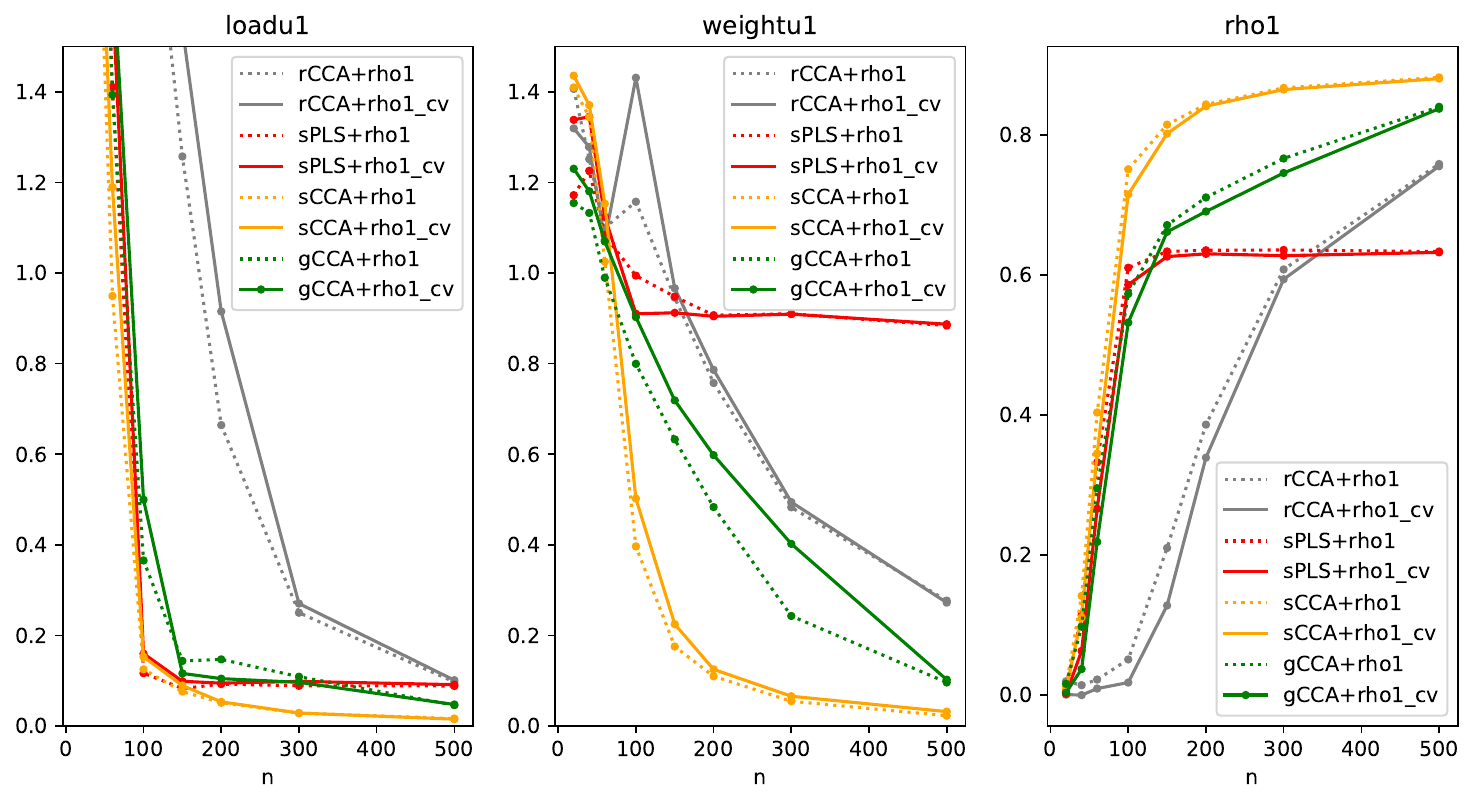}
         \caption{Evaluation of the four different algorithms for \texttt{suo-sp-rand} model, relative to the known covariance structure. Dimensions $p=q=150$ are fixed and we vary number of samples $n$. In each case we took the average over 32 random seeds. For the solid lines we use estimates with tuning parameters selected via CV on $\rho_1$ (\texttt{r2s1-cv}), dotted lines select tuning parameters using $\rho^\text{oracle}(\hat{u}_1,\hat{v}_1)$, i.e. the `oracle' \texttt{r2s1}.}
         \label{fig:1cc-suo_sp_rand}
     \end{figure}
    \begin{figure}[t]
         \centering
         \includegraphics[width=0.8\textwidth]{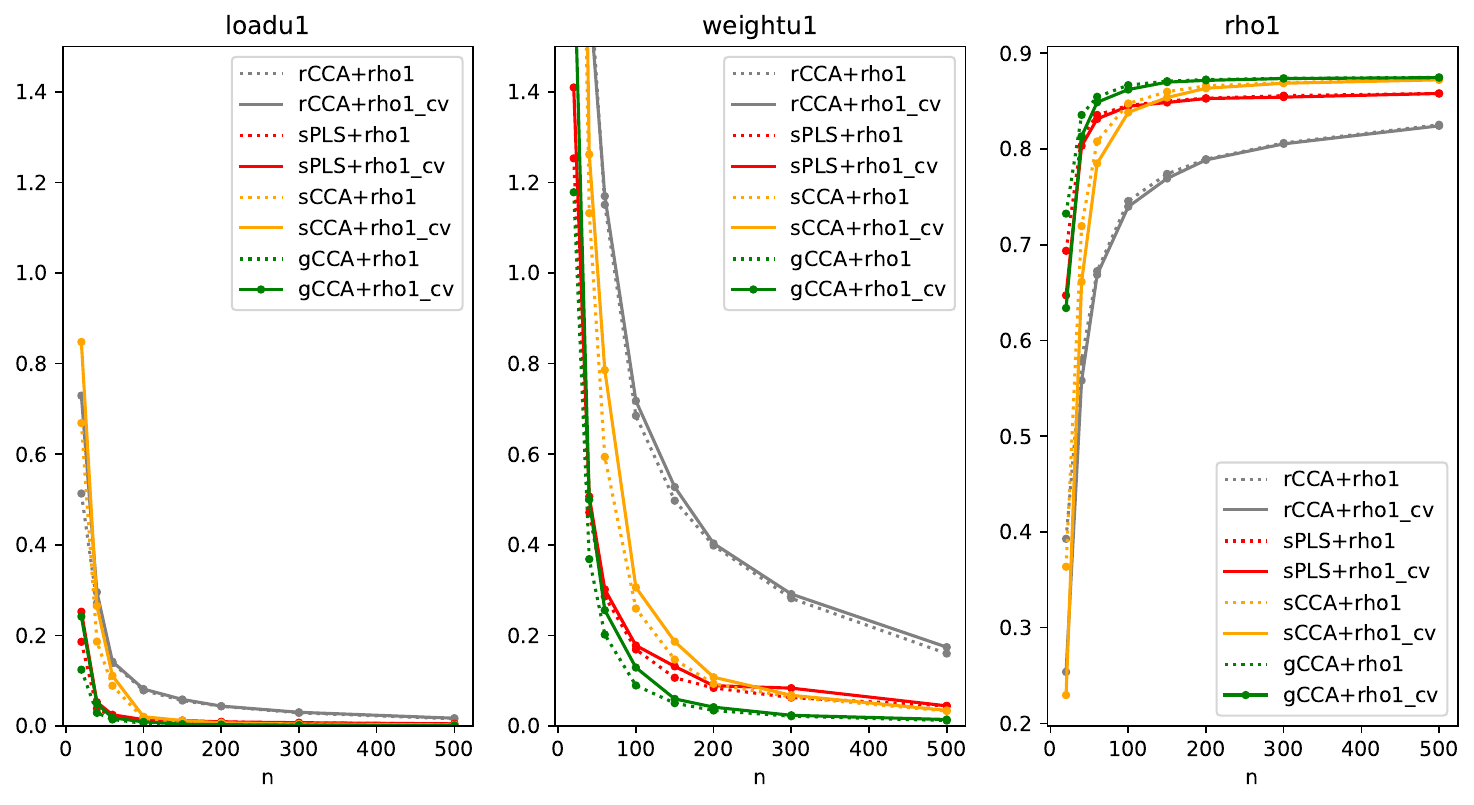}
         \caption{Analogous plot to Figure \ref{fig:1cc-suo_sp_rand} but now for \texttt{powerlaw} model; otherwise identical setup.}
         \label{fig:1cc-powerlaw}
     \end{figure}
     
\subsubsection{Competitive performance in single canonical pair models}
    We summarise some observations below:
	\begin{itemize}
	    \item Loading recovery is easier than weight recovery; and indeed we can recover high correlations even with poor weight estimates.
	    \item Even in this most benign case of a single strong canonical correlation, high sparsity, and moderate dimension, it requires a few hundred samples to get good estimates of the loadings, let alone the weights.
	    \item \wit{} performs well for small $n$ but fails to converge to the true values as $n\to\infty$ (i.e. is not consistent). This is not surprising because sPLS should converge to PLS decomposition rather than CCA in the limit of small regularisation, so there is no reason to expect consistency in general.
	    \item \glasso{} converges only somewhat slower than \suo{}, and much faster than \ridge{} .
	\end{itemize}
	
    In other similar scenarios, we generally saw that our gCCA method performed nearly as well as existing alternatives.

\subsubsection{Good performance in graph-structured models}

    Figure \ref{fig:1cc-powerlaw} is the analogue of Figure \ref{fig:1cc-suo_sp_rand} but now with the \texttt{powerlaw} model.
    We note some further observations:
    \begin{itemize}
        \item Generally this problem appears easier than the \texttt{suo-sp-rand} case. 
        \item Again loadings are easier to estimate than weights.
        \item Now \glasso{} is clearly the best performing model.
        \item Again \wit{} does not appear to converge to the CCA subspace. 
    \end{itemize}

\newpage
\section{Parametric bootstrap: further synthetic experiments}

\subsection{Convergence as number of samples grows}\label{app:parametric-bootstrap-vary-n-still-glasso-reg}
    \begin{figure}[t]
         \centering
         \includegraphics[width=\textwidth]{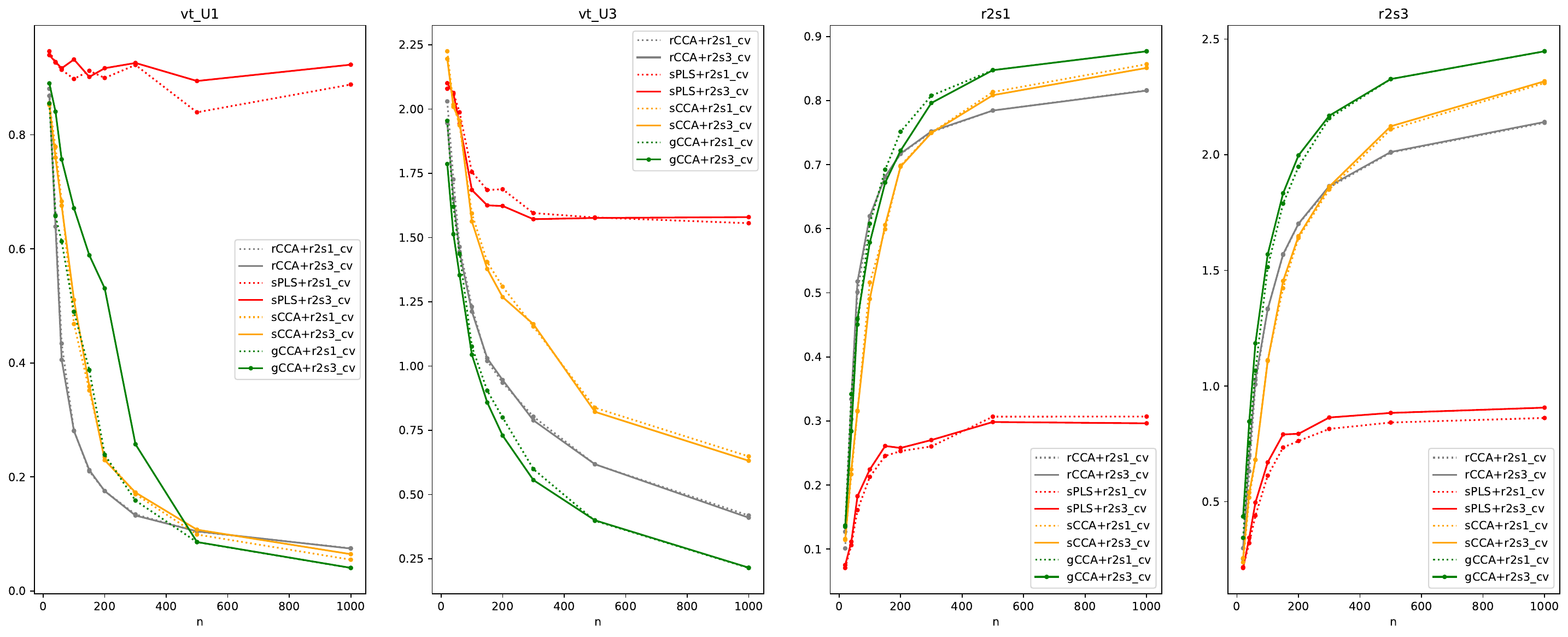}
         \caption{Evaluation of the four different algorithms for the parametric bootstrap model from the main text (derived from the Microbiome dataset). The x-axis varies the number of samples $n$. For each data-point we took the average over 32 random seeds. Note that now solid and dotted lines only distinguish between number of directions used for \texttt{r2sk-cv} penalty parameter selection ($k=1$ vs $k=3$).}
         \label{fig:vary-n-pboot-mb}
     \end{figure}
     
    We now return to the parametric bootstrap setup from the main text, and consider the behaviour of the estimates as number of samples grows.
    
    \Cref{fig:vary-n-pboot-mb} shows the results, and is generated much like the analogous \Cref{fig:1cc-powerlaw,fig:1cc-suo_sp_rand}. Indeed, we fix a single multivariate normal (MVN) data generating mechanism, whose covariance is a GLasso-regularised version of the Microbiome sample covariance matrix (as described in \Cref{sec:synthetic-data}), and vary the number of samples $n$ generated from the MVN.
    For each random sample of data of size $n$, we select a penalty parameter using the corresponding quantity in the legend (\texttt{r2s1-cv} or \texttt{r2s3-cv}) and observe the quantity in the title of the plot (\texttt{vt-U1}, \texttt{vt-U3}, \texttt{r2s1}, or \texttt{r2s3}). Note that, unlike in \Cref{app:toy-synthetic-experiments}, all the metrics used for tuning parameter selection are CV metrics (so could be used in practice), while the plotted quantities are all oracle quantities (so give meaningful evaluation of behaviour on the true distribution).
    We plot the mean value over 32 random seeds.

    The plot gives further justification to a number of observations from the maintext.
    \begin{itemize}
        \item \textbf{\wit{} does not appear to converge} to the CCA solution as $n$ increases, while the other CCA methods do.
        \item \textbf{\glasso{} outperforms the other methods}, particularly for the \texttt{vt-U3} and \texttt{r2s3} metrics. This out-performance is particularly stark in terms of sample efficiency: to get an \texttt{r2s3} value of 2.0, \glasso{} requires approximately two times fewer samples than \suo{}, and three times fewer than \ridge{}. Interestingly, \glasso{} seems to lag behind \ridge{} and \suo{} on the \texttt{vt-U1} metric; we suggest that this is due to a complex interaction between the GLasso regularisation in the \glasso{} with the structure of the population canonical directions, and may reflect the `jagged' nature of the \texttt{vt-U1} line in \Cref{fig:boots_mb_panel}.
        \item In general, it appears that the choice of method for regularisation is more important than the particular choice of CV correlation objective (the dotted lines closely track the filled lines, whereas there is much more variation between the different colours).
    \end{itemize}


\subsection{Alternative regularisation for oracle covariance matrix}\label{app:parametric-bootstrap-alternative-suo-regularisation}
    \begin{algorithm}
		\caption{\suo{}-and-ridge-regularised joint covariance matrix}\label{alg:suoR-regularised-pboot-cov}
		\begin{algorithmic}
			\Procedure{\texttt{sCCA-regularised-covariance}}{$\X,\Y; \lambda, \alpha, K$}
                \vspace{4pt}
                \State $\Sigxx, \Sigyy \gets \Cxx + \alpha I_p, \Cyy + \alpha I_q$
			\vspace{4pt}
                \State $\hat{U}_K, \hat{V}_K \gets \suo{}_K(\X, \Y; \lambda)$ \Comment{Columns are the successive $\hat{u}_k, \hat{v}_k$ from \Cref{eq:obj-suo}}
                \vspace{6pt}
                \State $\hat{D}_K = \diag\left(\empCorr(\X \hat{u}_k, \Y \hat{v}_k)_{k=1}^K \right)$
                \vspace{6pt}
                \State $\hat{\Sigma}_{xy} \gets \hat{\Sigma}_{xx} \hat{U}_K \hat{D}_K \hat{V}_K^\top \hat{\Sigma}_{yy}$ \Comment{C.f. \Cref{eq:sigxy-from-T-and-URV}}
                \vspace{6pt}
                \State \textbf{return} $\hat{\Sigma} \defeq \left(\begin{array}{cc}
        		\hat{\Sigma}_{xx} & \hat{\Sigma}_{xy} \\
        		\hat{\Sigma}_{xy}^\top & \hat{\Sigma}_{yy}
        	\end{array}\right)$ 
    			\EndProcedure
		\end{algorithmic}
	\end{algorithm}
 
    In \Cref{sec:synthetic-data} we considered a multi-variate normal model where the true (oracle) joint covariance matrix was obtained by regularising the sample covariance of the Microbiome dataset with the graphical lasso.
    Here, we construct this oracle joint covariance matrix using \suo{} to encourage structure with sparse canonical directions.

    In more detail, we construct this matrix using \Cref{alg:suoR-regularised-pboot-cov}.
    First we regularise the within-view covariance matrices using a ridge penalty, choosing small but arbitrary parameter of $\alpha = 0.01$.
    Then we obtain $K=10$ successive canonical direction estimates using \suo{} and the \texttt{r2s5-cv}-optimal penalty of $\lambda = 0.013$.
    Finally we construct the between view covariance matrices using \Cref{eq:sigxy-from-T-and-URV}.

    \Cref{fig:boots_mb_panel_suoR_regularisation} is the analogous plot to \Cref{fig:boots_mb_panel} but for this \suo{}-regularised joint covariance matrix. The primary observation is that the figure looks extremely similar to \Cref{fig:boots_mb_panel}; now \suo{} appears to perform slightly better than \glasso{} in terms of correlation captured (as one would expect); otherwise all the observations from \Cref{sec:metric-choice-bootstrap} still hold.

    \Cref{fig:boots-mb-corr-and-stab-suoR-regularisation} is the analogous plot to \Cref{fig:boots-mb-corr-and-stab} but for this \suo{}-regularised joint covariance matrix. Again, the primary observation is that the figure looks extremely similar to \Cref{fig:boots-mb-corr-and-stab}; again \suo{} appears to perform slightly better than \glasso{} in terms of correlation captured; otherwise all the observations from \Cref{sec:synth-cv-is-effective} still hold.
    
    \begin{figure}[t]
        \centering  \includegraphics[width=\textwidth]{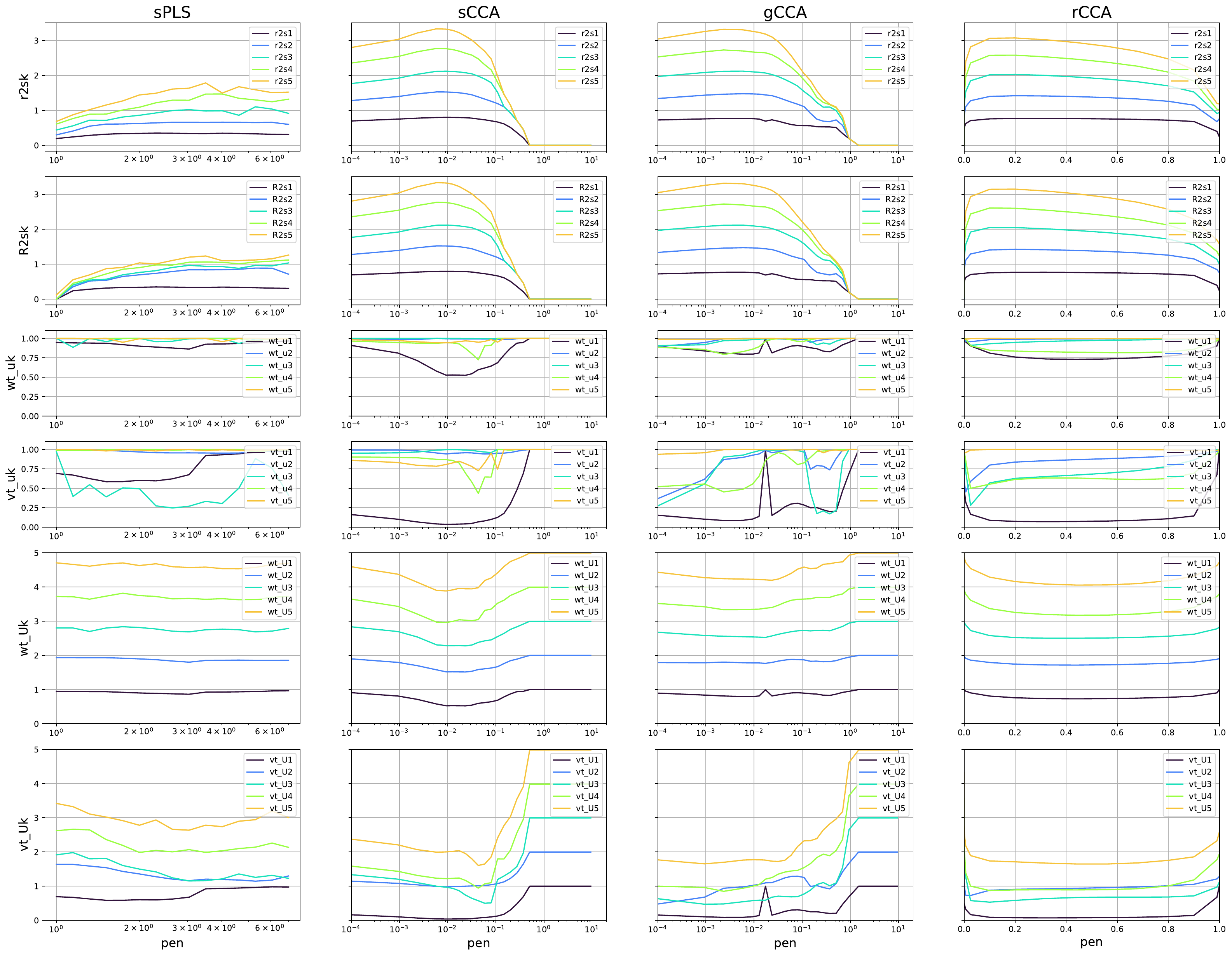}
        \caption{Oracle metrics: analogous plot to \Cref{fig:boots_mb_panel} on a parametric-bootstrapped Microbiome dataset, where now the `oracle' joint covariance matrix is constructed using \suo{}. See \Cref{tab:metric-summary-estimation,tab:metric-summary-correlation} for a glossary of the legends.}
        \label{fig:boots_mb_panel_suoR_regularisation}
    \end{figure}

    \begin{figure}[t]\centering
    \begin{minipage}{1\textwidth}\centering
            \includegraphics[width=0.9\textwidth]{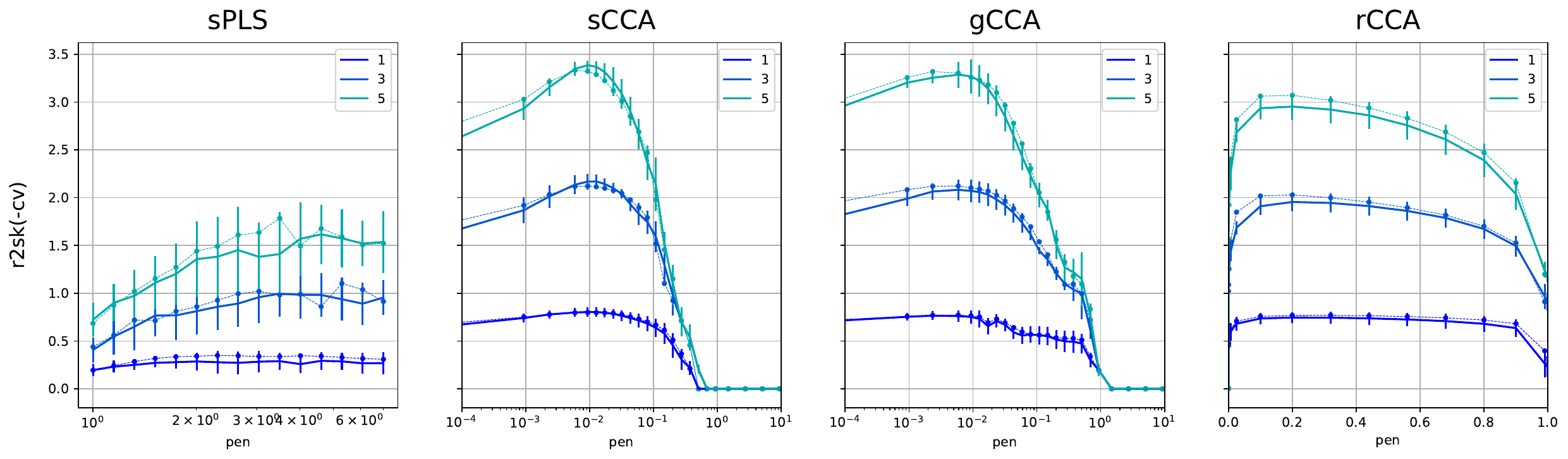}
        \\
            \includegraphics[width=0.9\textwidth]{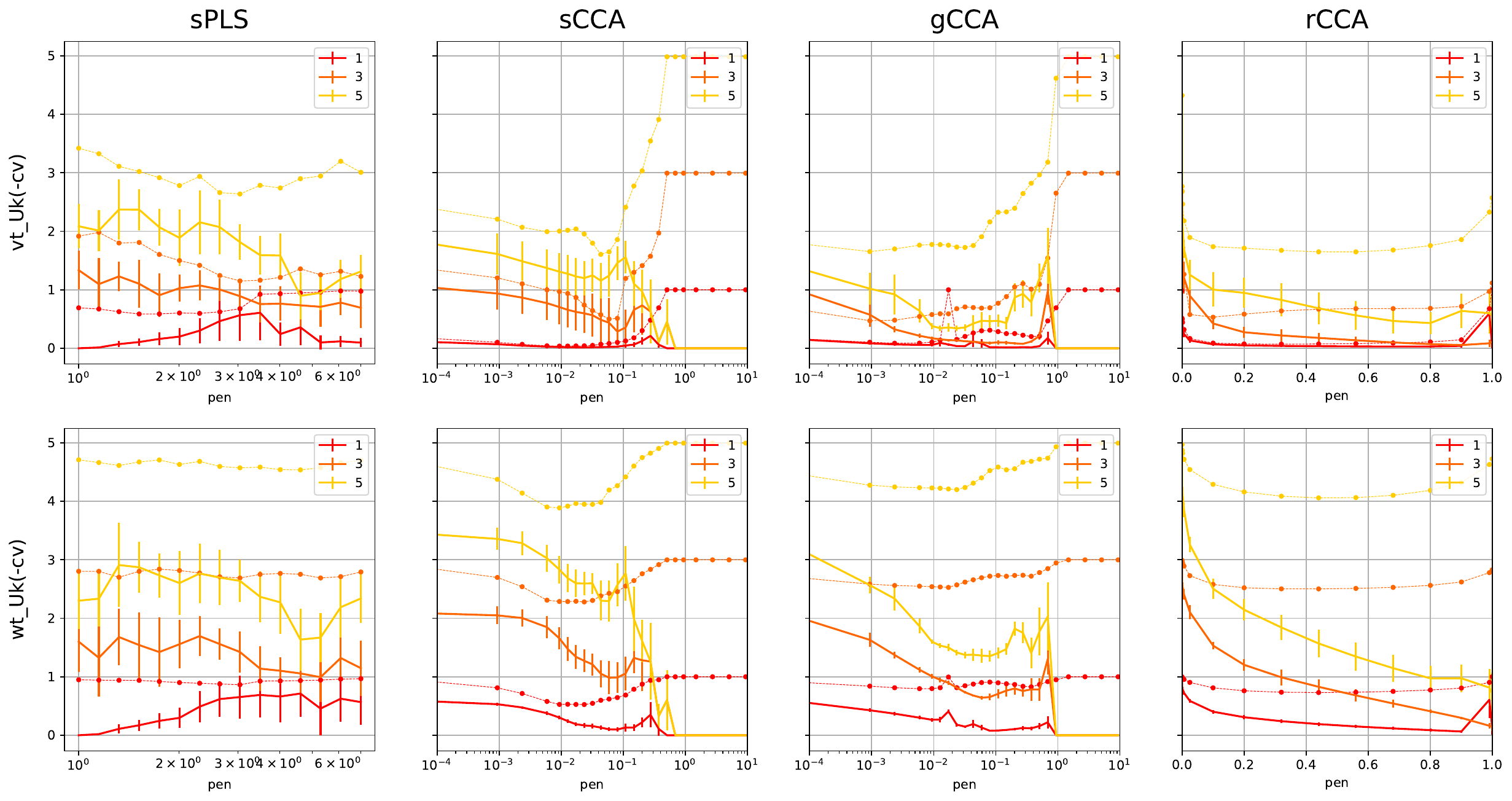}
    
        \caption{Sample splitting metrics (with comparable oracle quantities dotted and overlaid): analogous plot to \Cref{fig:boots-mb-corr-and-stab} on a parametric-bootstrapped Microbiome dataset, where now the `oracle' joint covariance matrix is constructed using \suo{}. 
        See \Cref{tab:metric-summary-correlation,tab:metric-summary-estimation} for a glossary of the legends.}
        \label{fig:boots-mb-corr-and-stab-suoR-regularisation}
    \end{minipage}
    \end{figure}

\newpage\section{Real data supporting plots}\label{app:real-data-supporting-plots}
\raggedbottom

\subsection{Microbiome}\label{app:microbiome-extra-plots}

We now provide additional figures to support the claims from \Cref{sec:pipeline-toolbox}. We summarise the main observations below, with figures in the following subsections.
\begin{itemize}
    \item \textbf{Trajectory comparison distances larger in weights space:} \Cref{fig:microbiome-weight-traj-comp} has far higher entries than the corresponding variate version \Cref{fig:microbiome-traj-comp-variate}.
    \item \textbf{Overlap matrices are versatile:} the three sets of overlap matrices in \Cref{fig:fold_stab_microbiome_gglasso,fig:microbiome-sqoverlap-algos-orthog,fig:microbiome-sqoverlap-suo-path} illustrate the following points:
    \begin{itemize}
        \item \textbf{gCCA is stable across folds:}
        In particular, the squared overlap matrices in \Cref{fig:fold_stab_microbiome_gglasso} show how estimates for the \texttt{r2s5-cv}-optimal gCCA estimator vary across 3 training folds, see caption for full details.
        The main observation from Figure \ref{fig:fold_stab_microbiome_gglasso} is that the Glasso estimates are reassuringly stable across the folds: the squared overlap matrices all appear close to identity, especially for the first 3 components;
        moreover, the row and column sums of the first few components are close to one.
        This is perhaps surprising giving the lack of separation of the correlation signals from Figure \ref{fig:corr_decay_microbiome_gglasso}.
        Note that we do also observe some mixing up of the 4/5 components in the middle plot.
        
        \item \textbf{visualising registration:}
        \Cref{fig:microbiome-sqoverlap-algos-orthog} uses the same estimates as \Cref{fig:microbiome-sqoverlap-algos} but now the \suo{} and \ridge{} solutions have been registered to the \glasso{} solution. Observe all the matrices look much more like identity matrices.
        We also plot \wit{} in this figure to show just how differently this method behaves.
        
        \item \textbf{mixing up over trajectories:}
        Overlap matrices can also illustrate how different estimators from a single algorithm change with penalty parameter; \Cref{fig:microbiome-sqoverlap-suo-path} shows how \suo{} estimates on the full dataset change with penalty parameter. In particular, observe that the third plot shows some significant `mixing-up' of the second and third directions; but otherwise the signal is very stable.
    \end{itemize}
\end{itemize}  

Figure \ref{fig:3D-plots} shows similar biplots but now in 3D, representing the first 3 canonical variates. 
    Again we have one plot for each regularised CCA method, and with variates registered to the central \glasso{} estimate.
    Our main point is that these three plots are remarkably similar.
    It is expected that the three plots should be similar given the overlap matrices.
    However it is remarkable quite how similar they are; moreover, now we see that much local structure is preserved as well as global structure.
    One might think of this as an illusion of being in 3D, but we suggest it is genuinely helpful.
    First note that the variables are fairly evenly spread around sphere: many of the variables have significant correlation with all 3 of the first 3 variate estimates
    The result of the extra dimension is that there are very fewer variables which are very close together than in the 2D case.
    A small amount of randomness will therefore keep the more local relationships intact.
    
\begin{figure}[b]\centering
    \subfloat{
        \includegraphics[width=0.31\textwidth]{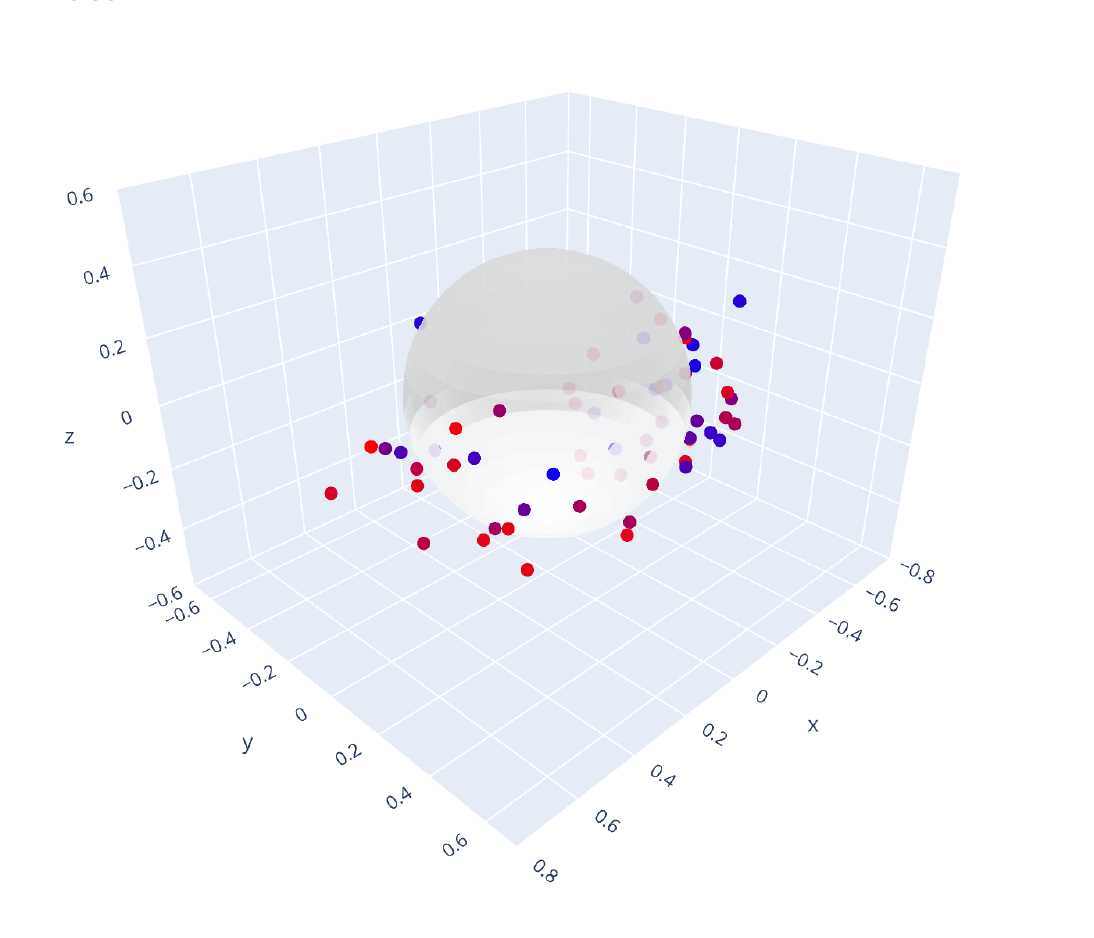}
    }
    \subfloat{
        \includegraphics[width=0.31\textwidth]{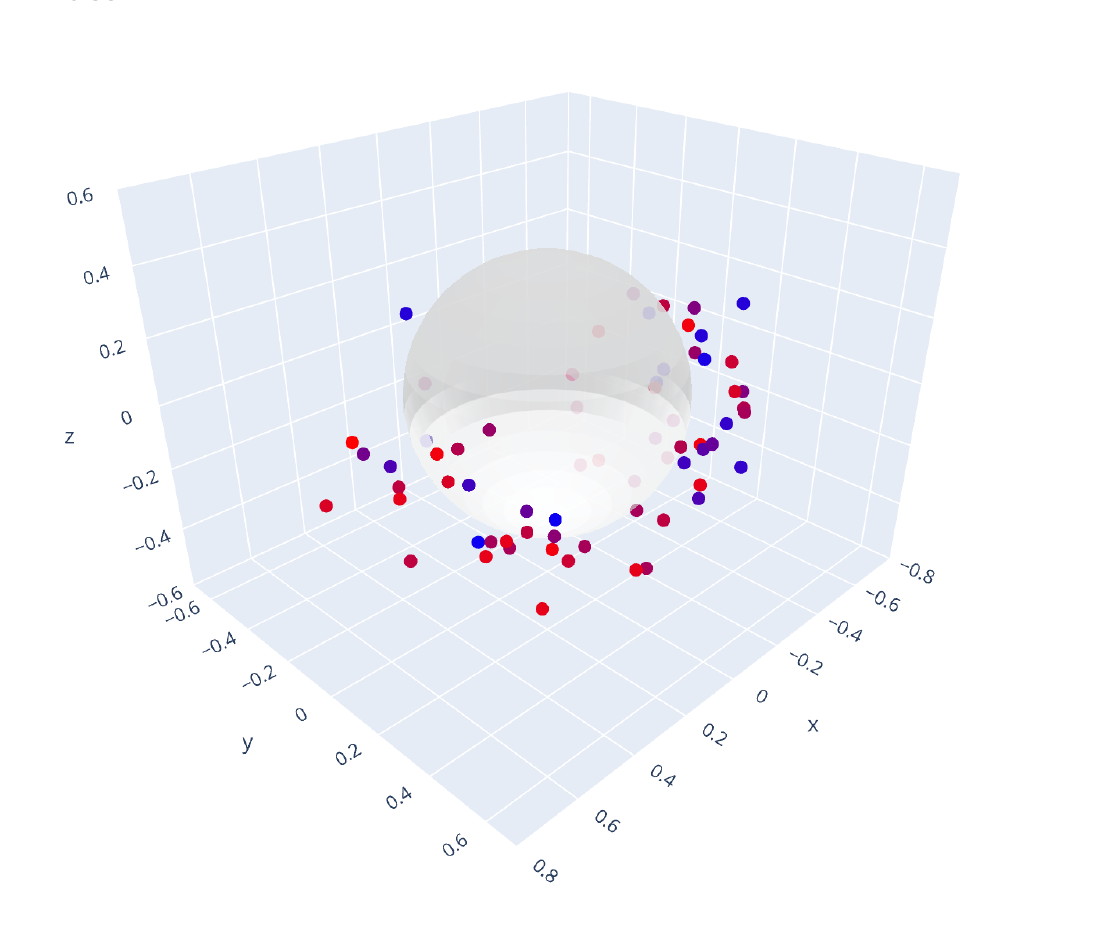}
    }
    \subfloat{
        \includegraphics[width=0.31\textwidth]{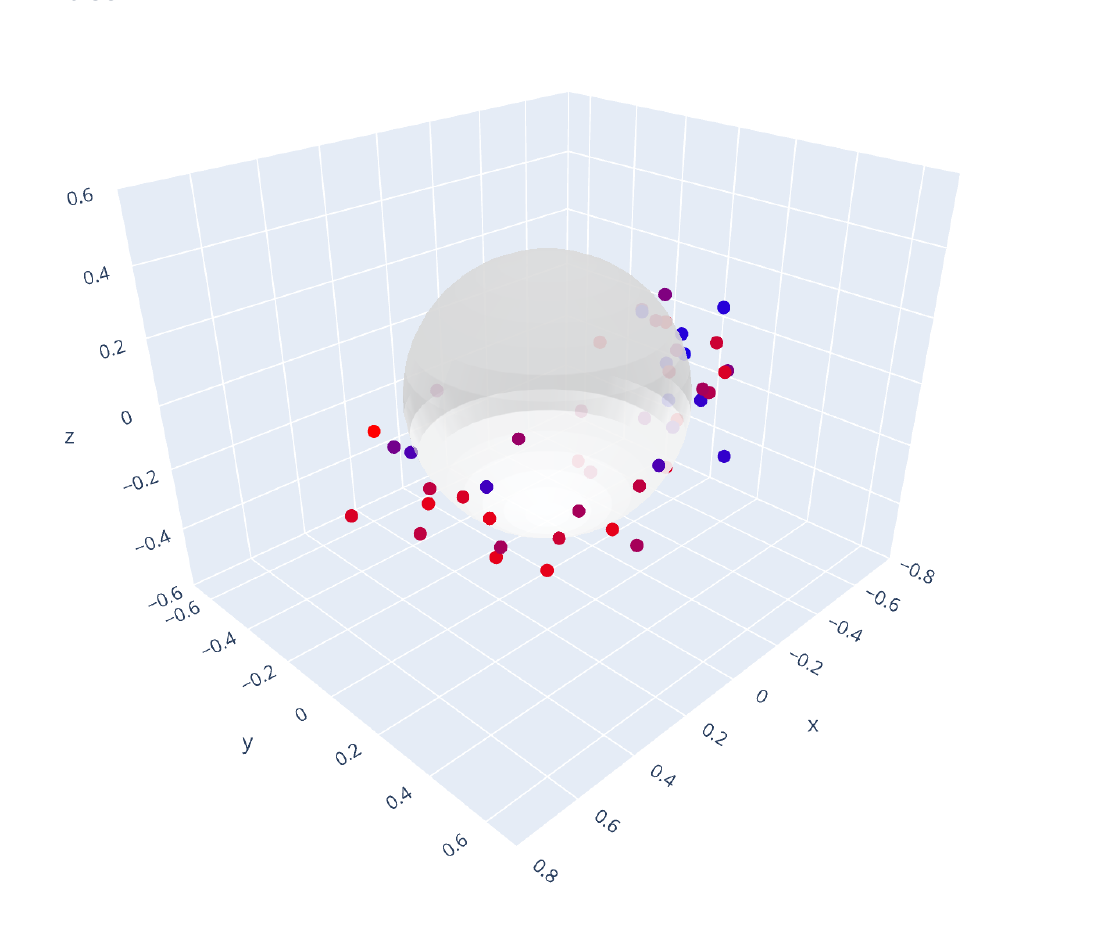}
    }
    \caption{3D `biplots' for microbiome dataset for \ridge{}, \glasso{}, \suo{} methods and penalty parameters 0.05, 0.01, 0.001 respectively. The \ridge{} and \glasso{} solutions have been registered with the \ridge{} up to orthogonal transformations in variate space. Variates in question are from K0s, but only the C0 variables are plotted, with colours consistent between the three plots but otherwise arbitrary.}
    \label{fig:3D-plots}
\end{figure}

    \begin{figure}[t]\centering
        \centering
        \includegraphics[width=\textwidth]{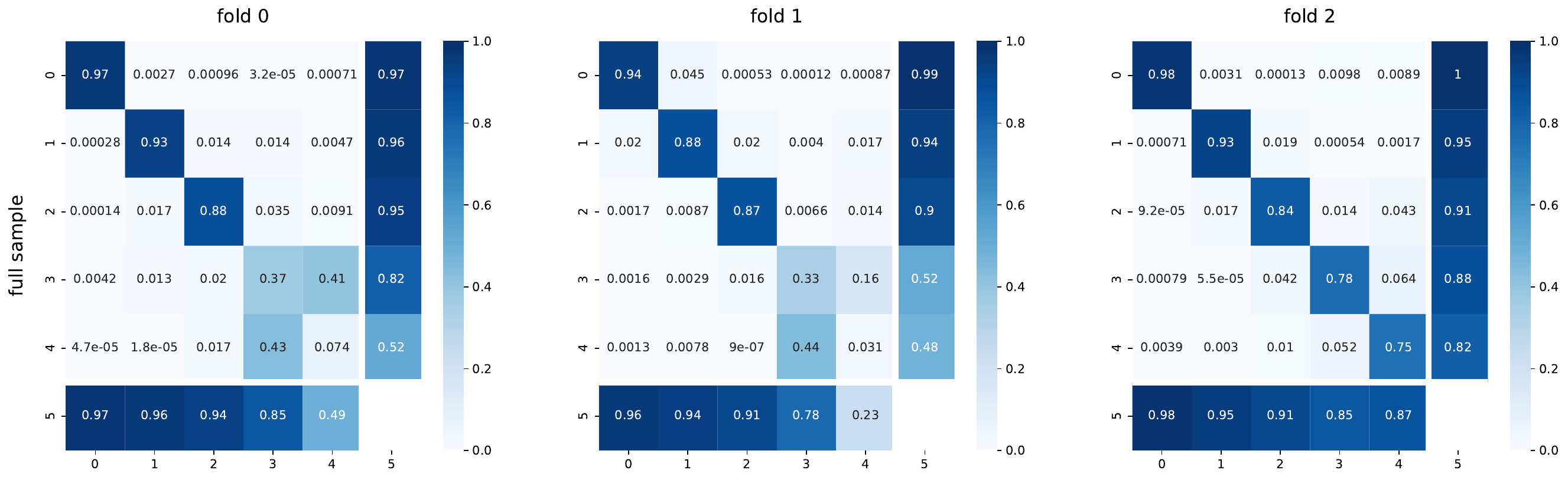}
    \caption{Squared overlap matrices comparing multiple \glasso{} estimates for different folds of the microbiome dataset, all using penalty parameter $\alpha=0.0059$. The y-axes correspond to weights trained on the full dataset, while the x-axes corresponds to weights fitted only on training data corresponding to folds $\nu=0,1,2$ respectively; but crucially, these weights are compared `in variate space' by using the full data set (i.e. take $\bZ = \X \hat{U}_K, \quad \bW = \X \hat{U}^{(-\nu)}_K$ in notation of \Cref{sec:overlap-introduction}).}
    \label{fig:fold_stab_microbiome_gglasso}
    \end{figure}

    \begin{figure}[t]\centering
        \includegraphics[width=\textwidth]{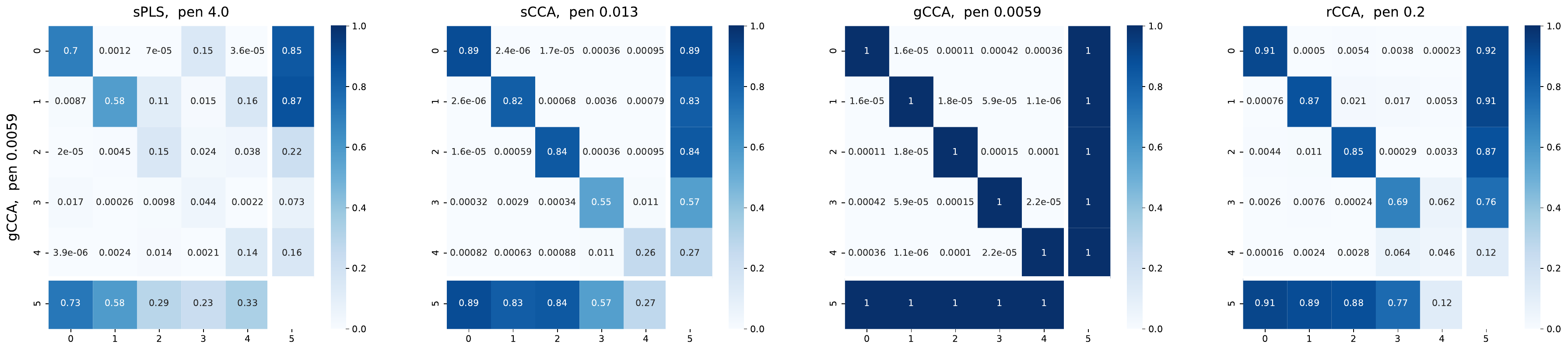}
    \caption{Squared overlap matrices comparing \wit{}, \suo{}, \glasso{}, \ridge{} estimators with maximal CV correlation (\texttt{r2s5-cv}), as in \Cref{fig:microbiome-sqoverlap-algos} but now the \wit{}, \suo{} and \ridge{} estimates have been registered to the \glasso{} estimates up to orthogonal matrices}.
    \label{fig:microbiome-sqoverlap-algos-orthog}
    \end{figure}

    \begin{figure}[t]\centering
        \includegraphics[width=\textwidth]{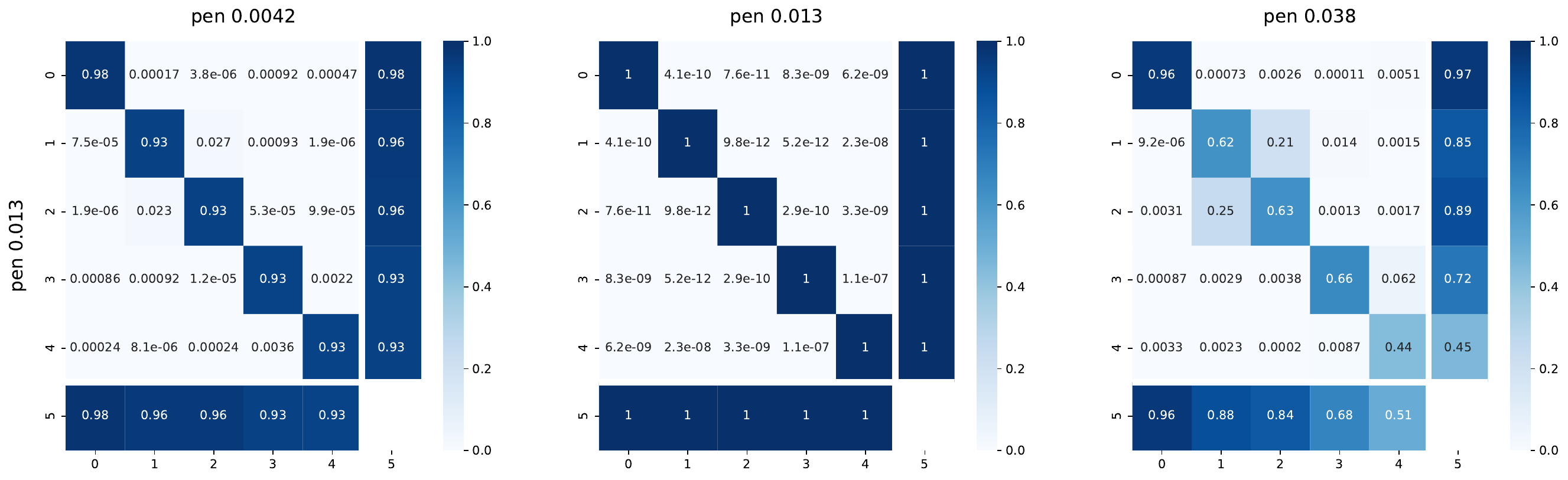}
    \caption{Overlap matrix of correlations between multiple \suo{} canonical variate estimates for multiple penalty parameters. The y-axis represents estimates using the \texttt{r2s5-cv}-optimal penalty; the three different x-axes represents that same \texttt{r2s5-cv}-optimal penalty and penalties a factor of 3 either side.}
    \label{fig:microbiome-sqoverlap-suo-path}
    \end{figure}

    \begin{figure}[t]\centering
         \includegraphics[width=0.8\textwidth]{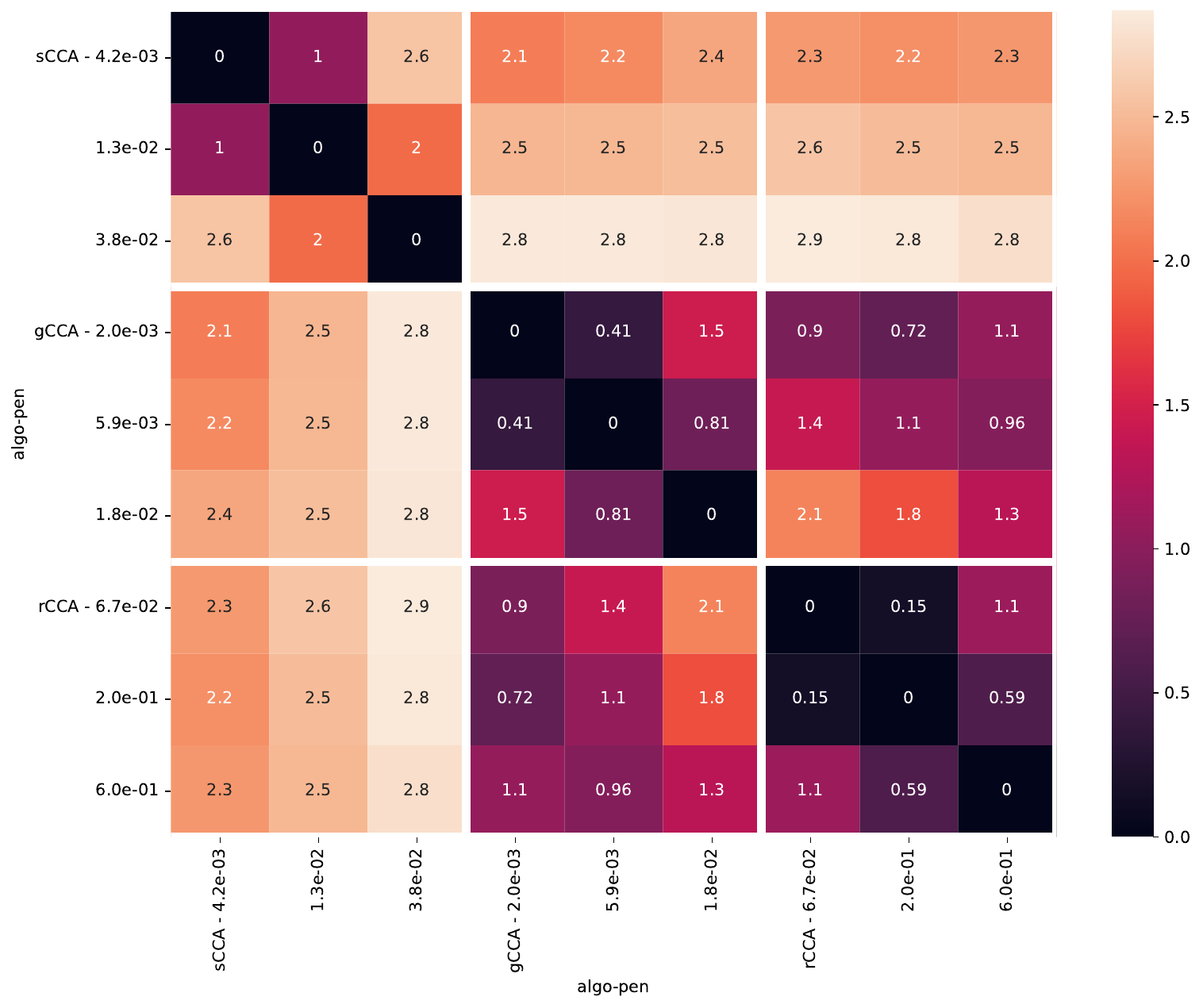}
         \caption{Sin$\Theta$ distance for top-3 weight subspaces for Microbiome dataset; axes are the same concatenated trajectories for \suo{}, \glasso{}, \ridge{} respectively}
         \label{fig:microbiome-weight-traj-comp}
    \end{figure}

\FloatBarrier
\subsection{Nutrimouse}\label{sec:nutrimouse-extra-plots}
We can now justify the following observations alluded to in the main text.
\begin{itemize}
    \item \textbf{correlations:} see \Cref{fig:correlation-stability-trajectories-nutrimouse,fig:corr-decay-nutrimouse}. Observe that there are many successive pairs with significant CV signal; moreover, genuine CCA methods capture similar signal, while sPLS captures significantly less signal.
    \item \textbf{very small penalty parameters:} appear near optimal in \Cref{fig:correlation-stability-trajectories-nutrimouse} for sCCA and gCCA.
    \item \textbf{sPLS behaves very differently} to the genuine CCA methods, as illustrated by the square overlap matrices in \Cref{fig:nutrimouse-sqoverlap-algos}.
    \item \textbf{gCCA very similar to rCCA:} this is clearest to see via trajectory comparison matrices \Cref{fig:nutrimouse-traj-comps}; indeed even the weight subspaces estimated are very similar. But this also manifests in all the other plots: 
    for small penalty parameters, gCCA and rCCA have very similar correlation and stability properties in \Cref{fig:correlation-stability-trajectories-nutrimouse}; 
    there are striking similarities in the overall shape of the correlation decay plots, and also specific arrangement of points, for example $k=4,6$ both have extremely high test correlation on fold 2 of 5; 
    finally, the relevant overlap matrix in \Cref{fig:nutrimouse-sqoverlap-algos} is very near the identity (and this was without any registration).\label{obs:color-coding}
    \item \textbf{sCCA is less stable:} for near-optimal CV penalty parameters. This is both for stability over folds \Cref{fig:correlation-stability-trajectories-nutrimouse} and stability with respect to tuning parameter \Cref{fig:nutrimouse-traj-comps} (higher distances between sCCA algorithms than between the rCCA or gCCA algorithms, especially in weight space).
\end{itemize}

\newpage

 \begin{figure}[t]
    \begin{minipage}{0.9\textwidth}
        \subfloat{
            \includegraphics[width=\textwidth]{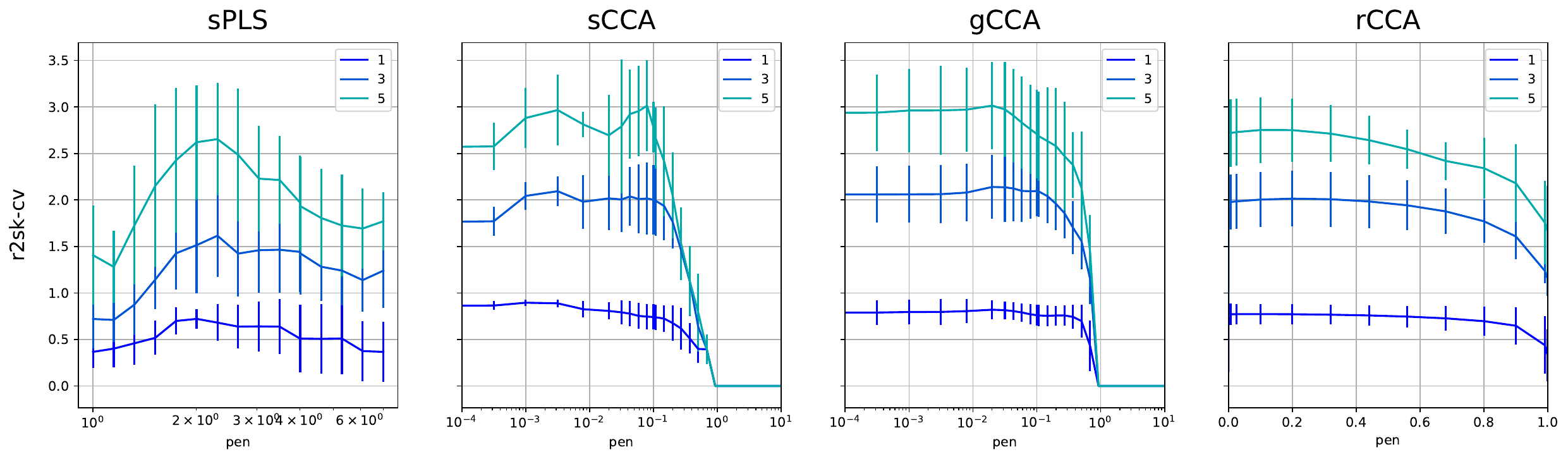}
        }\\
        \subfloat{
            \includegraphics[width=\textwidth]{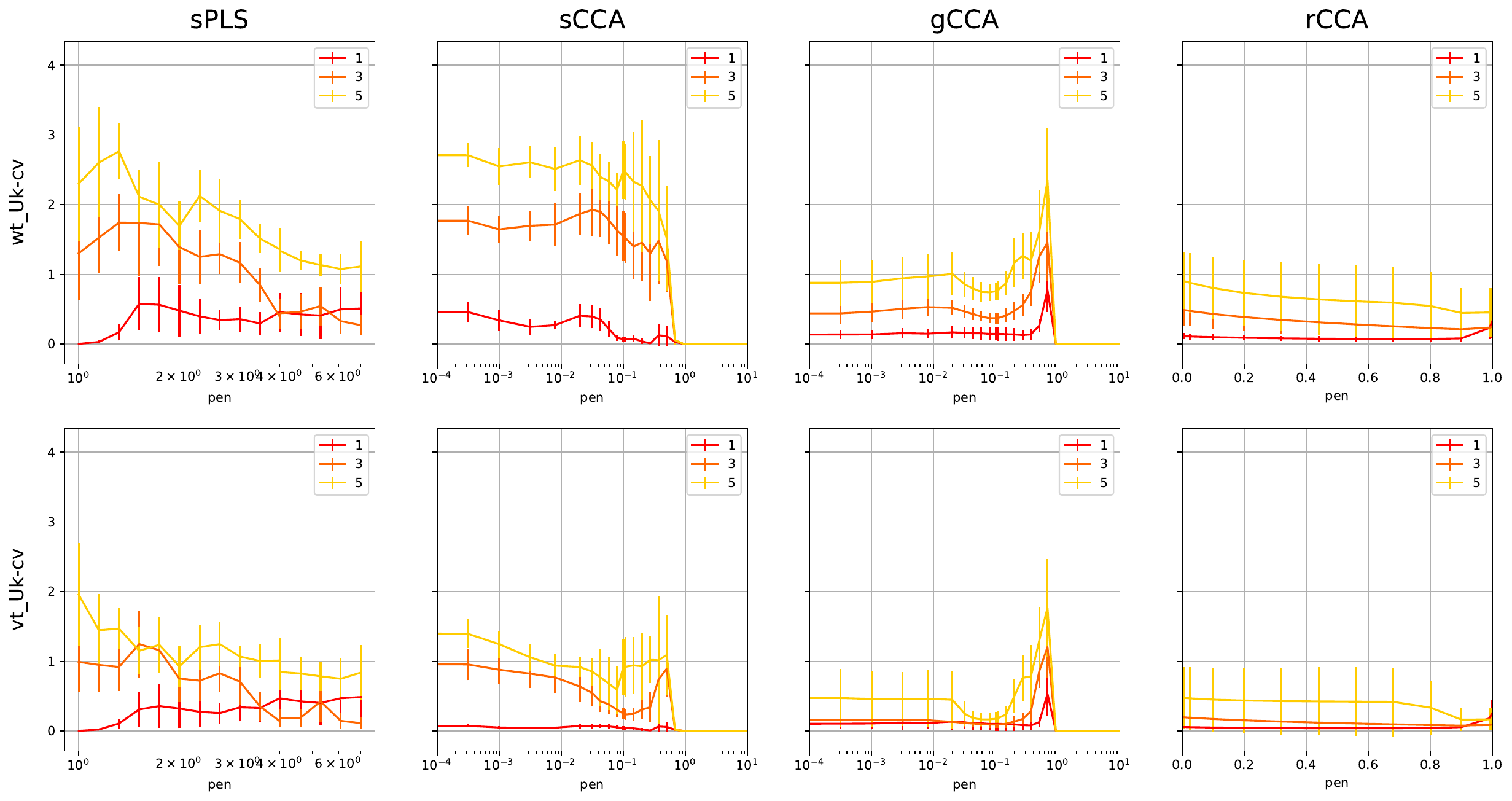}
        }        
        \caption{Top: CV sums of correlations as function of regularisation path for the four methods on the Nutrimouse dataset; error bars for the aggregated quantities.
        Bottom: stability both in weight space and variate space along the same trajectories.}
        \label{fig:correlation-stability-trajectories-nutrimouse}
    \end{minipage}
    \end{figure}
    
    \begin{figure}[t]\centering
         \includegraphics[width=0.9\textwidth]{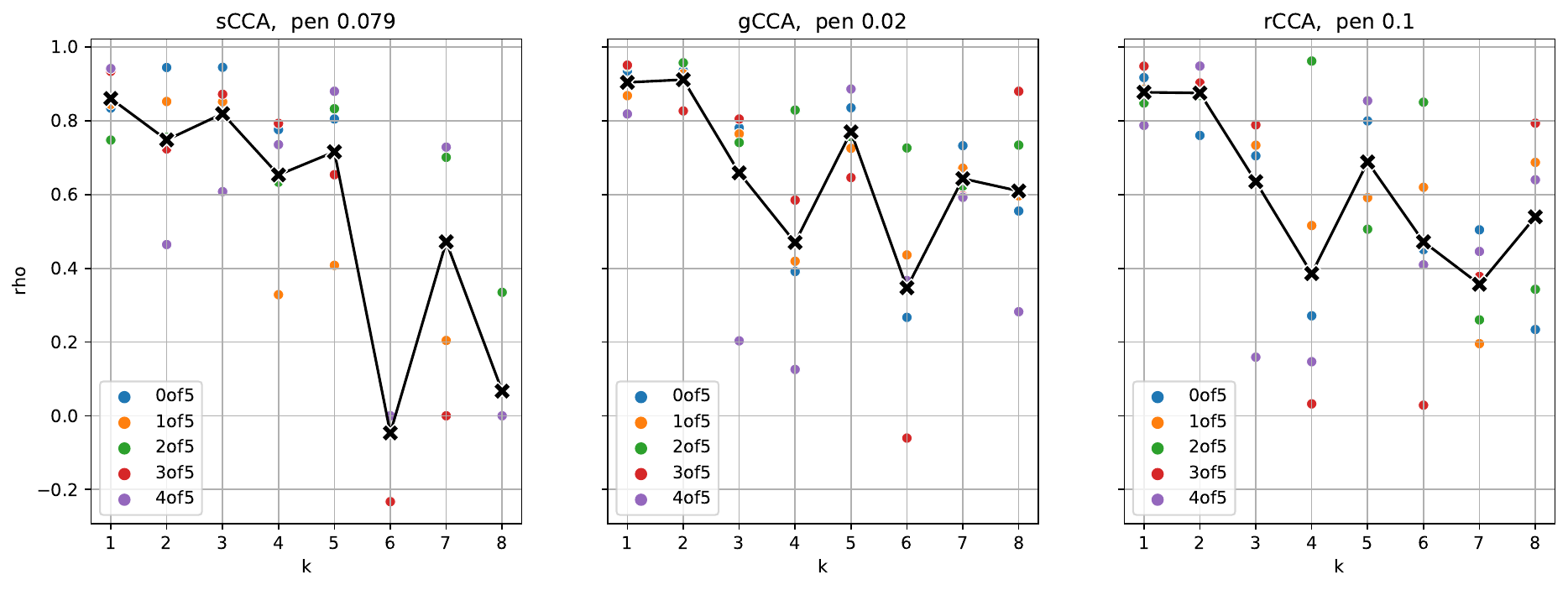}
         \caption{CV correlations (colours) and average values (black) for successive direction estimates using \suo{}, \glasso{}, \ridge{} on microbiome dataset; in each case \texttt{r2s5-cv} optimal penalty parameters were used.}
         \label{fig:corr-decay-nutrimouse}
    \end{figure}

    \begin{figure}[t]\centering
    \begin{minipage}{0.75\textwidth}
        \subfloat{
        \includegraphics[width=\textwidth]{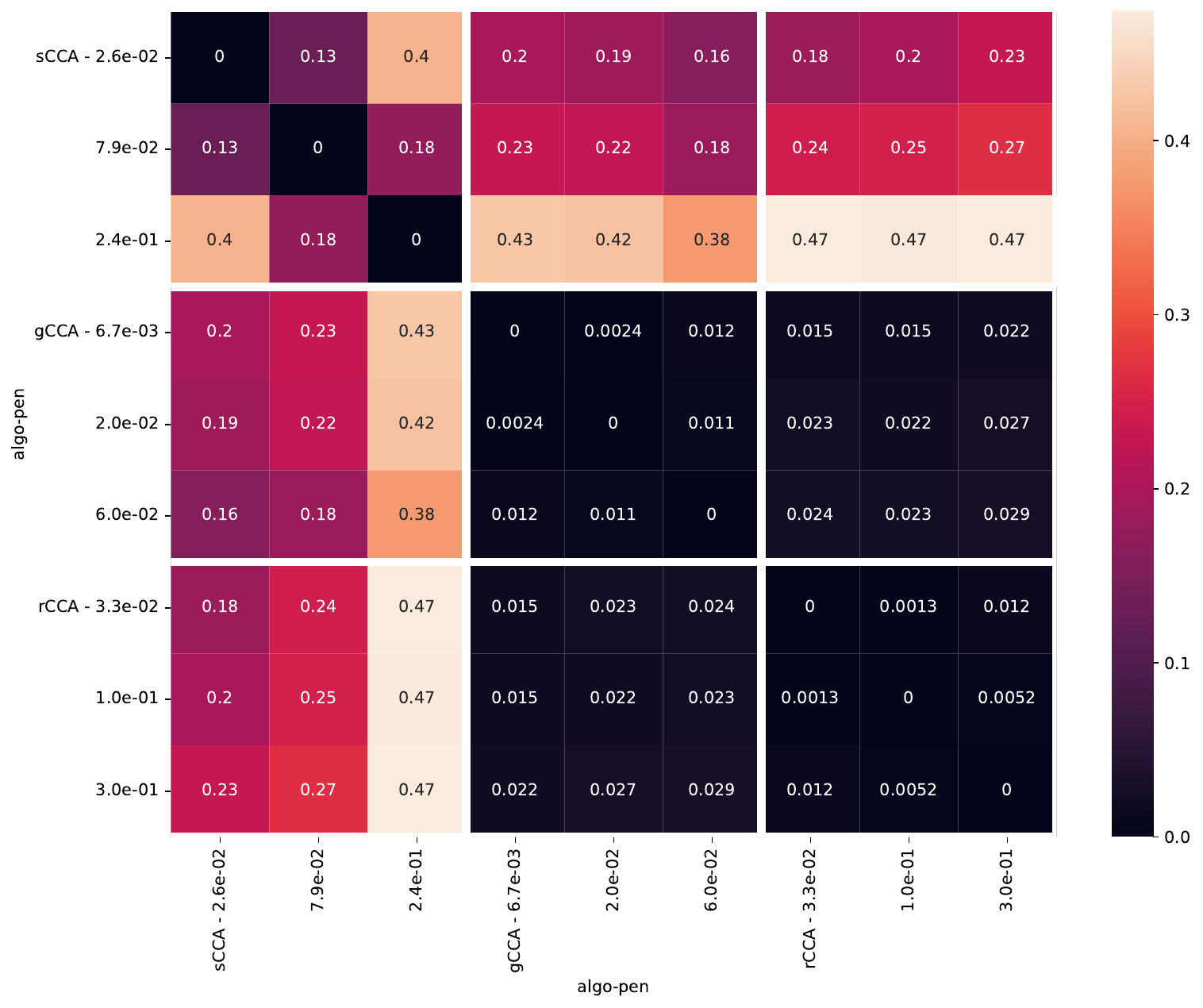}
        } \\
        \subfloat{
            \includegraphics[width=\textwidth]{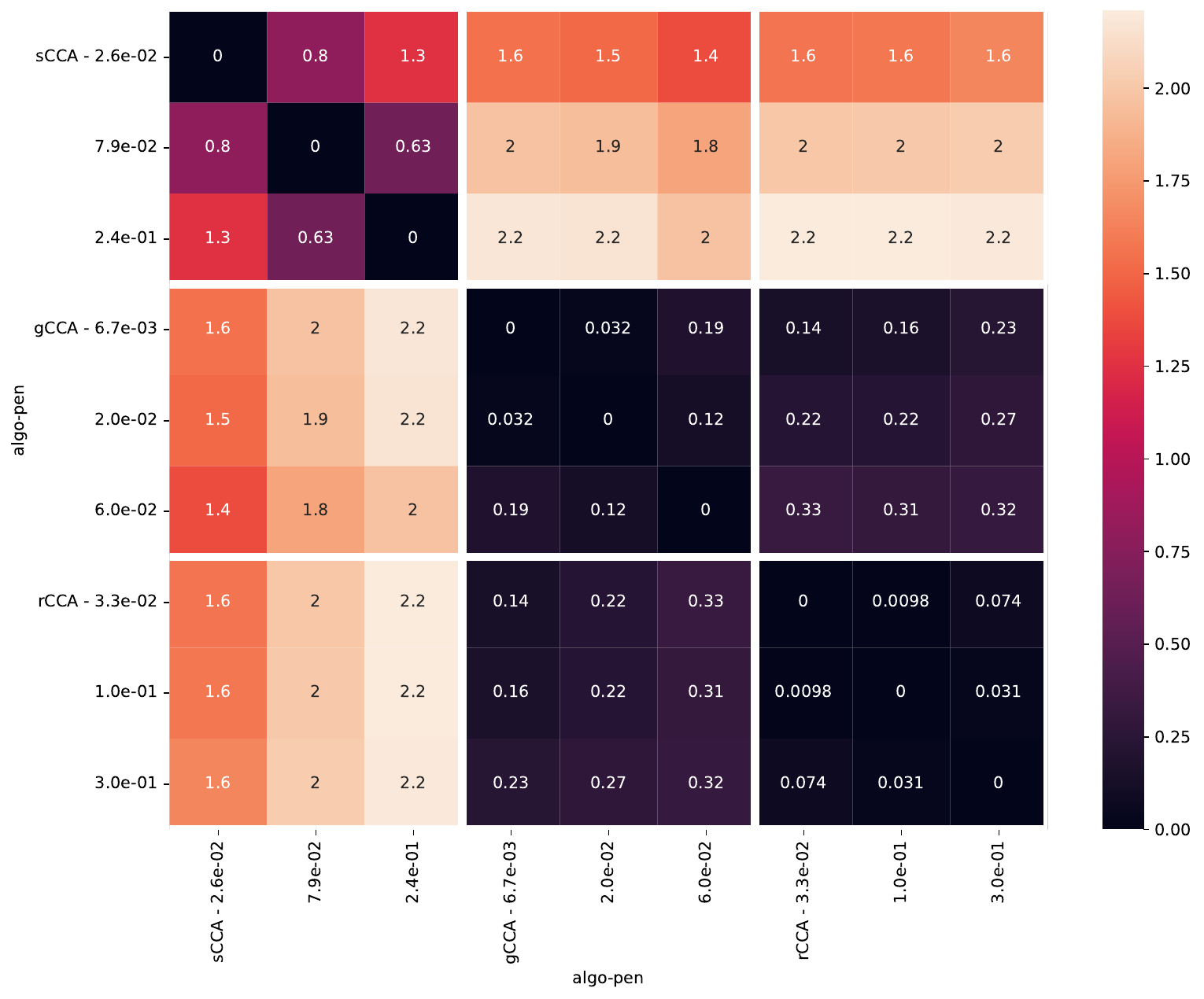}
        }
        \caption{Top: Nutrimouse trajectory comparison top-3 variate subspace distances. \\
        Bottom: Nutrimouse trajectory comparison top-3 weight subspace distances.}
        \label{fig:nutrimouse-traj-comps}
    \end{minipage}
    \end{figure}

    \begin{figure}[b]\centering
        \includegraphics[width=\textwidth]{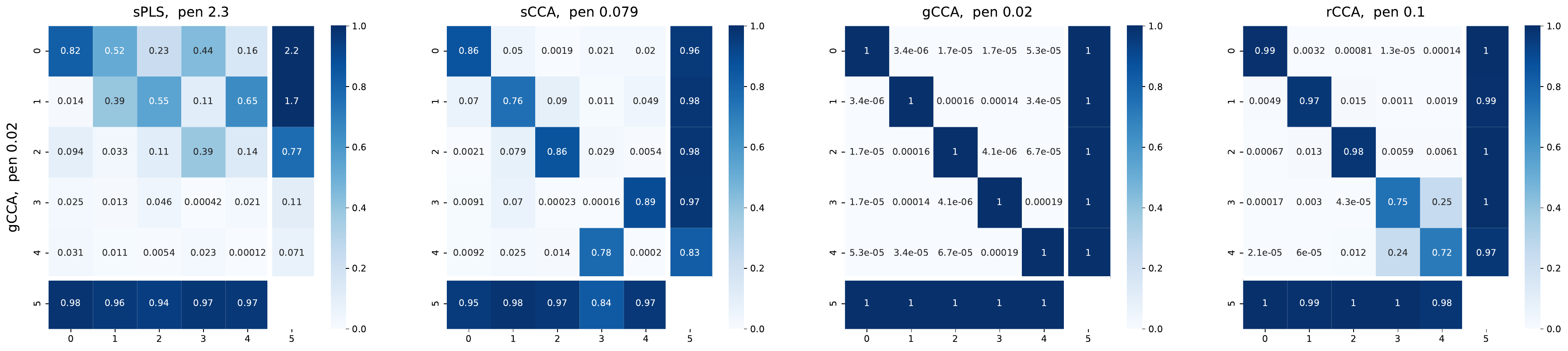}
    \caption{Squared overlap matrices for Nutrimouse dataset comparing the \wit{}, \suo{}, \glasso{}, \ridge{} estimators with maximal CV correlation (\texttt{r2s5-cv}); analogue of \Cref{fig:microbiome-sqoverlap-algos}.}.
    \label{fig:nutrimouse-sqoverlap-algos}
    \end{figure}

\FloatBarrier
\subsection{BreastData}\label{sec:breastdata-extra-plots}

We can now justify the following observations for the BreastData dataset alluded to in the main text.
\begin{itemize}
    \item \textbf{large correlations:} see \Cref{fig:correlation-stability-trajectories-BreastData,fig:corr-decay-BreastData}. Again there are many successive pairs with significant CV signal.
    Again, the genuine CCA methods capture similar signal.
    At first glance sPLS appears to capture similar signal for very small penalty parameters; on closer inspection this is due to the non-orthogonality, and is significantly different when a subspace notion is used -- see \Cref{fig:subsp-correlation-trajectories-BreastData}.
    \item \textbf{rCCA is almost invariant to penalty parameter:} as suggested by \Cref{fig:correlation-stability-trajectories-BreastData} and made clear in \Cref{fig:BreastData-traj-comps}; it is also fairly stable to sample splitting --- see \Cref{fig:correlation-stability-trajectories-BreastData}.
    \item \textbf{sCCA and sPLS very different to rCCA:} firstly there are some similarities between the sCCA and sPLS solutions, particularly in variate space \Cref{fig:BreastData-traj-comps,fig:BreastData-sqoverlap-algos}. However, these are very different to the (more stable) rCCA estimates.
    This manifests very clearly in the corresponding biplots \Cref{fig:bd-3d-plots}.
\end{itemize}

    Figure \ref{fig:bd-3d-plots} shows 3D biplots for rCCA and sCCA, which illustrate two main points.
    Firstly, they show how rCCA and sCCA can capture very different sorts of structure: the sCCA estimates are fairly `local' with only very few variables highly correlated with the canonical variates, while the rCCA solutions are more `global' with a much larger number of variables with comparable large correlations with the canonical variates.
    Secondly, they show that in this case the biplot structure reflects biologically meaningful information: the biplots successfully recover the chromosomal information (colors).
    
    \begin{figure}[h]\centering
    \begin{minipage}{\textwidth}
    \subfloat{
            \includegraphics[width=0.48\textwidth]{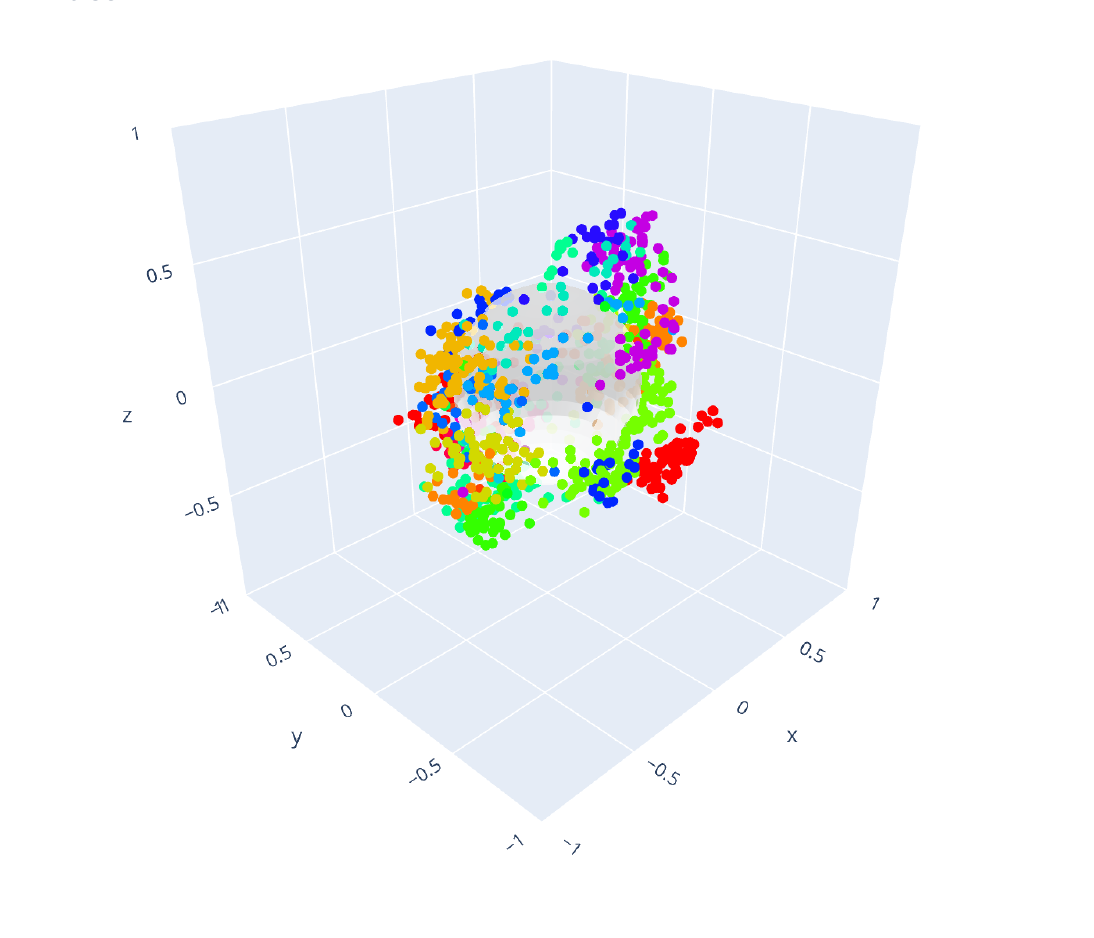}
    }
    \subfloat{
            \includegraphics[width=0.48\textwidth]{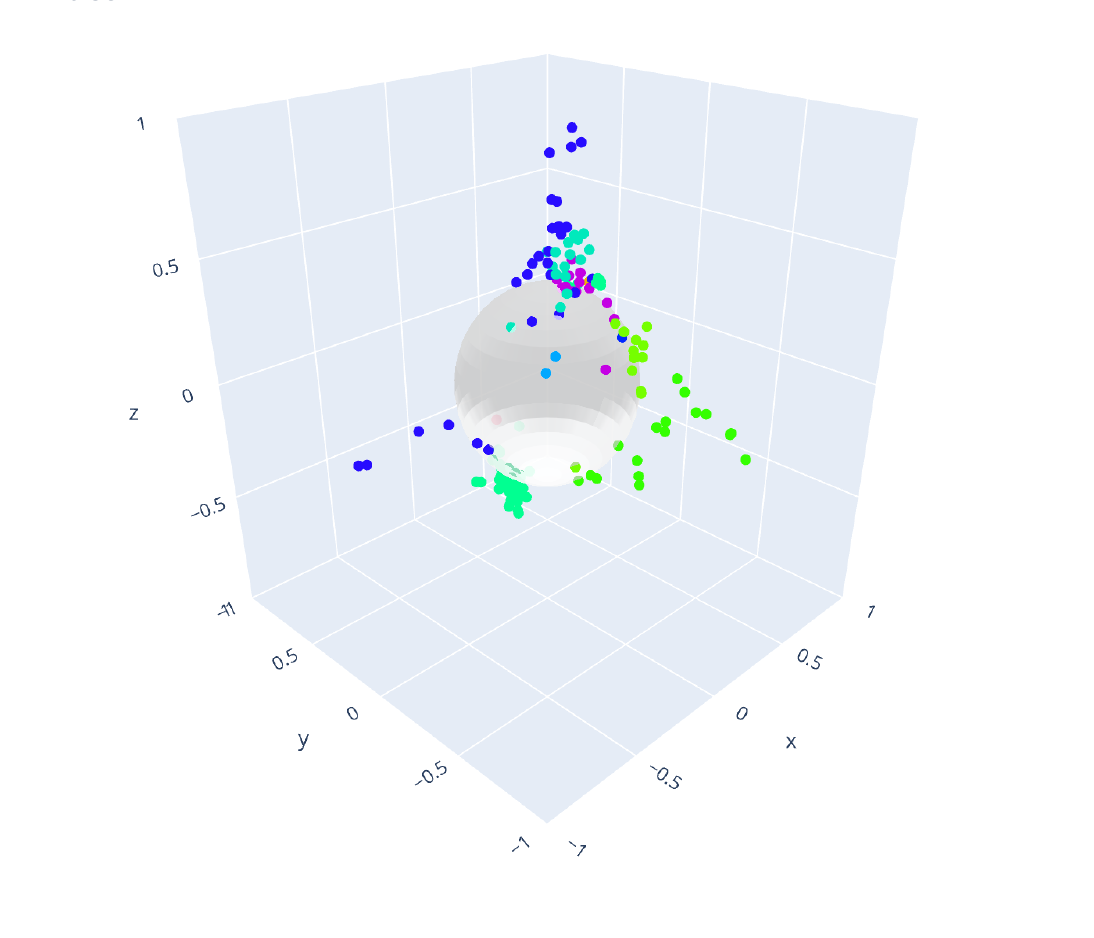}
    }
    \caption{3D `biplots' for BreastData for \ridge{} (left) and \suo{} (right) with penalty parameters 0.9, 0.072 respectively. Only the DNA copy number ($X$-view) variables are plotted, coloured by their chromosome.}
    \label{fig:bd-3d-plots}
    \end{minipage}
    \end{figure}

    \begin{figure}[t]\centering
        \begin{minipage}{0.75\textwidth}
            \subfloat{
            \includegraphics[width=\textwidth]{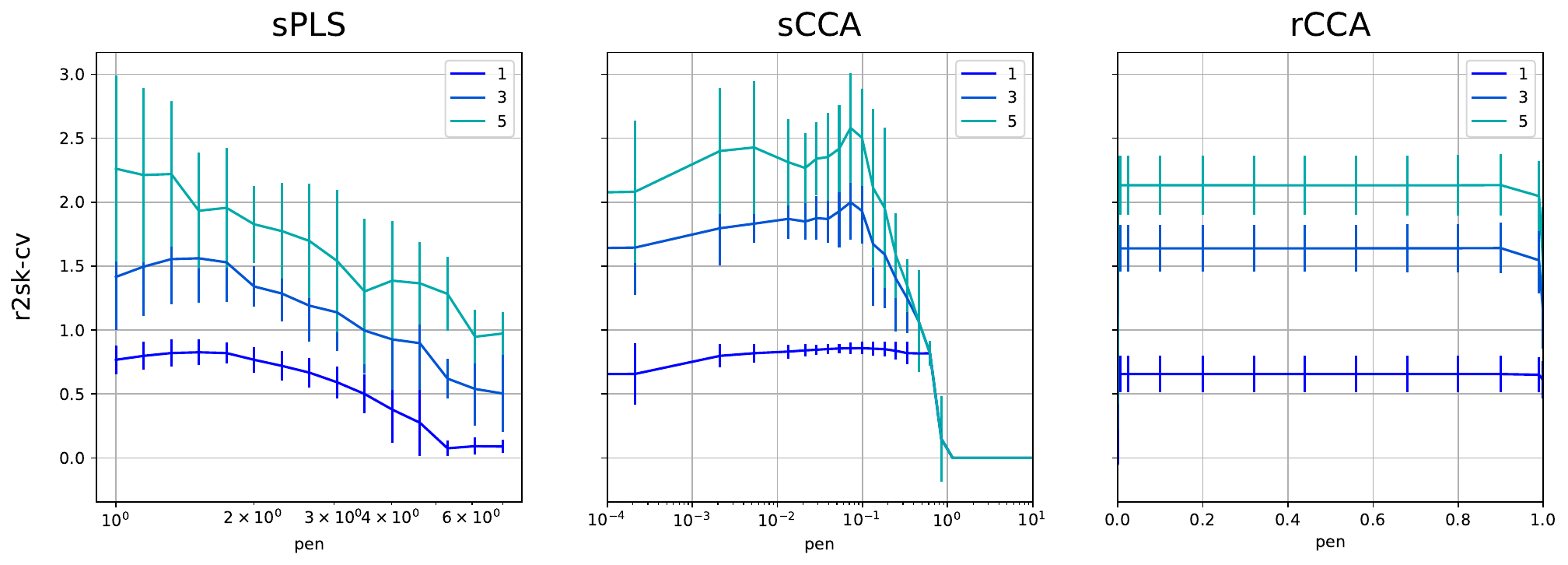}
            } \\
            \subfloat{
            \includegraphics[width=\textwidth]{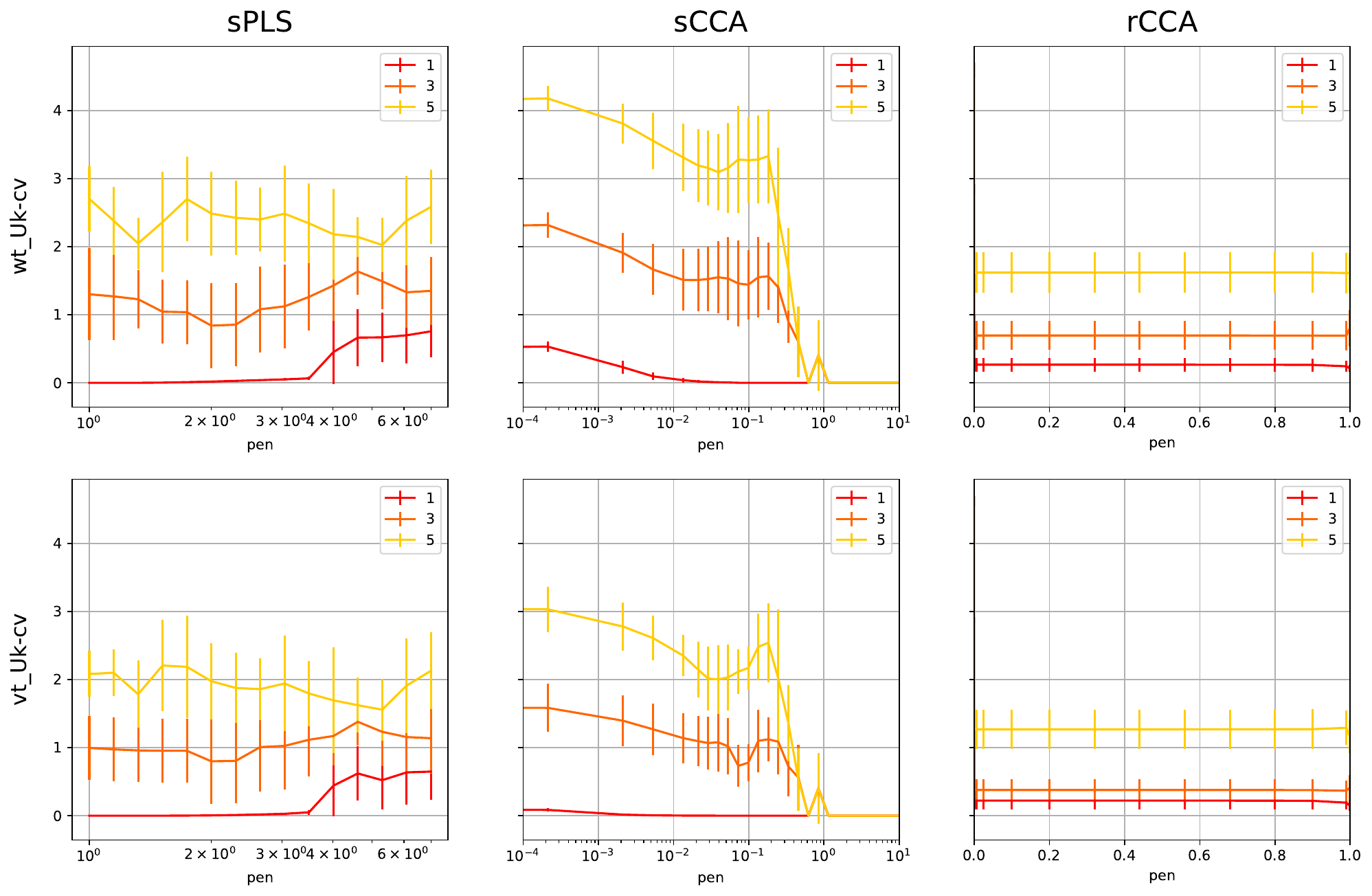}
            }
            \caption{Top: CV sums of squared correlations as function of regularisation path for  \wit{}, \suo{}, \ridge{} on the BreastData dataset; error bars for the aggregated quantities.
            Bottom: stability in variate space along those same trajectories.}
            \label{fig:correlation-stability-trajectories-BreastData}
        \end{minipage}
    \end{figure}

    \begin{figure}[t]\centering
         \includegraphics[width=0.9\textwidth]{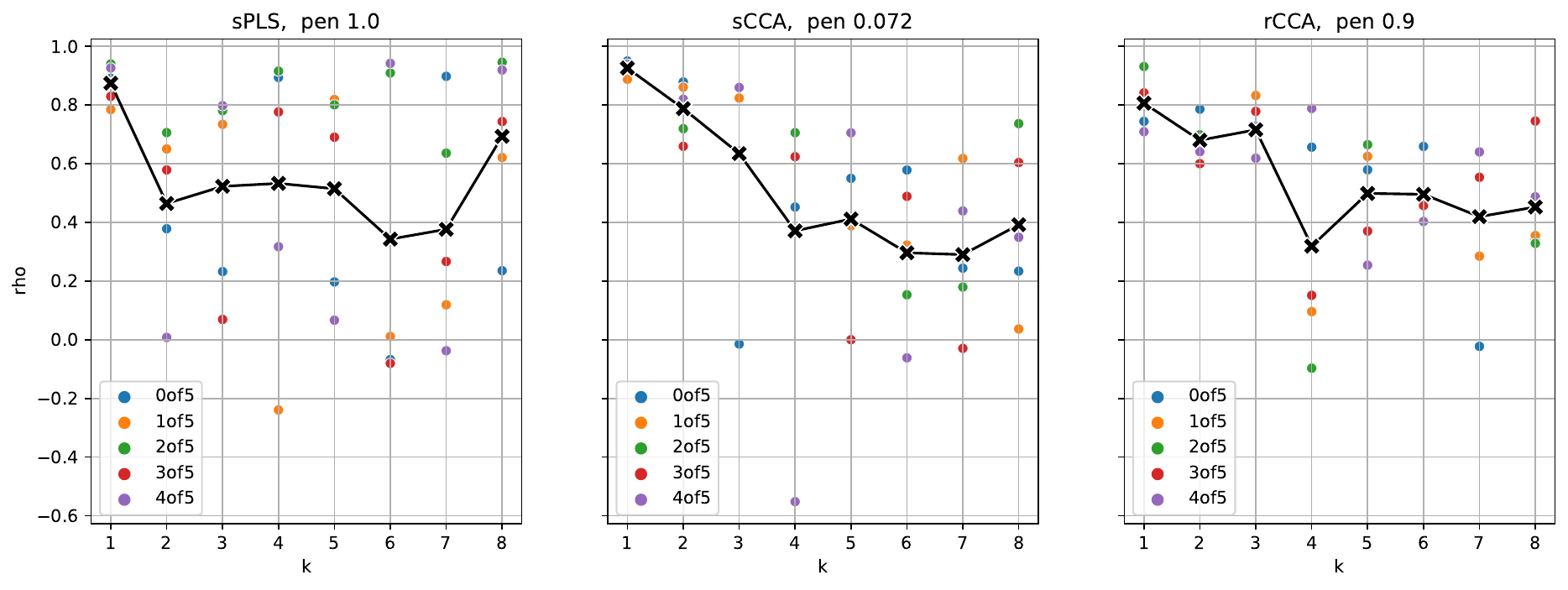}
         \caption{CV correlations (colours) and average values (black) for successive direction estimates using \wit{}, \suo{}, \ridge{} on microbiome dataset; in each case \texttt{r2s5-cv} optimal penalty parameters were used.}
         \label{fig:corr-decay-BreastData}
    \end{figure}

    \begin{figure}[t]\centering
    \begin{minipage}{0.75\textwidth}
        \subfloat{
                \includegraphics[width=\textwidth]{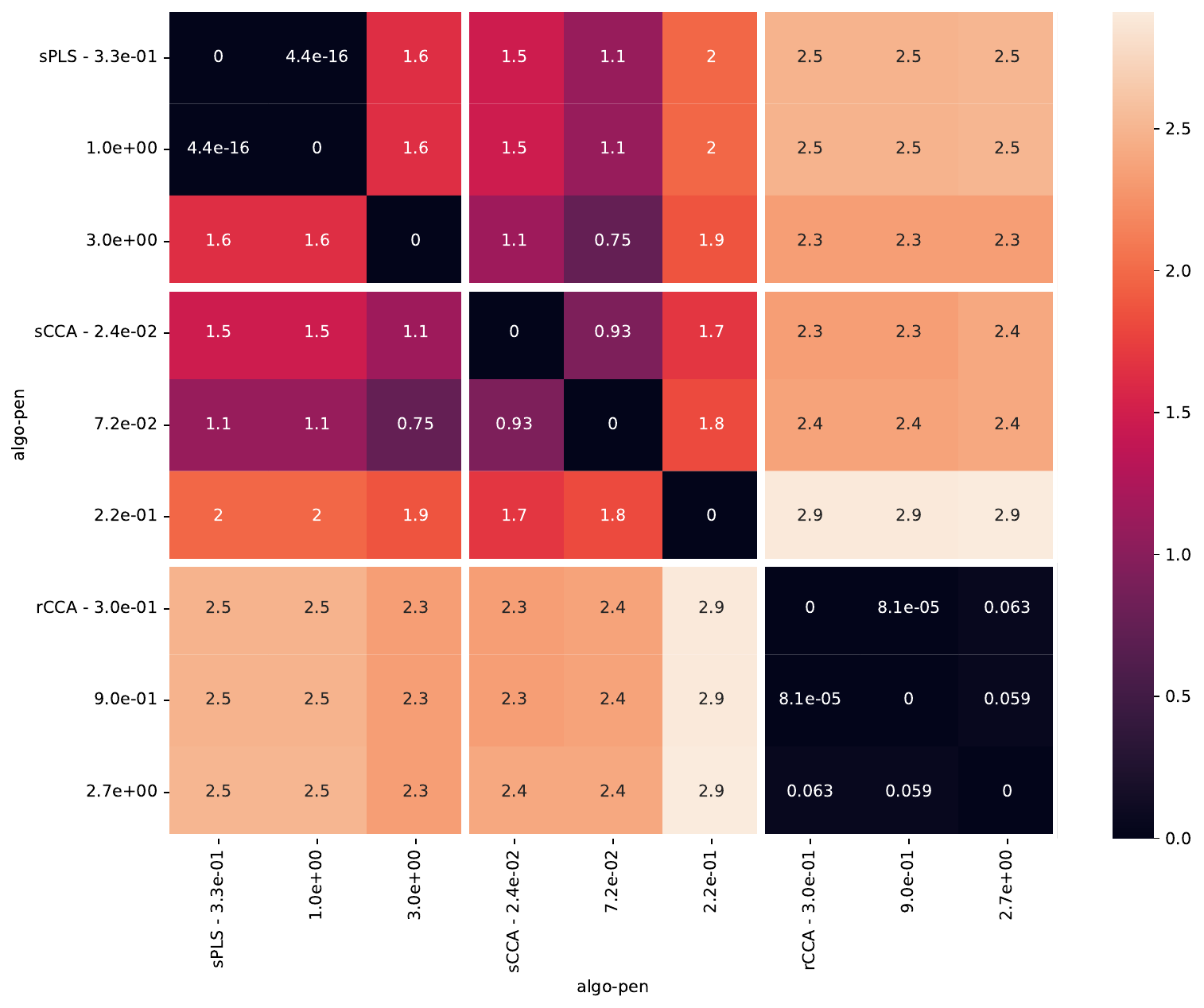}
        } \\
        \subfloat{
                \includegraphics[width=\textwidth]{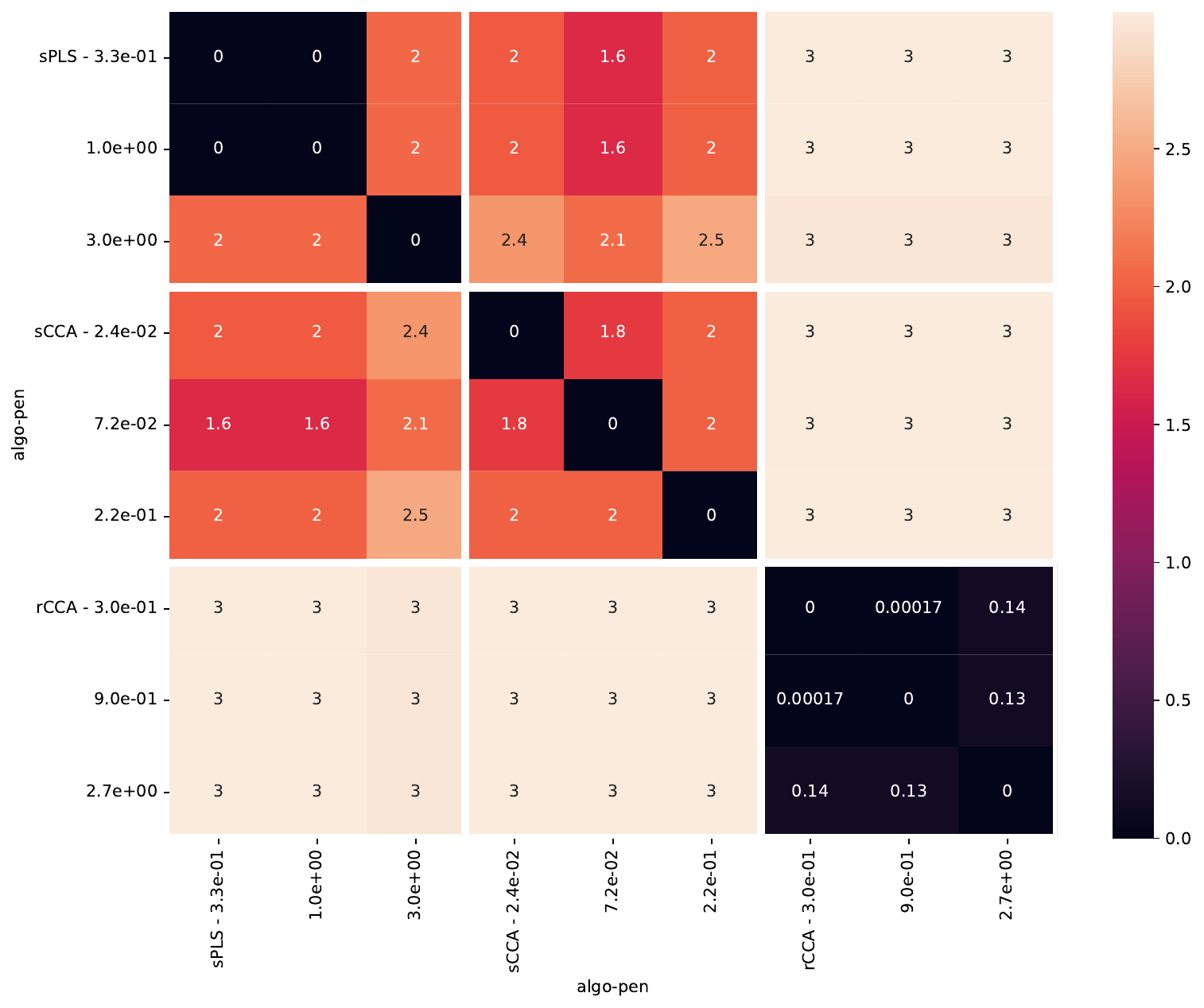}
        } 
        \caption{Top: Breastdata trajectory comparison top-3 variate subspace distances. \\
        Bottom: Breastdata trajectory comparison top-3 weight subspace distances.}
        \label{fig:BreastData-traj-comps}
    \end{minipage}
    \end{figure}

    \begin{figure}[t]\centering
        \includegraphics[width=\textwidth]{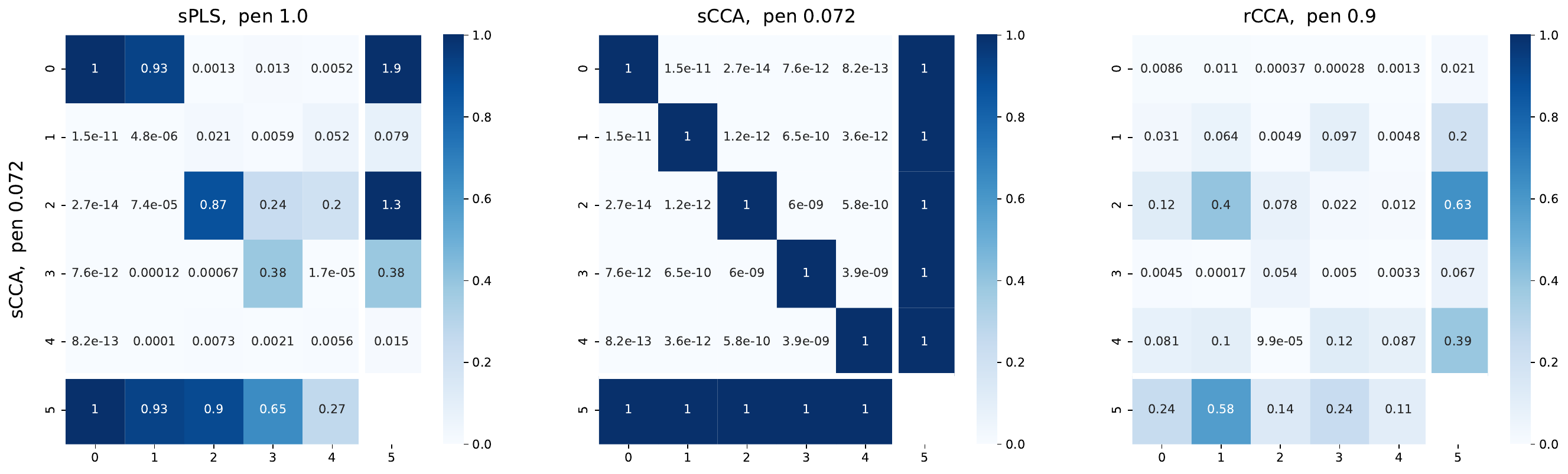}
    \caption{Squared overlap matrices for BreastData dataset comparing the \wit{}, \suo{}, \ridge{} estimators with maximal CV correlation (\texttt{r2s5-cv}); analogue of \Cref{fig:microbiome-sqoverlap-algos}.}.
    \label{fig:BreastData-sqoverlap-algos}
    \end{figure}    

    \begin{figure}[t]\centering
        \begin{minipage}{\textwidth} 
            \subfloat{
            \includegraphics[width=\textwidth]{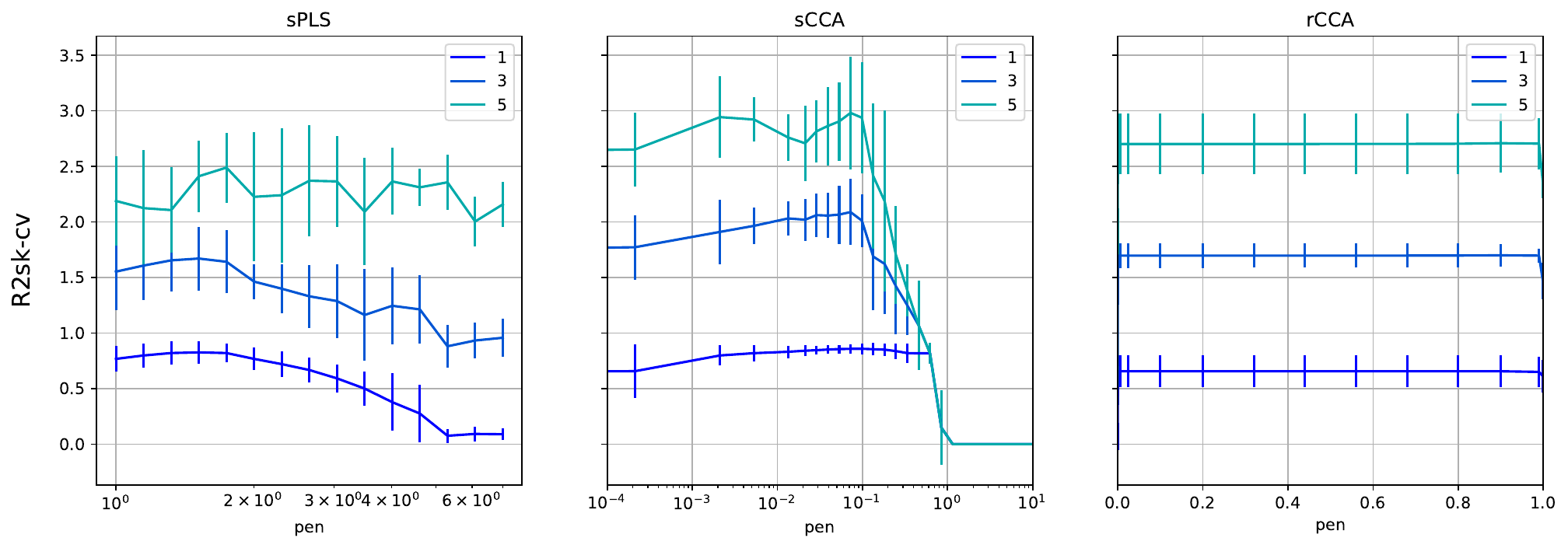}
            }
            \caption{CV sums of squared \textit{subspace} correlations as a function of regularisation path for  \wit{}, \suo{}, \ridge{} on the BreastData dataset; error bars for standard deviation.}
            \label{fig:subsp-correlation-trajectories-BreastData}
        \end{minipage}
    \end{figure}

\newpage
\flushbottom
\section{Implementation of sCCA}\label{app:sCCA-implementation}  
Though the arXiv preprint \citet{suo_sparse_2017} is widely cited, it does not in fact present a complete algorithm to solve \Cref{eq:obj-suo}: they only present a complete algorithm for the first pair of canonical direction estimates, and suggest how to extend it to the general case of later canonical pairs.
In addition, we found that when implemented naively, the algorithm was very slow to run, but could be sped up by an order of magnitude by interleaving the update steps for the estimates from each view.
We present this improved algorithm in \Cref{app:suo-impl-interleaving-algo}.
First though, we present the background for the CCA objective from \citet{suo_sparse_2017} in \Cref{app:suo-impl-obj} and for the necessary technique from convex optimisation \citep{parikh_proximal_2014} in \Cref{app:suo-impl-lADMM}.

\subsection{Objective for later canonical directions from \texorpdfstring{\citet{suo_sparse_2017}}{Suo et al. [2017]}}\label{app:suo-impl-obj}
For the rest of this section we follow \citet{suo_sparse_2017} and write $\X, \Y$ for data matrices that have been downscaled by a factor of $\sqrt{N}$; then we can drop all the factors of $\frac{1}{N}$ in \Cref{eq:sample-cov} (for example writing $\Cxx = \X^\top \X$). 
Their objective function \Cref{eq:obj-suo} for the $k^\text{th}$ canonical pair can then be written as
$$
\begin{array}{cl}
\underset{u, v}{\operatorname{minimize}} &-u^\top \X^\top \Y v+\tau_x|u|_{1}+\tau_y|v|_{1}+\ind\left\{u:\|\X u\|_{2} \leq 1\right\}+\ind\left\{v:\|\Y v\|_{2} \leq 1\right\} \\
\text { subject to } & U^\top \X^\top \X u=0 ; V^\top \Y^\top \Y v=0\;.
\end{array}
$$
where $U \in \R^{p \times (k-1)}, V \in \R^{q \times (k-1)}$ have columns consisting of the previous canonical direction estimates (we drop potential subscript $(k-1)$'s to avoid clutter). 

This problem is biconvex; the suggested alternating convex search (ACS) algorithm from \citet{suo_sparse_2017} iteratively fixes one variable and solves for the other one, alternating $u,v$. Fixing $v$, we get $\hat{u}$ by solving
\begin{equation}\label{eq:suo-app-u-update-two-constraints}
    \begin{array}{cl}
\underset{u}{\operatorname{minimize}} & -u^\top \X^\top \Y v+\tau_x|u|_{1}+\ind\left\{\z:\|\z\|_{2} \leq 1\right\} \\
\text { subject to } & \X u= \z ;\; U^\top \X^\top \X u =0_{k-1}
\end{array}
\end{equation}
and fixing $u$, we get $\hat{v}$ by solving
\begin{equation*}
    \begin{array}{cl}
\underset{v}{\operatorname{minimize}} & -u^\top \X^\top \Y v+\tau_y |v|_{1}+\ind\left\{\z:\|\z\|_{2} \leq 1\right\} \\
\text { subject to } & \Y v=\z ;\; V^\top \Y^\top \Y v =0_{k-1} \;.
\end{array}
\end{equation*}

In each case, the constraints can be combined as
$$
\left(\begin{array}{c} \X \\ U^\top \X^\top \X \end{array}\right) u - 
\left(\begin{array}{l} I_N \\ 0_{(k-1) \times N} \end{array}\right) \z=0,
\quad
\left(\begin{array}{c} \Y \\ V^\top \Y^\top \Y \end{array}\right) v - 
\left(\begin{array}{l} I_N \\ 0_{(k-1) \times N} \end{array}\right) \z=0 \;.
$$

This motivates the following definitions
\begin{align}\label{eq:suo-app-tilde-notation}
    \tilde{\X}=\left(\begin{array}{c}\X \\ U^\top \X^\top \X\end{array}\right), \quad \tilde{\Y}=\left(\begin{array}{c}\Y \\ V^\top \Y^\top \Y\end{array}\right), \quad
    \tilde{I} = \left(\begin{array}{l} I_N \\ 0_{(k-1) \times N} \end{array}\right)
\end{align}

We now highlight a mistake in \citet{suo_sparse_2017}.
They claim 
\begin{align}\label{eq:suo-mistake-dimensions}
    -u^\top \X^\top \Y v=-u^\top \tilde{\X}^\top \Y v
\end{align}
and that one can therefore use a linearised Alternating Direction Method of Multipliers (ADMM) routine analogous to that they propose for the first canonical pair, but replacing $\X, \Y$ by the extended matrices $\tilde{\X}$ and $\tilde{\Y}$.
However, \Cref{eq:suo-mistake-dimensions} is certainly false because the dimensions for the matrix multiplication $\tilde{\X}^\top \Y$ are inconsistent for $k>1$: indeed, $\tilde{\X}$ has $N + (k-1)$ rows while $\Y$ only has $N$ rows.

Instead, we follow the spirit of \citet{suo_sparse_2017} and rewrite the $u$-update from \Cref{eq:suo-app-u-update-two-constraints} in this notation as
\begin{equation}\label{eq:suo-app-u-update-general}
    \underset{v}{\operatorname{minimize}} -u^\top \X^\top \Y v+\tau_x |u|_{1}+\ind\left\{\z:\|\z\|_{2} \leq 1 \right\}
    \quad \text{ subject to } \quad  
    \tilde{\X} u - \tilde{I} \z = 0
    \;.
\end{equation}
We see that we need a version of linearised ADMM that can deal with the affine constraint $\tilde{\X} u - \tilde{I} \z = 0$.
Unfortunately, the standard reference \citep{parikh_proximal_2014} for linearised ADMM does not handle constraints of this form out-of-the-box; we consider how to extend linearised ADMM to this setting in the following subsection.

\subsection{Linearised ADMM with affine constraints}\label{app:suo-impl-lADMM}

Throughout this subsection we work with the notation from the monographs \citep{boyd_convex_2004,parikh_proximal_2014}. 
Section 4.4.2 of \citet{parikh_proximal_2014} presents linearised ADMM with constraints of the form $Ax = z$.
A more general form of ADMM given in Section 3 of \citet{boyd_distributed_2010} handles affine constraints of the form $Ax + Bz = c$.
We now work through the derivation of linearised ADMM algorithm in this more general case of affine constraint.
In the next subsection (\Cref{app:suo-impl-updates}), we then use this derivation to obtain an algorithm for the $u$-update objective \Cref{eq:suo-app-u-update-general}.

Consider the problem
$$
\begin{array}{ll}
\operatorname{minimize} & f(x)+g(z) \\
\text { subject to } & A x + B z = c
\end{array}
$$
with variables $x$ and $z$. The augmented Lagrangian for this problem can be written as
\begin{align*}
    L_{\rho}^{(y)}(x, z, y)=f(x) + g(z) + y^\top(A x + B z - c) + \frac{\rho}{2} \|A x + B z - c\|_{2}^{2} \; ,
\end{align*}
where $y \in \mathbb{R}^{m}$ is a dual variable and $\rho=1 / \lambda$.
Since the scale of $y$ is immaterial, we work with $\xi \defeq y / \rho$ for the rest of this section\footnote{Note that this substitution is not made explicit in our main reference of \citet{parikh_proximal_2014}.}; this transforms the objective to
\begin{align}\label{eq:augmented-lagrangian}
    L_{\rho}(x, z, \xi)=f(x) + g(z) + \rho \left( \xi^\top(A x + B z - c) + \frac{1}{2}\|A x + B z - c\|_{2}^{2}\right) \; ,
\end{align}
Usual ADMM would give updates
\begin{align}
    x^{k+1} &:=\underset{x}{\operatorname{argmin}}\; L_{\rho}\left(x, z^{k}, y^{k}\right) \label{eq:admm-x-update}\\
    z^{k+1} &:=\underset{z}{\operatorname{argmin}}\; L_{\rho}\left(x^{k+1}, z, y^{k}\right) \nonumber\\
    \xi^{k+1} &:= \xi^{k} + \left(A x^{k+1}+B z^{k+1}-c\right) \nonumber\;.
\end{align}

In linearised ADMM, we modify the usual $x$-update by replacing the $\frac{\rho}{2} \left\|A x + B z^{k} - c\right\|_{2}^{2}$ term in $L_{\rho}$ by
$$
\rho\left(A^\top A x^{k} + A^\top B z^{k}\right)^\top x + \frac{1}{2 \mu}\left\|x-x^{k}\right\|_{2}^{2}\;.
$$
where $\mu$ is an additional tuning parameter that satisfies $0< \mu \leq {\lambda}/{\norm{A}_\textrm{op}^2}$.
Plugging this into \Cref{eq:admm-x-update} and using standard properties of proximal operators \citep{parikh_proximal_2014} gives the alternative $x$-update defined by
\begin{align*}
    x^{k+1} 
    &= \underset{x}{\operatorname{argmin}}\; f(x) + \rho (A x)^\top \xi^k + \rho\left(A^\top A x^{k} + A^\top B z^{k}\right)^\top x +\frac{1}{2 \mu} \left\|x-x^{k}\right\|_{2}^{2} \\
    &= \underset{x}{\operatorname{argmin}}\; \mu f(x) + \frac{\mu}{\lambda} \, \langle x, A^\top ( A x^k + B z^k + \xi^k) \rangle + \frac{1}{2} \left\|x-x^{k}\right\|_{2}^{2} \vphantom{\biggl(\biggr)}\\
    &= \operatorname{prox}_{\mu f}\left(x^{k} - \frac{\mu}{\lambda} A^\top ( A x^k + B z^k + \xi^k) \right)
\end{align*}

Under the simpler setting where $Ax = z$, the $z$-update can also be simply expressed as a proximal operator, because the augmented Lagrangian \Cref{eq:augmented-lagrangian} because the quadratic part of $z$ has isotropic Hessian.
This does not hold for general $B$, but it does hold for our special case $B = \tilde{I}$ from \Cref{eq:suo-app-tilde-notation}.
We write out the update steps explicitly in the next subsection for this special case that we need to implement \suo{}.

\subsection{Applying linearised ADMM to obtain updates for \texorpdfstring{\Cref{eq:suo-app-u-update-general}}{later directions}}\label{app:suo-impl-updates}
We can now translate the method from \Cref{app:suo-impl-lADMM} to our setting of \Cref{app:suo-impl-obj} in order to obtain update steps to solve the optimisation problem \Cref{eq:suo-app-u-update-general}.
We identify
$$
(x,z,A,B,c)\Longleftrightarrow 
(u,\z,\tilde{\X}, - \tilde{I},0)
$$ 
and
$$
f: u \mapsto -u^\top \X^\top \Y v + \tau_x \|u\|_2, \qquad
g: \z \mapsto \ind\left\{\z:\|\z\|_{2} \leq 1\right\} \;.
$$

To show that the $\z$-update can be written in terms of a proximal operator, we exploit the block structure of $\tilde{\X}$ and $\tilde{I}$ to obtain
$$\begin{aligned}
\z^{k+1}  
&= \underset{\z}{\operatorname{argmin}}\; \ind\left\{\z:\norm{\z}_{2} \leq 1\right\} + \frac{\rho}{2}\; \norm{\tilde{I}\z - \tilde{\X}u^{k+1} - \xi_k}_2^2 \\
&= \underset{\z}{\operatorname{argmin}}\; \ind\left\{\z:\norm{\z}_{2} \leq 1\right\} + \frac{\rho}{2}\; \norm{\z-\X u^{k+1} - \tilde{I}^\top \xi^k}_2^2 \;.
\end{aligned}
$$
Note that $\tilde{I}^\top \xi^k$ consists precisely of the first $N$ elements of $\xi^k$, so is very easy to compute.

Finally we obtain the updates
    \begin{equation}\label{eq:suo-app-ladmm-u-updates}
        \begin{aligned}
            &u^{k+1} \leftarrow \operatorname{prox}_{\mu f}\left(u^{k}-\frac{\mu}{\lambda} \tilde{\X}^\top\left(\tilde{\X} u^{k}-  \tilde{I} \z^{k}+\xi^{k}\right)\right) \\
            &\z^{k+1} \leftarrow \operatorname{prox}_{\lambda g}\left(\X u^{k+1}+\tilde{I}^\top \xi^{k}\right) \\
            &\xi^{k+1} \leftarrow \xi^{k}+\X u^{k+1}-\z^{k+1} \; ,
        \end{aligned}
    \end{equation}
where the proximal operators are defined as follows, following \citet{suo_sparse_2017}.
Write $c=\X^\top \Y v$ (the gradient of the objective function with respect to $u$ for a fixed $v$).
Then $f(x)$ involves (element-wise) soft-thresholding and $g(x)$ is a projection to the unit ball defined by
\begin{align*}
    \operatorname{prox}_{\mu f}(x) 
    &= \begin{cases}x+\mu c-\mu \tau_x & \text { if } x+\mu c>\mu \tau_x \\ x+\mu c+\mu \tau_x & \text { if } x+\mu c<-\mu \tau_x \\ 0 & \text { else,}\end{cases}
\\[12pt]
    \operatorname{prox}_{\lambda g}(\z) 
    &= \begin{cases}\z & \text { if }\|\z\|_{2} \leq 1 \\ \frac{\z}{\|\z\|_{2}} & \text { else,}\end{cases}
\end{align*}
Note that the soft-thresholding depends on $c$ as well as $\tau_x$, which is how the dependence on $v$ manifests in the update steps.

The analogous step for the $v$-optimisation for a given $u$ can be obtained by symmetry/relabelling, as we exploit in the following subsection.

\subsection{Complete algorithm with computational improvements}\label{app:suo-impl-interleaving-algo}

We propose \Cref{algo:sCCA-next-pair} to solve \Cref{eq:obj-suo}.
This applies the updates from \Cref{app:suo-impl-updates} in a framework very similar to the original Alternating Convex Search (ACS) procedure proposed by \citet{suo_sparse_2017}, but with two small but significant algorithmic changes:
\begin{itemize}
    \item \textbf{Early stopping of the linearised ADMM sub-routines:} for a fixed $v$, update $u$ by running the linearised ADMM updates \Cref{eq:suo-app-ladmm-u-updates} for some small number of steps $n_{\text{steps\_ADMM}}$ rather than to convergence.
    \item \textbf{Recycling optimal dual variables:} after each of these sets of linearised ADMM steps, record the values of the optimal dual variables, and use these for initialising the subsequent set of linearised ADMM steps.
\end{itemize}
We found that these two changes reduced the number of total number of linearised ADMM update steps required for convergence of the weight estimates $u,v$ by an order of magnitude (on the real-world datasets of interest to us).

We now outline why these changes do not weaken the theoretical properties of the algorithm.
Suppose $u,v$ are returned by the algorithm.
This requires convergence in the outer loop, which implies convergence in the inner loop.
Since the inner loop consists of the linearised ADMM updates, this only occurs when $u$ is the global minimum of \Cref{eq:suo-app-u-update-general} for the given $v$ and vice versa.
Therefore $u,v$ must be some local minimum of the biconvex objective.
This is essentially all one could say about solutions found by the original Alternating Convex Search (ACS) algorithm, so our algorithm has similar theoretical justification.

A more careful convergence analysis may be possible by adapting the Lyapunov arguments from \citet{zhang_unified_2011}[Theorem 4.2]. However, the details become involved, so we leave this for an interested reader.

We are deliberately vague about initialisation in \Cref{algo:sCCA-next-pair}.
For our numerical experiments, we use an implementation based off ideas from \citet{suo_sparse_2017}, that starts by taking a SVD of a soft-thresholded version of the between view covariance matrix $\X^\top \Y$.
However, we found that various schemes led to similar performance, and so we leave full details to our code implementation at \githubrepo.


\begin{algorithm}
\caption{Sparse Canonical Correlation Analysis (sCCA) Next Pair Calculation}\label{algo:sCCA-next-pair}
\begin{algorithmic}[1]
    \Procedure{sCCA-next-pair}{$\X,\Y, U, V, \tau_x, \tau_y; \lambda, n_{\text{steps\_ADMM}}$}
        \State \textbf{recall:} \parbox[t]{\dimexpr\linewidth-\algorithmicindent}{
        $\tau_x, \tau_y \in [0, \infty)$ are $\littlel{1}$ penalty parameters, \\
        $\lambda \in [0, \infty)$ is an ADMM step size parameter, \\
        $n_{\text{steps\_ADMM}}$ is a small positive integer.
        }\vspace{4pt}
        \State \textbf{ensure:} \parbox[t]{\dimexpr\linewidth-\algorithmicindent}{Dimensional consistency: for some positive integers $n,p,q,K$ we have \\
        $\X \in \R^{N \times p}, \Y \in \R^{N \times q}, U \in \R^{p \times (k-1)}, V \in \R^{q \times (k-1)}$.}\vspace{7pt}
        \State Construct $\tilde{\X}, \tilde{\Y}, \tilde{I}$ as in \Cref{eq:suo-app-tilde-notation}
        \vspace{5pt}
        \State Initialise:
        \begin{itemize}
            \item $u, v$ in some sensible manner 
            \item $\z_u \gets \Call{Normalise}{\X u}, \quad \xi_u \gets \tilde{\X} u - \tilde{I} \z_u$
            \item $\z_v \gets \Call{Normalise}{\Y v}, \quad \xi_v \gets \tilde{\Y} v - \tilde{I} \z_v$
        \end{itemize}
        \vspace{5pt}
        \While{not converged} 
            \State $u, \z_u, \xi_u \gets \Call{lADMM-steps}{u, \z_u, \xi_u, \tilde{\X}, \X, \Y, v, \tau_x;\, \lambda, n_{\text{steps\_ADMM}}}$
            \State $v, \z_v, \xi_v \gets \Call{lADMM-steps}{v, \z_v, \xi_v, \tilde{\Y}, \Y, \X, u, \tau_y;\, \lambda, n_{\text{steps\_ADMM}}}$
        \EndWhile
        \vspace{5pt}
        \State \textbf{return} $u, v$
    \EndProcedure
    \vspace{15pt}
    
    \Procedure{lADMM-steps}{$u, \z, \xi, \tilde{\X}, \X, \Y, v, \tau;\, \lambda, n_{\text{steps}}$}
    \vspace{4pt}\State $\mu \gets \lambda / (2 \cdot \norm{\tilde{\X}}_\textrm{op}^2)$
    \vspace{5pt}
    \State $c \gets \X^\top \Y v$
    \vspace{2pt}
    \State Define $\operatorname{prox}_{\mu f}: \R^p \to \R^p$ by $x \mapsto 
    \begin{cases}x+\mu c-\mu \tau & \text { if } x+\mu c>\mu \tau \\ x+\mu c+\mu \tau & \text { if } x+\mu c<-\mu \tau \\ 0 & \text { else}\end{cases}$
    \vspace{8pt}
    \For{$\text{step} = 1$ to $n_{\text{steps}}$}
        \vspace{4pt}
        \State $u \leftarrow \operatorname{prox}_{\mu f}\left(u-\frac{\mu}{\lambda} \tilde{\X}^\top\left(\tilde{\X} u-  \tilde{I} \z+\xi\right)\right)$
        \vspace{5pt}
        \State $\z \leftarrow \Call{Normalise}{\X u+ \tilde{I}^\top \xi}$
        \vspace{5pt}
        \State $\xi \leftarrow \xi + \X u-\z$.
        \vspace{4pt}
    \EndFor
    \vspace{5pt}
    \State \textbf{return} $u, \z, \xi$
    \EndProcedure
    \vspace{15pt}

    \Procedure{Normalise}{$\z$}
    \If{$\norm{\z} > 1$}
    \State $\z \gets \z / \norm{\z}$
    \EndIf
    \State \textbf{return} $\z$
    \EndProcedure

\end{algorithmic}
\end{algorithm}